\documentclass[11pt]{article}

\usepackage{xcolor}
\definecolor{byzantine}{rgb}{0.74, 0.2, 0.64}
\definecolor{brilliantrose}{rgb}{1.0, 0.33, 0.64}
\usepackage[letterpaper]{geometry}
\usepackage{comment}

\usepackage{relsize}
\usepackage{url,amsfonts, amsmath, amssymb, amsthm,color, enumerate, multicol, cancel}

\renewcommand{\vec}[1]{\mathbf{#1}}

\newcommand{\E}{\mathbb{E}}

\newtheorem{lemma}{Lemma}
\newtheorem{proposition}{Proposition}
\newtheorem{prop}{Proposition}

\newtheorem{definition}{Definition}
\newtheorem{example}{Example}
\newtheorem{theorem}{Theorem}
\newtheorem{observation}{Observation}

\newtheorem{corollary}{Corollary}

\usepackage{amstext,amsopn} 

\usepackage{graphicx}
\usepackage[colorlinks=true, allcolors=blue]{hyperref}
\usepackage{fullpage}
\usepackage{fancyvrb}
\usepackage{stmaryrd}
\usepackage{bm}
\usepackage{color}
\usepackage{fancyhdr}
\usepackage{subfig}
\usepackage{cite}
\usepackage{verbatim}
\usepackage{dsfont}

\newcommand{\explain}[1]{\tag{\textcolor{gray}{#1}}}

\date{}

\begin{document}
	
	\title{Searching, Sorting, and Cake Cutting in Rounds\footnote{A part of this work was done while visiting the Simons Institute for the Theory of Computing. S. Br\^anzei was supported in part by US National Science Foundation
			CAREER grant CCF-2238372.}}
	
	\author{Simina Br\^anzei\footnote{Purdue University, USA. E-mail: \href{mailto:simina.branzei@gmail.com}{simina.branzei@gmail.com}.}
		\and
		Dimitris Paparas\footnote{Google Research, USA. E-mail: 
			\href{mailto:dpaparas@google.com}{dpaparas@google.com}.}
		\and
		Nicholas  Recker\footnote{Purdue University, USA. E-mail: \href{mailto:nrecker@purdue.edu}{nrecker@purdue.edu}.}
	}

	\maketitle
	\thispagestyle{empty}

	\begin{abstract}
		We analyze {the query complexity of} search in rounds motivated by a fair division question. This leads us to the discovery of a substantial gap between the \emph{expected randomized} and  \emph{distributional} query complexity for a natural  search problem with error probability  $\delta$.
		To the best of our knowledge, this is the first natural example of a function with a large gap between its randomized and distributional query complexity.
		
		While the expected query complexity was the focus of Yao's seminal paper \cite{yao_minimax}, it has   been understudied since then. Most literature has focused on the worst-case cost setting, where the randomized complexity is defined as the worst case cost incurred by the best algorithm with error probability $\delta$ and for this setting there is no gap between the two complexities~\cite{CC_book}.  In recent years, there has been renewed interest in the expected cost setting  \cite{small_minimax_bias_BBGM22}.

		The catalyst for our findings is  a fair division problem  of  independent interest: given a cake cutting instance with $n$ players, compute a fair allocation in at most $k$ rounds of interaction with the players, where $k$ ranges from $1$ to $\infty$.
		We  {show} that proportional cake cutting in rounds is equivalent to sorting with rank queries in rounds.
		Inspired by the rank query model, we then consider two fundamental search problems: ordered and unordered search.
		
		In unordered search, we get an array $\vec{x} = (x_1, \ldots, x_n)$ and an element $z$ promised to be in $\vec{x}$. The size $n$  is known, but $z$ and the elements of $\vec{x}$ are not and cannot be accessed directly. Instead, we have access to an oracle that receives queries of the form: ``How is $z$ compared to the element at location $i$?'', answering ``$=$'' or ``$\neq$''.  The goal is to find the location of $z$ with success probability at least $p \in (0,1]$ using at most  $k$ rounds of interaction with the oracle.

		We find and quantify a gap between the expected query complexity of  
		\begin{itemize}
			\item randomized algorithms on a worst case input, which is $np\bigl(\frac{k+1}{2k}\bigr) \pm O(1)$;  and 
			\item deterministic algorithms on a worst case input distribution, which  is $np \bigl(1 - \frac{k-1}{2k}p \bigr) \pm O(1)$.
		\end{itemize}		
		
		The gap grows with the number of rounds and is maximized in the fully adaptive unordered search problem, where for each $p \in (0,1)$,   the ratio between the two complexities converges to $2-p$ as the size of the input  grows, while their additive difference grows linearly with the input size.  In particular, the ratio is strictly greater than $1$ for each fixed success probability $p \in (0,1)$. 
		
		In ordered search, the setting is the same with the difference that the array $\vec{x}=(x_1, \ldots, x_n)$ is promised to be sorted and  the answer given by the oracle is one of ``$<$'', ``$=$'', or ``$>$''.
		Here we find that the expected query complexity of randomized algorithms on a worst case input and deterministic algorithms on a worst case input distribution is essentially the same: at least $p k  n^{\frac{1}{k}} - 2pk$ and at most $k \lceil p n^{\frac{1}{k}} \rceil + 2$.
	\end{abstract}
	
	\newpage
	\pagenumbering{arabic}
	
	\section{Introduction}
	
	{W}e explore the randomized and distributional query complexity of search problems in the expected cost setting.
	This leads us to the discovery of a substantial gap between the randomized and distributional {query} complexity of a natural function induced by a search problem. 
	
	To make these concepts more precise, consider a function  $f: \mathcal{X}_n \to \mathcal{Y}_m$, where $\mathcal{X}_n \subseteq \{0,1\}^n$ and $\mathcal{Y}_m \subseteq \{0,1\}^m$ with $m,n \in \mathbb{N}$. 
	Given as input a bit vector $\vec{x} = (x_1, \ldots, x_n) \in \mathcal{X}_n$, an algorithm can query a location $j$ in $\vec{x}$ and receive the bit $x_j$ in one step. The goal is to compute  $f(\vec{x})$ with as few queries  as possible. 
	
	The randomized and distributional complexity~\cite{yao_minimax} of computing the function $f$ are defined as follows.
	The \emph{randomized complexity with error $\delta$}, denoted  $\mathcal{R}_{\delta}(f)$, is 
	the \emph{expected} number of queries issued on the worst-case input of an optimal randomized algorithm that computes $f$ with an error probability of at most  $\delta \in [0,1]$ on each input. See    Section~\ref{sec:our_results} for precise definitions. 
	
	When the input $x$ is drawn from a distribution $\Psi$,  a deterministic algorithm $\mathcal{A}$ (not necessarily correct on all inputs) has expected number of queries $cost(\mathcal{A},\Psi)$ and  error probability $e(\mathcal{A}, \Psi)$. 
	Let $\mathcal{A}_{\Psi,\delta}$ be an algorithm with error probability $\delta$ and  minimum expected cost for distribution $\Psi$. The \emph{distributional complexity with error $\delta$}, denoted $\mathcal{D}_{\delta}(f)$, represents the \emph{expected} number of queries made on the worst case  distribution $\Psi$ by the best algorithm for it: $\mathcal{D}_{\delta}(f) = \sup_{\Psi} \left\lbrace cost(\mathcal{A},\Psi)\right\rbrace$.

	For error probability $\delta = 0$,   von Neumann's minimax theorem~\cite{vonNeumann_minimax} gives $\mathcal{R}_{0}(f) = \mathcal{D}_{0}(f)$. Clearly, we also   have $\mathcal{R}_{1}(f) = \mathcal{D}_{1}(f) =0$. 
	For all   $\delta \in [0, 1/2]$, \cite{yao_minimax} showed   $\mathcal{R}_{\delta}(f) \geq 1/2 \cdot \mathcal{D}_{2 \delta}(f)$. \cite{vereshchagin1998randomized} showed an inequality in the other direction: $\mathcal{R}_{\delta}(f) \leq 2 \mathcal{D}_{\delta/2}(f)$ for all $\delta \in [0,1]$. 
	\cite{vereshchagin1998randomized} also showed that $\mathcal{R}_{\delta}(f) \leq \mathcal{D}_{\delta}(f)$ and observed that there can be a difference of  additive $1$ between $\mathcal{R}_{\delta}(f)$ and $\mathcal{D}_{\delta}(f)$ for  $\delta \in [1/4, 1/2]$  when the function is  $f : \{0,1\} \to \{0,1\}$ with $f(x) = x$. 

	
	While the expected query complexity was the focus of Yao's seminal paper \cite{yao_minimax}, it has   been understudied since then. Most literature has focused on the worst-case cost setting, where the randomized complexity is defined as the worst case cost incurred by the best algorithm with error probability $\delta$ and for this setting there is no gap between the two complexities~\cite{CC_book}. In recent years, there has been renewed interest in the expected cost setting  \cite{small_minimax_bias_BBGM22}, 
	as it has important  applications in complexity such as  the randomized composition of functions. 
	
	Our work  contributes  to this area by showing that a natural search problem has a substantial  gap between the randomized and distributional complexity in the expected cost setting.
	Specifically, we consider a natural function $u_n$ induced by the problem of finding  an element $z$ in an unordered array of size $n$. We show that for each   $\delta \in (0,1)$,
	\begin{align} 
		\lim_{n \to \infty} \frac{\mathcal{D}_{\delta}({u}_n)}{\mathcal{R}_{\delta}({u}_n)} =   1+\delta \qquad \mbox{and} \qquad \lim_{n \to \infty} \mathcal{D}_{\delta}({u}_n) - \mathcal{R}_{\delta}({u}_n) = \infty\,. \notag 
	\end{align}
	To the best of our knowledge, this is the first example demonstrating a large gap.  In fact, 
	\cite{small_minimax_bias_BBGM22} asked whether there exist constants $c,d \geq 0 $ such that $\mathcal{D}_{\delta}(f) \leq c \cdot \mathcal{R}_{\delta}(f) + d$ for each partial function $f$ and $\delta > 0$. Our results show that the two complexities can be substantially different, in particular implying that if such constants $c$ and $d$ exist, it must be the case that $c \geq 2$.
	
	\paragraph{Connections to Cake Cutting and Rounds of Interaction:}
	
	The catalyst  for our findings is a cake cutting problem that we believe is of independent interest. Suppose we are given a cake  represented as the interval $[0,1]$ and  $n$ players, each with an additive valuation over the cake induced by a private value density function. The task is to compute a fair allocation using at most $k$ rounds of interaction with the  players. Each round of interaction $i$ consists of a batch $i$ of queries issued simultaneously. Queries in batch $i$ can depend on the responses to queries from rounds $j<i$ but not to queries from rounds $\ell\geq i$. When $k=1$, all the communication between the algorithm computing the allocation and the players takes place in one simultaneous exchange, while  $k= \infty$ represents  the fully adaptive setting, where the algorithm issues one query at a time (see \cite{Val75}).

	We design an efficient protocol for proportional cake cutting in rounds, finding that this fair division problem is equivalent to sorting with rank queries in  rounds, where a rank query has  the form ``Is rank$(x_i) \leq j$?''.
	A lower bound for sorting with rank queries in rounds was given in \cite{AlonAzar88}, while the first connection to proportional cake cutting was implicitly made in \cite{woeginger2007complexity}. 
	
	Inspired by the rank query model, we then consider two fundamental search problems that are implicit in sorting: ordered and unordered search. 
	In unordered search, we get an array $\vec{x} = (x_1, \ldots, x_n)$ and an element $z$ promised to be in $\vec{x}$. The size $n$  is known, but $z$ and the elements of $\vec{x}$ are not and cannot be accessed directly. Instead, we have access to an oracle $\mathcal{O}_z$ that receives queries of the form: $\mathcal{O}_z(i)=$``How is $z$ compared to the element at location $i$?'', answering ``$=$'' or ``$\neq$''.  The goal is to find the location of $z$ with success probability at least $p \in (0,1]$ using at most  $k$ rounds of interaction with the oracle. 
	
	
	In ordered search, the setting is the same with the difference that (1) the array $\vec{x}=(x_1, \ldots, x_n)$ is promised to be sorted and (2) the answer given by the oracle is one of ``$<$'', ``$=$'', or ``$>$''.
	
	

	\subsection{Our results} \label{sec:our_results}
	
	Here we summarize our results after establishing the notation necessary for stating them.
	
	\paragraph{Notation.} For $m,n \in \mathbb{N}$, we consider functions of the form 
	$f : \mathcal{X}_n \to \mathcal{Y}_m $, where $\mathcal{X}_n \subseteq \{0,1\}^n$ and $\mathcal{Y}_m \subseteq \{0,1\}^m$.   For promise problems, as in our setting, the set $\mathcal{X}_n$ is  a strict subset of $\{0,1\}^n$.
	
	For each $x \in \{0,1\}^n$ and randomized algorithm $R$ for computing $f$, the error probability of $R$  on input $x \in \mathcal{X}_n $  is  
	\begin{align} 
		\text{err}_{f}(R,x) = \Pr[R(x) \neq f(x)] \qquad \forall x \in \mathcal{X}_n, \notag 
	\end{align}
	
	where  $R(x)$ is the output of the algorithm and can be the empty string. For the functions we consider, the empty string is never the right answer.
	
	For each  $\delta \in [0,1]$, we consider the randomized complexity $\mathcal{R}_{\delta}(f)$ with error 
	at most $\delta$ and  the distributional complexity $\mathcal{D}_{\delta}(f)$  with error at most $\delta$, formally defined as 
	\begin{align} 
		\mathcal{R}_{\delta}(f) = \inf_{R \in R(f, \delta)} \max_{x \in \mathcal{X}_n} cost(R,x) \qquad \mbox{and} \qquad 
		\mathcal{D}_{\delta}(f) = \sup_{\mu} \inf_{D \in D(f,\delta, \mu)} cost(D,\mu),  
	\end{align}
	where 
	\begin{itemize} 
		\item $R(f,\delta)$ is  the set of randomized algorithms $R$ such that $\text{err}_{f}(R,x) \leq \delta$ for all  $x \in \mathcal{X}_n$.
		\item $\mu$ is a distribution over strings in $\mathcal{X}_n$; that is, $\sum_{x \in \mathcal{X}_n} \mu(x) = 1$.
		\item  $D(f, \delta, \mu)$ is the set of deterministic algorithms  $D$ with $\E_{x \sim \mu}[\text{err}_{f}(D,x)] \leq \delta$.
		\item $cost(R,x)$ is the expected number of queries issued by a randomized algorithm $R$ on input $x$.
		\item  $cost(D,\mu) $ represents the expected number of queries issued by a deterministic algorithm $D$ when given as input a string $x$ drawn from the distribution $\mu$.
	\end{itemize}
	
	\subsubsection{Unordered Search} 
	
	The unordered search problem is formally defined as follows.

	\begin{definition}[Unordered search]
		The input is a  bit vector $\vec{x} = (x_1, \ldots, x_n) \in \{0,1\}^n$ with the promise that exactly one bit is $1$. The vector can be accessed via oracle queries of the form: ``Is the $i$-th bit equal to $1$?''. The   answer to a query is ``Yes' or ``No''.
		
		The task is to find the location of the hidden bit   in at most $k \in \mathbb{N}$ rounds of interaction with the oracle. An index must be queried before getting returned as the solution \footnote{This requirement is benign as it only makes a difference of $\pm 1$ in the bounds.}. 
	\end{definition}

	Let $\mbox{\em unordered}_{n,k}$ denote the unordered search problem on an input vector of length $n$ in $k$ rounds. 
	We have the following bounds for the randomized complexity of this problem.

	\begin{theorem}[Unordered search, randomized algorithms on worst case input] \label{thm:intro:search_unordered_rand}
		For all $k,n \in \mathbb{N}^*$ and  $p \in [0,1]$, we have: 
		$ n p  \Bigl(  \frac{k+1}{2k} \Bigr) \leq 	\mathcal{R}_{1-p}(\mbox{\em unordered}_{n,k})  \leq   n p \Bigl( \frac{k+1}{2k} \Bigr) + p + p/n  \,.
		$
	\end{theorem}
	
	We also analyze the distributional complexity. We say the input is drawn from distribution $\Psi = (\Psi_1, \ldots, \Psi_n)$ if the hidden bit is at location $i$ with probability $\Psi_i$,  where $\Psi_j \geq 0 $ for all $j \in [n]$ and $\sum_{j=1}^n \Psi_j = 1$. The distributional complexity is bounded as follows.
	
	\begin{theorem}[Unordered search, deterministic algorithms on worst case input distribution] \label{thm:intro:search_unordered_det_uniform}
		For all $k, n \in \mathbb{N}^*$ and   $p \in [0,1]$: 
		\begin{align}
			n p \Bigl( 1 - \frac{k-1}{2k} \cdot p  \Bigr) \leq 		\mathcal{D}_{1-p}(\mbox{\em unordered}_{n,k}) \leq      n p \Bigl( 1 - \frac{k-1}{2k}  \cdot p \Bigr) + 1 + p +  2/n   \,. 
		\end{align}
		The uniform distribution is  the worst case for  unordered search.
	\end{theorem}

	\begin{figure}[h!]
		\centering
		\includegraphics[scale=1.4]{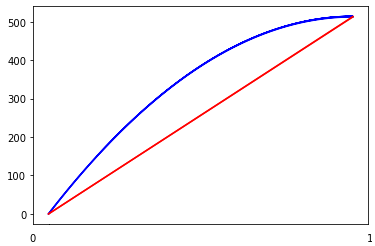}
		\label{fig:search_unsorted_fixed_n_k}
		\caption{The query complexity of fully adaptive unordered search for $n=2^{10}$ elements,  with success probability $p$ ranging from $0$ to $1$. The $X$ axis is for the success probability $p$, while the $Y$ axis is for the expected number of queries. The randomized complexity  is plotted in red (both upper and lower bounds and they coincide) and similarly the distributional complexity in blue. }
	\end{figure}
	
	Combining Theorem~\ref{thm:intro:search_unordered_rand} and \ref{thm:intro:search_unordered_det_uniform}, we obtain that for  each  $p \in (0,1)$, there exists $n_{0} = n_{0}(p) \in \mathbb{N}$ such that  for all $k,n \in \mathbb{N}$ with  $n \geq n_0$, the multiplicative gap between the distributional and randomized complexity of unordered search in $k$ rounds with success probability $p$ is 
	\begin{align}
		\frac{\mathcal{D}_{1-p}(\mbox{unordered}_{n,k})}{\mathcal{R}_{1-p}(\mbox{unordered}_{n,k})} = 
		1 + \frac{(k-1)(1-p)}{k+1}  \pm o(1) \,. \label{eq:multiplicative_gap_unordered_rounds}
	\end{align}
	The gap in \eqref{eq:multiplicative_gap_unordered_rounds} grows from $1$ to  $\approx (2-p)$  as the number of rounds grows from $k=1$ to $k=n$.
	
	\paragraph{Fully adaptive unordered search.}
	By taking $k=n$, the bounds in Theorem~\ref{thm:intro:search_unordered_rand} and \ref{thm:intro:search_unordered_det_uniform} characterize the query complexity of the fully adaptive unordered search problem, denoted $\mbox{unordered}_{n}$. 
	
	\begin{corollary}[Fully adaptive unordered search] \label{corr:fully_adaptive_unordered}
		
		Let  $n \in \mathbb{N}^*$ and  $p \in [0,1]$. The randomized and distributional complexity of fully adaptive unordered search with  success probability $p$ are:
		\begin{itemize} 
			\item   
			$ np \left( \frac{n+1}{2n} \right) \leq  \mathcal{R}_{1-p}(\mbox{\em unordered}_{n})  \leq  np \left(  \frac{n+1}{2n} \right) + p + {p}/{n}\,.
			$ 
			\item  $np\bigl(1-\frac{n-1}{2n} \cdot p \bigr) \leq \mathcal{D}_{1-p}(\mbox{\em unordered}_{n})  \leq   np \bigl( 1 - \frac{n-1}{2n} \cdot p \bigr) +  1 + p + {2}/{n} \,.$
		\end{itemize}
	\end{corollary}
	The randomized complexity  is roughly  
	${np}/{2}$ and  the distributional complexity  roughly $np \bigl(1 - \frac{p}{2} \bigr)$.

	\begin{corollary}[Multiplicative and additive gap for fully adaptive unordered search]
		For each success probability $p \in (0,1]$,  we have 
		\begin{align}
			\lim_{n \to \infty} \frac{\mathcal{D}_{1-p}(\mbox{\em unordered}_{n})}{\mathcal{R}_{1-p}(\mbox{\em unordered}_{n})} =  2-p \;\; \mbox{and} \;\; 
			\lim_{n \to \infty} {\mathcal{D}_{1-p}(\mbox{\em unordered}_{n})}-{\mathcal{R}_{1-p}(\mbox{\em unordered}_{n})} = \infty \,. 
		\end{align}
	\end{corollary}
	
	
	\subsubsection{Ordered Search} 
	
	
	The ordered search problem is formally defined next. The difference from unordered search is that the array is sorted and the oracle gives feedback about the direction in which to continue the search in case of a ``No'' answer.
	
	\begin{definition}[Ordered search]
		The  input is  a bit vector $\vec{x} = (x_1, \ldots, x_n) \in \{0,1\}^n$ with the promise that exactly one bit is set to $1$. The vector can be accessed  via oracle queries of the form: ``Is the $i$-th bit equal to $1$?''. The answer to a query is: ``Yes', ``No, go left'', or ``No, go right''.
		
		The task is to find the location of the hidden bit using at most $k \in \mathbb{N}$ rounds of interaction with the oracle. An index must be queried before getting returned as the solution. 
	\end{definition}

	
	Let $\mbox{ordered}_{n,k}$ denote the ordered search problem on an input vector of length $n$ in $k$ rounds. 
	For ordered search the number of rounds need not be larger than $\lceil \log_2{n} \rceil$, since binary search is an optimal fully adaptive algorithm for success probability $1$.
	We have the following bounds.
	
	\begin{theorem}[Ordered search, randomized and distributional complexity] \label{thm:main_ordered}
		For all  $k,n \in \mathbb{N}^*$ and $p \in [0, 1]$, we have:  
		\begin{align}
			k p n^{\frac{1}{k}} - 2 pk \leq  \mathcal{R}_{1-p}(\mbox{\em ordered}_{n,k}) \leq \mathcal{D}_{1-p}(\mbox{\em ordered}_{n,k}) \leq k  \lceil p n^{\frac{1}{k}} \rceil + 2 \,. \notag 
		\end{align}
		Moreover, $np \leq R_{1-p}(\mbox{\em ordered}_{n,1}) \leq D_{1-p}(\mbox{\em ordered}_{n,1}) \leq \lceil np \rceil$. 
		The uniform distribution is the worst case for ordered search.
	\end{theorem}
	
	Theorem~\ref{thm:main_ordered} shows that  for ordered search in constant rounds, there is essentially no gap between the randomized and distributional complexity.

	
	
	\subsubsection{Cake Cutting in Rounds and Sorting with Rank Queries} 
	We consider the cake cutting problem of finding a proportional allocation with contiguous pieces in $k$ rounds. The cake is the interval $[0,1]$ and the goal is to divide it among $n$ players with private additive valuations. A \emph{proportional} allocation, where each player gets a piece worth $1/n$ of the total cake according to the player's own valuation, always exists and can be computed in the standard (RW) query model for cake cutting. 
	
	We establish a connection between proportional cake cutting with contiguous pieces and sorting in rounds in the \emph{rank query model}. In the latter, we have oracle access to a list $x$ of $n$ elements that we cannot inspect directly. The oracle accepts rank queries of the form ``How is $rank(x_i)$ compared to $j$?'', where the answer is ``$<$'', ``$=$'', ``$or >$''\footnote{Equivalently, the queries are ``Is $rank(x_i) \leq j$?'', where the answer is \emph{Yes} or \emph{No}.}. 

	\begin{theorem}(Informal). \label{thm:intro:sorting_cake_equiv} For all $k,n \in \mathbb{N}^*$, the following problems are equivalent:
		\begin{itemize}
			\item computing a proportional cake allocation with contiguous pieces for $n$ agents in the standard (RW) query model
			\item sorting a vector with $n$ elements using rank queries.
		\end{itemize}
		The randomized query complexity of both problems (for constant success probability) is $\Theta\bigl(k \cdot n^{1 + \frac{1}{k}}\bigr)$.
	\end{theorem}
	
	We prove Theorem \ref{thm:intro:sorting_cake_equiv} in Appendix~\ref{app:cake_cutting_and_sorting}.	
	We design an optimal protocol for proportional cake cutting in $k$ rounds. En route, we re-examine the implicit reduction from sorting with rank queries to proportional cake cutting as presented in Woeginger (2007), and make it completely precise.
	
	\cite{AlonAzar88} gave a lower bound of $\Omega(k n^{1+\frac{1}{k}})$ for sorting a vector of $n$ elements  with rank queries in $k \leq  \log{n}$ rounds. We also show a slightly improved deterministic lower bound for sorting with rank queries that has a simpler proof. 
	
	Finally, to highlight the connection to Ordered Search, we point out that an operation implicit in sorting with rank queries is \emph{Locate}:  given a vector $x = (x_1, \ldots , x_n)$ and an element $x_i$, find its rank via rank queries. Locate with rank queries is equivalent to the ordered search problem.


	\subsection{Related work}

	\paragraph{Parallel complexity.}

	
	Parallel complexity is a fundamental concept with a long history in areas such as sorting and optimization; see, e.g.  \cite{NEMIROVSKI} on the parallel complexity of optimization and  more recent results on submodular optimization~\cite{BalkanskiS18}. An overview on parallel sorting algorithms is given in the book \cite{Akl_book} and many works on sorting and selection in rounds~ \cite{Val75,pippenger1987sorting,bollobas1988sorting,alon1986tight,wigderson1999expanders,GasarchGK03}, aiming to understand the tradeoffs between the number of rounds of interaction and the query complexity.
	
	\cite{Val75} initiated the study of parallelism using the number of comparisons as a complexity measure and showed that $p$ processor parallelism can offer speedups of at least $O\bigl(\frac{p}{\log \log p}\bigr)$ for problems such as sorting and finding the maximum of a list of $n>p$ elements. The connection to the problem of sorting in rounds is straight-forwards since one parallel step of the $p$ processors (e.g. $p$ comparisons performed in parallel) can be viewed as one round of computations.  
	
	\cite{haggkvist1981parallel} showed that $O\bigl(n^{\frac{3\cdot2^{k-1}-1}{2^k-1}}\log{n}\bigr)$ comparisons suffice to sort an array in $k$ rounds. 
	\cite{bollobas1983parallel} showed a bound of $O(n^{3/2}\log{n})$ for two rounds. \cite{pippenger1987sorting} made a connection between expander graphs and sorting and proved that $O\bigl(n^{1+\frac{1}{k}}(\log{n})^{2-\frac{2}{k}}\bigr)$ comparisons are enough. This was improved to $O\bigl(n^{3/2}\frac{\log{n}}{\sqrt{\log\log{n}}}\bigr)$ in  \cite{alon1988sorting}, which also showed that $\Omega ( n^{1 + 1/k} ( \log n )^{1/k} )$ comparisons are needed.
	
	\cite{bollobas1988sorting} generalized the latter upper bound to $O\bigl(n^{1 + 1/k}\frac{(\log{n})^{2-2/k}}{(\log\log{n})^{1-1/k}}\bigr)$ for $k$ rounds. The best upper bound known to us is due to  \cite{wigderson1999expanders}, which  obtained a $k$-rounds algorithm that performs $O\bigl(n^{1 + 1/k+o(1)}\bigr)$ comparisons. For randomized algorithms, \cite{AZV86} obtained an algorithm that runs in $k$ rounds and issues $O\bigl(n^{1+1/k}\bigr)$ queries, thus demonstrating that randomization helps in the comparison model. Local search in rounds was considered in~\cite{BL22_COLT}.

	\paragraph{Randomized complexity.} 
	
	The expected cost setting that we consider is the one studied in \cite{yao_minimax}. However, most of the literature since then has focused on the worst case setting, where the cost of an algorithm is the worst case cost among all possible inputs and coin-flips (for randomized algorithms). 
	In more detail, consider a function $f : \mathcal{X}_n \to \mathcal{Y}_m$, with $m, n \in \mathbb{N}$ and $\mathcal{X}_n \subseteq \{0,1\}^n$ and $\mathcal{Y}_m \subseteq \{0,1\}^m$.
	
	The worst-case randomized complexity for error $\delta$, denoted $\widehat{\mathcal{R}}_{\delta}(f)$, is defined as the maximum number of queries issued by a randomized algorithm $R$, where the maximum is taken over all inputs $x \in \mathcal{X}_n$ and coin-tosses, and $R$ has the property $\mbox{err}_f(R,x) \leq \delta$ for all $x \in \mathcal{X}_n$. 
	The worst-case distributional complexity for error $\delta$, denoted $\widehat{\mathcal{D}}_{\delta}(f)$, is  the maximum number of queries issued by an optimal deterministic algorithm $\mathcal{A}$  that computes $f$ with error probability  $\delta$ when the input is drawn from a worst case input distribution $\Psi$, where the optimality of $\mathcal{A}$ is with respect to $\Psi$. It is known that: $\widehat{\mathcal{R}}_{\delta}(f) = \widehat{\mathcal{D}}_{\delta}(f)$ \cite{CC_book}.
	
	In recent work, \cite{small_minimax_bias_BBGM22} also focus on the expected cost setting and analyze the gap between the expected query complexity of randomized algorithms on worst case input and the   expected query complexity of randomized algorithms on a worst case input distribution, in the regime where the error  probability is $\delta \approx 1/2$. 
	
	\paragraph{Group testing.} 
	In fault detection, the goal is to identify all the defective items from a finite set items via a minimum number of tests. More formally, there is a universe of $\mathcal{U}$ of $n$ items, $d$ of which are defective. Each test is executed on a subset $S \subseteq \mathcal{U}$ and says whether $S$ is contaminated (i.e. has at least one defective item) or pure (i.e. none of the items in $S$ are defective). Questions include how many tests are needed to identify all the defective items and how many stages are needed, where the tests performed in round $k+1$ can depend on the outcome of the tests in round $k$. An example of group testing is to identify which people from a set are infected with a virus, given access to any combination of individual blood samples; combining their samples allows detection using a smaller number of tests compared to checking each sample individually.
	
	The group testing problem was  posed in~\cite{dorfman} and a lower bound of $\Omega\Bigl(d^2 \frac{\log{n}}{\log{d}}\Bigr)$ for the number of tests required in the one round setting was given in~\cite{dyachkov}. One round group testing algorithms with an upper bound of $O(d^2 \log{n})$ on the number of tests were designed in~\cite{AMS06,PR08,INR10,NPR11}. Two round testing algorithms were studied in~\cite{BGV05,EGH07}. The setting where the number of rounds is allowed is given by some parameter $r$ and the number of defective items is not known in advance was studied  in~\cite{DP94,CDZ15,Damashke19,GV20}; see  \cite{book_group_testing} for a survey. 
	
	\paragraph{Fair division.}
	The cake cutting model was introduced in~\cite{Steinhaus48} to study the allocation of a heterogeneous resource among agents with complex preferences. Cake cutting was studied in mathematics, political science, economics \cite{RW98,BT96,Moulin03}, and computer science \cite{Pro13,socialchoice_book,GP14}. 
	There is a hierarchy of fairness notions such as  proportionality, envy-freeness (where no player prefers the piece of another player), equitability, and necklace splitting~\cite{Alon_necklace}, with special cases such as consensus halving and perfect partitions. See ~\cite{branzei15,pro16_handbook} for  surveys.
	
	Cake cutting protocols are often studied in the Robertson-Webb~\cite{woeginger2007complexity} query model, where a mediator asks the players queries until  it has enough information to output a fair division.
	\cite{EP84} devise an algorithm for computing a proportional allocation with connected pieces that asks $O(n \log{n})$ queries, with matching lower bounds due to \cite{woeginger2007complexity} and \cite{EP06b}. 
	
	For the query complexity of exact envy-free cake cutting (possibly with disconnected pieces), a lower bound of $\Omega(n^2)$ was given by \cite{Pro09} and an upper bound of $O\bigl(n^{n^{n^{n^{n^{n}}}}}\bigr)$ by \cite{AM16}. ~\cite{AFMPV18} found a simpler algorithm  for $4$ agents. An upper bound on the query complexity of equitability was given by \cite{CDP13} and a lower bound by \cite{PW17}. \cite{BN17} analyzed the query complexity of envy-freeness, perfect, and equitable partitions with minimum number of cuts. 
	
	The issue of rounds in cake cutting was studied in \cite{BN19}, where the goal is to bound the communication complexity of protocols depending on the fairness notion. 
	The query complexity of proportional cake cutting with different entitlements was studied by \cite{Halevi18}. 
	The query complexity of consensus halving was studied in \cite{DFH20} for  monotone valuations, with an appropriate generalization of the Robertson-Webb query model.
	The query complexity of cake cutting in one round, i.e. in the simultaneous setting, was studied in \cite{BBKP14}.

	Many  works analyzed the complexity of fair division in models such as cake cutting and  multiple  divisible and indivisible goods. \cite{DeligkasFH21,GoldbergHS20,Cheze_whp_ef,GoldbergHIMS20,FHSZ20,FHHH21,AlonGraur,PlautR19,BiloCFIMPVZ19,deligkas2022complexity} studied the complexity of  cake cutting. Indivisible goods were studied, e.g., in~\cite{OPS21_few_queries} for their query complexity and in ~\cite{MS21_closing_gaps,ChaudhuryKMS21} for algorithms.
	Cake cutting with separation was studied in~\cite{ElkindSS21}, fair division of a graph or graphical cake cutting in \cite{graph_fair_division,Bei2019DividingAG},  multi-layered cakes in ~\cite{IgarashiM21},  fair cutting in practice  in~\cite{KyropoulouOS22}, and cake cutting where some parts are good and others bad  in ~\cite{cake_burned} and  when the whole cake is a ``bad'' in \cite{ijcai2018p31}.
	Branch-choice protocols  were developed  in~\cite{IaruGoldberg21} as a simpler yet expressive alternative  for GCC protocols from \cite{BranzeiCKP16}. A body of work analyzed truthful cake cutting both in the standard (RW) query model~\cite{mossel_tamuz,BM15} and in the direct revelation model~\cite{truth_justice_and_cake_cutting,BU2023103904,bei2022truthful,Tao22}.
	
	\section{Ordered search}
	
	In  this section we focus on ordered search and  prove Theorem~\ref{thm:main_ordered}.
	The omitted proofs of this section can be found in Appendix~\ref{app:ordered_search}. 
	
	
	\subsection{Deterministic ordered search algorithm on worst case input}
	
	We first design a deterministic algorithm ${D}^o$ for ordered search that always succeeds and asks at most $k  \lceil n^{\frac{1}{k}} \rceil $ queries on each input.
	
	\begin{proposition} \label{prop:k_round_deterministic_algorithm}
		For each $n \in \mathbb{N}^*$ and $k \in [\lceil \log{n} \rceil]$, 
		there is a deterministic $k$-round algorithm for ordered search that  succeeds on every input and asks at most $k  \lceil n^{\frac{1}{k}}\rceil$ queries in the worst case.
	\end{proposition}
	The algorithm ${D}^o$ that achieves this upper bound  issues $n^{\frac{1}{k}}$ queries in the first round, which are as equally spaced as possible, partitioning the array in $n^{\frac{k-1}{k}}$ blocks. If the element is found at one of the locations queried in the first round, then ${D}^o$ returns it and halts. Otherwise, ${D}^o$ recurses on the  block that contains the solution in the remaining $k-1$ rounds. 
	\subsection{Randomized  ordered search algorithm on worst case input}

	Using ${D}^o$, for each $p \in (0,1]$, we design a randomized algorithm ${R}^o$ that succeeds with probability at least $p$ and asks at most $p k \lceil n^{\frac{k}{k}} \rceil$ queries in expectation. 
	
	\begin{proposition} \label{cor:ordered_randomized_p}
		Let $p \in (0,1]$ and $k,n \in \mathbb{N}^*$. Then 
		$ \mathcal{R}_{1-p}(\mbox{unordered}_{n,k}) \leq  pk   \lceil n ^{\frac{1}{k}} \rceil$.
	\end{proposition}
	The  randomized algorithm ${R}^o$ has an all-or-nothing structure:
	\begin{itemize} 
		\item  with probability $1-p$, do nothing (i.e. output the empty string); 
		\item  with probability $p$, run the deterministic algorithm ${D}^o$ from Proposition~\ref{prop:k_round_deterministic_algorithm}.
	\end{itemize}
	
	\subsection{Deterministic  ordered search algorithm on worst case input distribution}
	
	Next we  upper bound the distributional complexity of ordered search.
	\begin{proposition} \label{prop:deterministic_alg_on_hard_distribution} Let $p \in (0,1]$ and $k, n \in \mathbb{N}^*$. Then  $\mathcal{D}_{1-p}(\mbox{\em ordered}_{n,k}) \leq k  \lceil p n^{\frac{1}{k}} \rceil + 2 $. Moreover, $ \mathcal{D}_{1-p}(\mbox{\em ordered}_{n,1}) \leq \lceil np \rceil $.
	\end{proposition}
	\begin{proof}[Proof sketch] We include the proof sketch, while the formal details can be found in Appendix~\ref{app:ordered_search}.
		
		Using ${D}^o$ and ${R}^o$, we  show how for each $p \in (0,1]$, if the input is drawn from an arbitrary distribution $\Psi = (\Psi_1, \ldots, \Psi_n)$, one can design  a deterministic algorithm ${D}_{\Psi}^o$ that asks at most $k \lceil pn^{\frac{1}{k}} \rceil + 2$ queries in expectation and succeeds with probability at least $p$.
		The distribution-dependent deterministic algorithm ${D}_{\Psi}^o$ will simulate the execution of ${R}^o$ using the following steps. 
		
		\begin{description}
			\item[{Step 1.}]	Given   $\Psi$, define probability density $v: [0,1] \to \mathbb{R}$ by $v(x) = n  \Psi_i$ $\; \forall i \in [n]$ $\forall x \in \left[\frac{i-1}{n}, \frac{i}{n}\right]$.
			
			Let $\mathcal{C}$ denote the circle obtained by bending the interval $[0,1]$ so that the point $0$ coincides with $1$.
			A fixed point theorem  (Lemma~\ref{lem:length_p_probability_mass_p_ordered}) ensures  there is a  point $c \in [0,1]$ such that the interval $[c,c+p]$ on the circle $\mathcal{C}$ has probability mass $p$ (and length $p$). That is:
			\begin{enumerate}[(a)]
				\item $\int_{c}^{c+p} v(x) \, dx = p$, where $0 \leq c \leq 1-p$;  or 
				\item 
				$\int_{0}^{c}v(x) \, dx + \int_{c+1-p}^{1} v(x) \, dx = p $, where $ 1-p < c < 1$.
			\end{enumerate} 
			\begin{figure}[h!]
				\centering
				\includegraphics[scale=0.85]{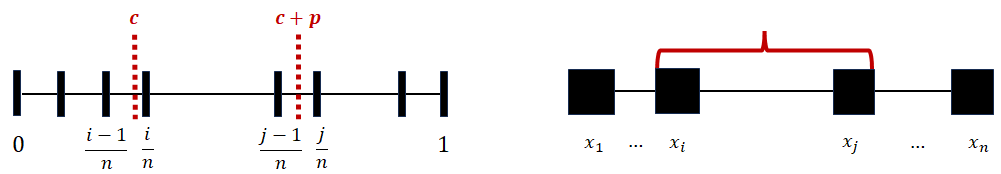}
				\caption{\small{Illustration for case (a) in step 1. Given  $\Psi = (\Psi_1, \ldots, \Psi_n)$, define  $v: [0,1] \to \mathbb{R}_{\geq 0}$  by $v(x) = n \cdot \Psi_{\ell}$ for all $\ell \in [n]$ and $x \in [(\ell-1)/n, \ell/n]$. The left figure shows the point $c$ with $\int_{c}^{c+p} v(x)\, dx  = p$. The right figure shows the queried sub-array  $\vec{y}_{\Psi} = [x_i, \ldots, x_j]$, of  length $\approx np$ and probability mass $\approx p$.}} \label{fig:intro:illustration_deterministic_on_hard_distribution_case_I}
			\end{figure}
			\item[{Step 2.}] The points $c$ and $c+p$ can be mapped to indices $i \in [n]$ and $j \in [n]$, respectively, so that one of the following conditions holds:
			\begin{description}
				\item[$\bullet$] $\vec{y}_{\Psi} = [x_i, \ldots, x_j]$ has length  $\approx np$ and probability mass   $\sum_{\ell=i}^{j} \Psi_{\ell} \approx p$; or 
				\item[$\bullet$] $\vec{y}_{\Psi} = [x_1, \ldots, x_i, x_j, \ldots, x_n]$ has length  $\approx np$ and probability mass $\sum_{\ell=1}^{i} \Psi_{\ell} + \sum_{\ell=j}^{n} \Psi_{\ell}\approx p$.
			\end{description}
			\item[{Step 3.}]		In the first round, algorithm ${D}_{\Psi}^o$ queries locations $i$ and $j$, as well as $\approx pn^{\frac{1}{k}}$ other equally spaced locations in the sub-array $\vec{y}_{\Psi}$. These queries create approximately $ pn^{\frac{1}{k}}$ blocks of size roughly $  \bigl( \frac{np}{pn^{{1}/{k}}} \bigr) \approx n^{\frac{k-1}{k}}$ each. Then:
			\begin{description}
				\item[$\bullet$] If the first round queries reveal the hidden element is not in $\vec{y}_{\Psi}$, then ${D}_{\Psi}^o$ gives up right away (i.e. outputs the empty string).
				\item[$\bullet$] Else, if the element is found at a location queried in round $1$, then ${D}_{\Psi}^o$ returns it and halts.
				\item[$\bullet$] Else, in the remaining $k-1$ rounds, run   ${D}^o$ on the block  identified to contain the element.
			\end{description} 
		\end{description}

		\paragraph{Expected number of queries of ${D}_{\Psi}^o$.} The block identified at the end of the first round has length  $\approx  n^{\frac{k-1}{k}}$. Moreover, ${D}_{\Psi}^o$ continues to the second round with probability $\approx p$. Thus the success probability is roughly $ p$ and the  
		total expected number of queries is approximately   
		\[  \bigl( pn^{\frac{1}{k}} + 2 \bigr) + p \cdot (k-1)\bigl(n^{\frac{k-1}{k}}\bigr)^{\frac{1}{k-1}} = pk n^{\frac{1}{k}} + 2 \,.
		\]
		
		In summary,  the deterministic algorithm ${D}_{\Psi}^o$ is able to generate an event of probability  $\approx p$ via the first round queries while also pre-partitioning a relevant sub-array. If the event does not take place, then ${D}_{\Psi}^o$ gives up. Otherwise, it  runs an optimal deterministic $(k-1)$-round algorithm on the block identified via the first round queries. 
		This strategy enables ${D}_{\Psi}^o$ to simulate the all-or-nothing structure of the optimal randomized algorithm and catch up with it fast enough so that the query complexity remains essentially the same. 
	\end{proof}
	
	
	\subsection{Lower bounds for  ordered search}
	
	We prove the next lower bound for randomized algorithms that succeed with probability $p$. 

	\begin{proposition} \label{thm:randomized_lb_k_rounds_ordered_search}
		Let $k,n \in \mathbb{N}^*$ and $p \in (0,1]$. Then 
		$\mathcal{R}_{1-p}(\mbox{\em ordered}_{n,k})  \geq pk n^{\frac{1}{k}} - 2 pk$ for $k\geq 2$ and $\mathcal{R}_{1-p}(\mbox{\em ordered}_{n,1})  \geq np$.
	\end{proposition}

	This lower bound has the same leading term as the upper bound achieved by  $\mathcal{D}_{\Psi}$, thus showing that the randomized and distributional complexity have the same order. The uniform distribution is the   worst case.

	We prove   Proposition~\ref{thm:randomized_lb_k_rounds_ordered_search}   by induction on the number $k$ of rounds. The induction step requires   showing polynomial inequalities, where the  polynomials involved have high degrees that are themselves functions of $k$. For $k \geq 4$, the roots of such polynomials cannot be found by a formula. To overcome this, we use  delicate  approximations of the polynomials by simpler ones that are more amenable to study yet  close-enough to the original polynomials to yield the required inequalities.
	
	Finally, we obtain the proof of Theorem~\ref{thm:main_ordered}  by combining the upper bound from    Proposition~\ref{prop:deterministic_alg_on_hard_distribution} and  the lower bound from Proposition \ref{thm:randomized_lb_k_rounds_ordered_search}.
	
	
	\section{Unordered search}

	In this section we analyze the unordered search problem and  prove  Theorems~\ref{thm:intro:search_unordered_rand} and \ref{thm:intro:search_unordered_det_uniform}, which quantify the 
	randomized and distributional complexity of unordered search algorithms, respectively. Theorem~\ref{thm:intro:search_unordered_rand} will follow from Propositions~\ref{thm:select_ub_rand} and \ref{thm:randomized_lb_k_rounds_unordered_search} stated next.  Theorem~\ref{thm:intro:search_unordered_det_uniform} will follow from Propositions~\ref{thm:select_ub} and \ref{thm:det_lb_k_rounds_unordered_search}.  
	The omitted proofs of this section are in Appendix~\ref{app:unordered_search}. 
	
	
	\subsection{Deterministic and randomized algorithms for unordered search on a worst case input}
	The maximum number of rounds for unordered search is  $n$.  Since with each location queried the only information an algorithm  receives is whether the element is at that location or not, a  $k$-round deterministic unordered search algorithm  that succeeds on every input  cannot do better than querying roughly $n/k$ queries in each round until finding the element. This gives a total of $n$ queries in the worst case.
	However, randomized algorithms can do better by querying locations uniformly at random. 
	
	\begin{proposition} \label{thm:select_ub_rand} Let $p \in (0,1]$ and $k,n \in \mathbb{N}^*$. Then 
		$\mathcal{R}_{1-p}(\mbox{unordered}_{n,k}) \leq np \cdot \frac{k+1}{2k} + p + \frac{p}{n} \,.$
	\end{proposition}
	The optimal randomized algorithm given by  Proposition~\ref{thm:select_ub_rand} has an all-or-nothing structure: 
	\begin{description}
		\item[$(i)$] with probability $1-p$, do nothing;  
		\item[$(ii)$] with probability $p$, select a uniform random permutation $\pi$ over $[n]$. For all $j \in [k]$, define  $S_j = \{\pi_1, \ldots, \pi_{m_j}\}$, where $m_j = \lceil nj/k\rceil$. In each round $j$, query the locations of $S_j$ that have not been queried in the previous $j-1$ rounds. Once the element is found, return it and halt.
	\end{description}
	
	\subsection{Deterministic algorithms for unordered search on random input}
	We have the following upper bound on the distributional complexity of unordered search.
	
	\begin{proposition} \label{thm:select_ub}
		Let $p \in (0,1]$ and $k,n \in \mathbb{N}^*$. 
		Then 
		\[\mathcal{D}_{1-p}(\mbox{unordered}_{n,k}) \leq 
		np \bigl( 1 - \frac{k-1}{2k} \cdot p\bigr) + 1+ p + \frac{2}{n}\,.
		\]
	\end{proposition}	
	
	Since the unordered search problem has less structure than ordered search, a deterministic algorithm receiving an element drawn from some distribution $\Psi$ will no longer be able to extract enough randomness from the answers to the first round queries to simulate the optimal randomized algorithm. Instead, the optimal deterministic algorithm will establish in advance a fixed set of $np$ locations and query those in the same manner as  step $(ii)$ of the optimal randomized algorithm. 
	
	However, since the search space becomes smaller as an algorithm checks more locations, the fact that the deterministic algorithm is forced to stop after at most $np$ queries  regardless of whether it found the element or not (to avoid exceeding  the optimal expected query bound), is a source of inefficiency. This is the main reason for which a deterministic algorithm receiving a random input cannot do as well as the optimal randomized algorithm that decided in advance to either do nothing or search all the way until finding the solution.


	\subsection{Lower bounds for unordered search}

	Finally, we   lower bound  the randomized and distributional complexity of unordered search. 
	
	\begin{proposition} \label{thm:randomized_lb_k_rounds_unordered_search}
		Let $p \in (0,1]$ and $k,n \in \mathbb{N}^*$. Then $\mathcal{R}_{1-p}(\mbox{unordered}_{n,k}) \geq np \cdot \frac{k+1}{2k} $.
	\end{proposition}
	
	\begin{proposition} \label{thm:det_lb_k_rounds_unordered_search}
		Let $p \in (0,1]$ and $k,n \in \mathbb{N}^*$. Then $\mathcal{D}_{1-p}(\mbox{unordered}_{n,k}) \geq  np\left(1-\frac{k-1}{2k}p\right) $.
	\end{proposition}

	\bibliographystyle{alpha}
	\bibliography{sorting-cake-rounds}
	
	\appendix 
	
	\section*{Roadmap to the appendix}  Appendix \ref{app:ordered_search}  contains the analysis of  ordered search. Appendix \ref{app:unordered_search} contains the analysis for unordered search. Appendix \ref{app:cake_cutting_and_sorting} contains the analysis for cake cutting and sorting in rounds. Appendix~\ref{app:folklore_lemmas} contains folklore lemmas that we use.
	
	\section{Appendix: Ordered search} \label{app:ordered_search}

	In this section we include the omitted proofs for ordered search, which constitute the proof of Theorem~\ref{thm:main_ordered}.  
	
	\subsection{Ordered search upper bounds} 

	In this section we describe an optimal deterministic algorithm for a worst case input,  an optimal randomized algorithm for a worst case input, and an optimal deterministic algorithm for an arbitrary input distribution.
	
	\paragraph{Deterministic algorithms for  a worst case input.} The optimal deterministic algorithm for a worst case input is given in the next proposition. 
	
	\bigskip 
	
	\noindent \textbf{Proposition~\ref{prop:k_round_deterministic_algorithm} (restated).} \emph{For each $n \in \mathbb{N}^*$ and $k \in [\lceil \log{n} \rceil]$, 
		there is a deterministic $k$-round algorithm for ordered search that  succeeds on every input and asks at most $k  \lceil n^{\frac{1}{k}}\rceil$ queries in the worst case.}
	\begin{proof}
		We design a $k$-round algorithm recursively, using induction on $k$.
		\medskip

		\noindent {\em Base case: $k=1$.} Let $\mathcal{A}_1$ be the following algorithm:
		\begin{itemize}
			\item Query all the elements of the array simultaneously. Return the correct location based on the results of the queries.
		\end{itemize}
		Then $\mathcal{A}_1$ runs in one round,  succeeds on every input, and the number of queries is at most $n$.
		
		\medskip

		\noindent{\em Induction hypothesis.} For $k \geq 2$, assume there is a $(k-1)$-round algorithm $\mathcal{A}_{k-1}$ that always succeeds and asks at most $(k-1) \cdot \lceil n^{\frac{1}{k-1}}\rceil$ queries on each array of length $n$. 
		
		\medskip 
		
		\noindent{\em Induction step.} Using the induction hypothesis, we will design a $k$-round algorithm $\mathcal{A}_k$ with the required properties. For each $s \in [n]$, write $n  = s \cdot  u_s + v_s$, for $u_s = \lfloor \frac{n}{s} \rfloor$ and $v_s = n  \; (\text{mod} \; s)$.  Let $\mathcal{A}_k(s)$ be the following algorithm:
		\begin{itemize}
			\item [$(i)$] In round $1$, query  locations $i_1, \ldots, i_s \in [n]$ with the property that   $1 < i_1 < \ldots < i_s = n$. Let $i_0 = 0$. Then these queries create $s$ contiguous blocks $B_1, \ldots, B_s$,  such that $B_j = [i_{j-1} + 1, i_j]$ for $j \in [s]$.
			
			For each $j \in [s]$, 
			set the size of  each block $B_j$ to $\lfloor \frac{n}{s} \rfloor$ if  $j \leq s-v_s$ and to $\lceil \frac{n}{s} \rceil$ if $j > s-v_s$. This uniquely determines indices $i_1, \ldots, i_{s}$. 
			
			If the element searched for is found at one of these $s$ locations, then return that location and halt. Otherwise, identify the index $\ell \in[s]$ for which the block $B_{\ell}$   contains the answer.
			\item [$(ii)$] Given index $\ell$ from step $(i)$ such that block  $B_{\ell} = [i_{\ell-1}+1, i_{\ell}]$ contains the answer, we observe that position $i_{\ell}$ is the only one from block $B_{\ell}$ that has  been queried so far. If $i_{\ell}-1 \geq i_{\ell-1}+1$, let  $\widetilde{B}_{\ell} = [i_{\ell-1}+1, i_{\ell}-1]$ and run algorithm $\mathcal{A}_{k-1}$  on  block $\widetilde{B}_{\ell}$. Else,  halt. 
		\end{itemize}
		We first show algorithm $\mathcal{A}_k(s)$ is correct for every choice of $s$, and then obtain $\mathcal{A}_k$ by optimizing  $s$.
		
		Algorithm $\mathcal{A}_k(s)$ is correct if  the choice of indices $i_1, \ldots, i_s$ is valid. This is the case if the sizes of the blocks  $B_1, \ldots, B_s$ sum up to $n$. We have 
		$\sum_{j=1}^s |B_j| = \left \lfloor {n}/{s} \right \rfloor \cdot (s-v_s) + \left \lceil  {n}/{s} \right \rceil \cdot v_s  \,.  
		$
		\begin{enumerate}[(a)]
			\item If $ v_s = 0$ then  $\lfloor n/s \rfloor = \lceil n/s \rceil = u_s$, so the sum of block sizes is 
			$	\sum\limits_{j=1}^s |B_j| = u_s \cdot (s-v_s) + u_s \cdot v_s =n\,.$
			\item If $v_s > 0$ then $\lceil n/s \rceil = u_s +1$, so  $
			\sum_{j=1}^s |B_j| = u_s \cdot (s-v_s) + (u_s+1) \cdot v_s = u_s \cdot s +v_s =n \,.
			$
		\end{enumerate}
		Combining (a) and (b), we get that the block sizes are valid. Thus $\mathcal{A}_k(s)$ does not skip any indices, so it always finds the element.
		
		\medskip 
		
		Next we argue that there is a choice of $s$ such that by setting $\mathcal{A}_k = \mathcal{A}_k(s)$, we obtain a  $k$-round algorithm that issues at most $k \lceil n^{\frac{1}{k}}  \rceil$ queries.
		
		For a fixed $s \in [n]$, the array size at the beginning of round $2$ is at most $m(s) = \max_{j \in [s]} |B_j| - 1$, since the rightmost element of each block $B_j$ has been queried in round $1$ while the rest of block $B_j$  has not been queried.  Then $m(s) = \max\left\{ \lfloor \frac{n}{s} \rfloor - 1, \lceil \frac{n}{s} \rceil -1 \right\} = \lceil \frac{n}{s} \rceil -1 $.
		
		The total number of queries of algorithm $\mathcal{A}_k(s)$ is at most 
		\begin{align} \label{eq:f_s_def}
			f(s) = s + (k-1) \cdot \left \lceil  m(s)^{\frac{1}{k-1}} \right \rceil = s + (k-1) \cdot \left \lceil  \left(  \left \lceil \frac{n}{s} \right \rceil -1 \right)^{\frac{1}{k-1}} \right \rceil \,.  
		\end{align}
		Taking $s=  \lceil n^{\frac{1}{k}}  \rceil$ in \eqref{eq:f_s_def}, we get 
		\begin{align}
			f\left( \lceil n^{\frac{1}{k}} \rceil\right) & =   \lceil n^{\frac{1}{k}}  \rceil + (k-1) \cdot \left \lceil  \left(  \left \lceil \frac{n}{ \lceil n^{\frac{1}{k}} \rceil} \right \rceil -1 \right)^{\frac{1}{k-1}} \right \rceil \notag \\
			& \leq  \lceil n^{\frac{1}{k}}  \rceil + (k-1) \cdot \left \lceil  \left(  \left \lceil \frac{n}{ n^{\frac{1}{k}}} \right \rceil -1 \right)^{\frac{1}{k-1}} \right \rceil \notag   \\
			& \leq  \lceil n^{\frac{1}{k}}  \rceil + (k-1) \cdot \left \lceil  \left(   \frac{n}{ n^{\frac{1}{k}}}  \right)^{\frac{1}{k-1}} \right \rceil  \notag \\
			& = k \cdot  \lceil  n^{\frac{1}{k}} \rceil \,. \notag 
		\end{align}
		Setting $\mathcal{A}_k = \mathcal{A}_k(\lceil n^{\frac{1}{k}}  \rceil)$, we obtain a correct $k$-round algorithm that issues at most $k \cdot  \lceil  n^{\frac{1}{k}} \rceil$ queries on every array with $n$ elements. This completes the induction step and the proof. 
	\end{proof}
	
	\paragraph{Randomized algorithms for a worst case input.} Building on the optimal deterministic algorithm for worst case input, we design next an optimal randomized algorithm.
	
	\bigskip 
	
	\noindent \textbf{Proposition~\ref{cor:ordered_randomized_p} (restated).} \emph{Let $p \in (0,1]$ and $k,n \in \mathbb{N}^*$. Then 
		$ \mathcal{R}_{1-p}(\mbox{unordered}_{n,k}) \leq  pk   \lceil n ^{\frac{1}{k}} \rceil$.}
	\begin{proof}
		Consider the following randomized algorithm:
		\begin{itemize}
			\item With probability $p$, run the deterministic algorithm $\mathcal{A}_{k}$ from Proposition \ref{prop:k_round_deterministic_algorithm}.
			\item With probability $1-p$, do nothing.
		\end{itemize} 
		On each input, by Proposition \ref{prop:k_round_deterministic_algorithm}, this algorithm succeeds with probability $p$ and issues at most $ pk   \lceil n ^{\frac{1}{k}} \rceil$ queries in expectation, as required.
	\end{proof}
	
	\paragraph{Deterministic algorithms for a random input.} We consider first  the case of $k=1$ rounds. With one round, there is no distinction between ordered and unordered search.
	\begin{proposition} \label{lem:ordered_search_one_round_deterministic_distribution}  
		Let $p \in (0,1]$ and $n \in \mathbb{N}^*$. Then 
		\begin{align}
			\mathcal{D}_{1-p}(\mbox{\em unordered}_{n,1}) \leq \lceil np \rceil \; \; \mbox{and} \; \; \mathcal{D}_{1-p}(\mbox{\em ordered}_{n,1}) \leq \lceil np \rceil \,.
		\end{align}
	\end{proposition}
	\begin{proof}
		Sort the elements of $\vec{x}$ in decreasing order by $\Psi$ and let $\pi$ be the permutation obtained, that is, $\Psi_{\pi_1} \geq \ldots \geq \Psi_{\pi_n}$. Let $\ell$ be  the smallest index for which $\sum_{i=1}^{\ell} \Psi_{\pi_i} \geq p$. Let $q = \sum_{i=1}^{\ell} \Psi_{\pi_i} \geq p$. Consider the following algorithm $\mathcal{A}$: 
		\begin{itemize}
			\item Query elements $x_{\pi_1}, \ldots, x_{\pi_{\ell}}$, i.e. compare each of them with $z$. If there is $i \in [\ell]$ such that $z = x_{\pi_i}$, then return  $\pi_i$.
		\end{itemize}
		By choice of $\ell$, the success probability of this algorithm is $q \geq p$. 
		The number of queries is $\ell$. Let $m = \lceil np\rceil$.  Then  $(m-1)/n < p \leq m/n$. 
		By Lemma~\ref{eq:simple_lemma_sorted_array}, we have 
		$\Psi_{\pi_1} + \ldots + \Psi_{\pi_m} \geq m/n\,.$ Since $\ell$ is the smallest index with $\Psi_{\pi_1} + \ldots + \Psi_{\pi_{\ell}} \geq p$, it follows that  $\ell \leq m  = \lceil np \rceil$.
	\end{proof}
	
	Using the deterministic algorithm of Proposition ~\ref{prop:k_round_deterministic_algorithm} and the randomized algorithm of Proposition~\ref{cor:ordered_randomized_p}, we can now design a deterministic algorithm that is designed to be optimal when the input is drawn from a distribution $\Psi$.
	
	\bigskip 
	
	\noindent \textbf{Proposition~\ref{prop:deterministic_alg_on_hard_distribution} (restated).} \emph{ Let $p \in (0,1]$ and $k, n \in \mathbb{N}^*$. Then  $\mathcal{D}_{1-p}(\mbox{\em ordered}_{n,k}) \leq k  \lceil p n^{\frac{1}{k}} \rceil + 2 $. Moreover, $ \mathcal{D}_{1-p}(\mbox{\em ordered}_{n,1}) \leq \lceil np \rceil $.}
	\begin{proof}
		
		The upper bound of $ D_{1-p}(\mbox{\em ordered}_{n,1}) \leq \lceil np \rceil $ for $k=1$ rounds holds by Proposition~\ref{lem:ordered_search_one_round_deterministic_distribution}. Thus from now on we can assume $k \geq 2$.
		
		At a high level, given input distribution $\Psi$, the deterministic algorithm for this distribution will consists of two steps:
		\begin{itemize}
			\item First, observe there  exists an interval $[i,j]$ on the array viewed on the circle (i.e. where index $n+1$ is the same as index $1$) that has probability mass roughly $pn$ and length roughly $pn$ as well. Find this interval offline without any queries. 
			\item Second,  use the interval identified in the first step to generate an event with probability $p$, thus simulating the randomized algorithm from Proposition~\ref{prop:k_round_deterministic_algorithm}. 
		\end{itemize}

		Formally, given input distribution $\Psi$,  define a probability density function $v: [0,1] \to \mathbb{R}_{\geq 0}$  by  \[
		v(x) = n \cdot \Psi_{i} \qquad \forall  i \in [n] \; \mbox{and} \; x \in [(i-1)/n, i/n] \,.
		\] 
		Then $\int_{0}^{1} v(x) \, dx = \sum_{i=1}^{n} \frac{1}{n} \cdot n \Psi_i = \sum_{i=1}^{n} \Psi_i = 1$.
		By Lemma~\ref{lem:length_p_probability_mass_p_ordered}, there exists  a point $c \in [0,1]$ such that  one of the following holds:
		\begin{enumerate}[(a)]
			\item $\int_{c}^{c+p} v(x) \, dx = p$, where $0 \leq c \leq 1-p$; 
			\item 
			$\int_{0}^{c}v(x) \, dx + \int_{c+1-p}^{1} v(x) \, dx = p $, where $ 1-p < c < 1$.
		\end{enumerate} 
		
		\paragraph{Case (a).} In this case there exists $c \in [0, 1-p]$ such that $\int_{c}^{c+p} v(x) \, dx = p$.
		
		We first make a few observations and then define the protocol. Let $i,j \in [n]$ be such that 
		\[ 
		\frac{i-1}{n} \leq c < \frac{i}{n} \qquad \mbox{and}  \qquad \frac{j-1}{n} \leq c+p <   \frac{j}{n} \,.
		\] 
		Let $T = j - i +1$. Then $np \leq T \leq np+2.$
		Since each interval $\left[ ({\ell-1})/{n}, {\ell}/{n}\right]$ corresponds to element $x_{\ell}$ of the array, we have $\sum_{\ell=i}^{j} \Psi_i \geq p$. 
		By choice of $i$ and $j$, we have:
		\begin{description} 
			\item[$\bullet$] 
			the sub-array $\vec{y} = [x_i, \ldots, x_j]$ has length $T \leq np+2$ and probability mass $\sum_{\ell=i}^{j} \Psi_i \geq p$. 
			\item[$\bullet$] if $T \geq 2$, the sub-array $\widetilde{\vec{y}} = [x_{i+1}, \ldots, x_{j-1}]$ has length  $T-2 \leq np$ and probability mass $\sum_{\ell=i+1}^{j-1} \Psi_i \leq p$. 
		\end{description}
		
		\begin{figure}[h!]
			\centering
			\includegraphics[scale=1]{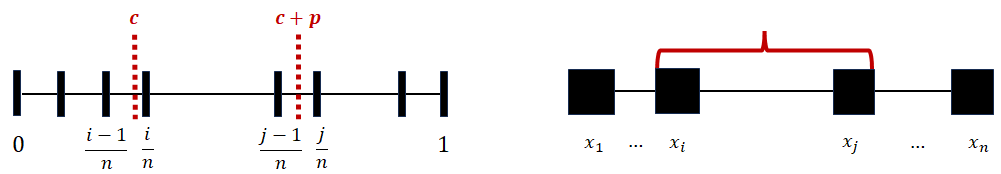}
			\caption{\small{Given  distribution $\Psi = (\Psi_1, \ldots, \Psi_n)$, define probability density $v: [0,1] \to \mathbb{R}_{\geq 0}$  by $v(x) = n \cdot \Psi_{\ell}$ for all $\ell \in [n]$ and $x \in [(\ell-1)/n, \ell/n]$. The left figure shows an interval $[c, c+p]$ of length $p$ and probability mass $\int_{c}^{c+p} v(x) \, dx = p$. The right figure shows the queried sub-array $\vec{y} = [x_i, \ldots, x_j]$, which has  length $T = j - i +1 \leq np+2$ and probability mass $\sum_{\ell=i}^{j} \Psi_{\ell} \geq p$. When $T \geq 2$, the sub-array $\widetilde{\vec{y}} = [x_{i+1}, \ldots, x_{j-1}]$ has length $T-2 \leq np$ and probability mass  $\sum_{\ell=i+1}^{j-1} \Psi_{\ell} \leq  p$. }} \label{fig:illustration_deterministic_on_hard_distribution_case_I}
		\end{figure}
		Let $\mathcal{A}$ be the following  $k$-round protocol:
		\begin{description}
			\item[\emph{Step a.(i)}]  
			\begin{description}
				\item[If $T \leq 2$:] query locations $i$ and $j$ in round $1$. If the element is found, return it and halt.
				\item[Else:] since the element is guaranteed to be in the array $\vec{x}$, it must be the case that  $T \geq 3$. 
				Let $r =  \lceil p \cdot n^{\frac{1}{k}}  \rceil$. 
				Query in round $1$  locations $i$ and $j$, together with additional locations  $t_1, \ldots, t_r$ set as equally spaced as possible. 
				
				More precisely, require   
				$ i+1 \leq t_1 \leq \ldots \leq t_r = j-1   
				$, with    
				$t_0 = i.$
				For each $\ell \in [r]$, let 
				\[ B_{\ell} = [x_{(t_{\ell-1}+1)}, \ldots, x_{t_{\ell}}]
				\] 
				be the $\ell$-th block created by the queries $t_1, \ldots, t_r$. Define  indices $t_1, \ldots, t_r$   so that  each block $B_{\ell}$ has size  at most $\left\lceil \frac{T-2}{r} \right\rceil $, which is possible since  the sub-array $\widetilde{\vec{y}}$ has length $T-2$ and there are $r$ blocks.
				
				If the element is found at one of the indices $i, j, t_1, \ldots, t_r$ queried in round $1$, then return it and halt. Otherwise,  continue to step \emph{a.(ii)}.
			\end{description}
			
			\item[\emph{Step a.(ii)}] If the answers to round $1$  queries show the element is not at one of the indices $[i, \ldots,  j]$, then halt.
			Else, 
			let   $B_{\ell} = [x_{(t_{\ell-1}+1)}, \ldots, x_{t_{\ell}}]$ be the block identified to contain the element, where location $t_{\ell}$  has  been queried. Run the $(k-1)$-round deterministic protocol from Proposition~\ref{prop:k_round_deterministic_algorithm} on the sub-array $\overline{\vec{y}} = [x_{(t_{\ell-1}+1)}, \ldots, x_{(t_{\ell}-1 ) }]$, which always succeeds and asks at most $(k-1)  \cdot \left( \mbox{len}(\overline{\vec{y}}) \right)^{\frac{1}{k-1}}$ queries.
		\end{description}
		
		We now analyze the success probability and expected number of queries of  algorithm $\mathcal{A}$ described in steps \emph{a.(i-ii)}.
		
		\medskip 
		
		\noindent{\emph{Success probability.}} The algorithm is guaranteed to find the element precisely when it is located  in the sub-array $[x_i, \ldots, x_j]$. Since $\sum_{\ell=i}^{j} \Psi_{\ell} \geq p$, the success probability of the algorithm is at least $p$.
		
		\medskip 
		
		\noindent{\emph{Expected number of queries.}} We count separately the expected queries for round $1$ and the remainder.
		The number of queries issued in round $1$ is at most  
		\begin{align}  \label{eq:case_1_contiguous_interval_round_1_ordered}
			2 + r = 2 + \lceil p \cdot n^{\frac{1}{k}} \rceil \,.
		\end{align}
		The algorithm continues beyond round $1$  when the element is  in the sub-array $\widetilde{\vec{y}} = [x_{i+1}, \ldots, x_{j-1}]$, which has length $T-2 \leq np$ and probability mass $\sum_{\ell=i+1}^{j-1} \leq p$.
		
		Thus with probability at least $1-p$, the algorithm halts at the end of round $1$. With  probability at most $p$,  it continues beyond round $1$ by running step \emph{a.(ii)}. The  number of queries in step \emph{a.(ii)} is bounded by 
		\[ (k-1) \bigl( \left \lceil ({T-2})/{r}  \right \rceil  - 1 \bigr)^{\frac{1}{k-1}}
		\] by Proposition~\ref{prop:k_round_deterministic_algorithm} since  $\mbox{len}(\overline{\vec{y}}) \leq \lceil \frac{T-2}{r} \rceil - 1 $.
		Since $T-2 \leq np$ and $r = \lceil p \cdot n^{\frac{1}{k}} \rceil$,  the expected number of queries from step \emph{a.(ii)} can be bounded by  
		\begin{align} 
			& p \cdot (k-1) \left( \left \lceil \frac{T-2}{r} \right \rceil - 1  \right)^{\frac{1}{k-1}} + (1-p) \cdot 0  =  p \cdot (k-1) \left( \left \lceil \frac{T-2}{\lceil p \cdot n^{\frac{1}{k}} \rceil} \right \rceil - 1 \right)^{\frac{1}{k-1}} \notag \\
			& \qquad \leq p \cdot (k-1) \left( \left \lceil \frac{np}{\lceil p \cdot n^{\frac{1}{k}} \rceil} \right \rceil - 1 \right)^{\frac{1}{k-1}} \explain{Since $T-2 \leq np$} \\
			& \qquad \leq p \cdot (k-1) \cdot \left( \left \lceil \frac{np}{ p \cdot n^{\frac{1}{k}} } \right \rceil - 1 \right)^{\frac{1}{k-1}} \explain{Since  $\frac{np}{\lceil p \cdot n^{{1}/{k}} \rceil } \leq \frac{np}{p \cdot n^{{1}/{k}}}$}  \\
			& \qquad  \leq p \cdot (k-1) \cdot \left(   \frac{np}{ p \cdot n^{\frac{1}{k}} } \right)^{\frac{1}{k-1}} = p \cdot (k-1) \cdot n^{\frac{1}{k}}\,. \label{eq:case_1_contiguous_interval_round_2_to_k_ordered}
		\end{align}
		
		Combining \eqref{eq:case_1_contiguous_interval_round_1_ordered} and \eqref{eq:case_1_contiguous_interval_round_2_to_k_ordered}, the expected 
		number of queries of algorithm $\mathcal{A}$ is at most 
		\begin{align}
			2 + \lceil p \cdot n^{\frac{1}{k}} \rceil  + p \cdot (k-1) \cdot n^{\frac{1}{k}}  \leq k \lceil p  n^{\frac{1}{k}} \rceil + 2 \,.
		\end{align}
		
		\paragraph{Case (b).} In this case, there exists $c \in (1-p, 1)$ such that $\int_{0}^{c} v(x) \, dx + \int_{c+1-p}^{1} v(x) \, dx = p$. Let $i, j \in [n]$ be such that $(i-1)/n \leq c \leq i/n$ and  $(j-1)/n \leq c +p-1\leq j/n$. By choice of $i$ and $j$, we have $np \leq T \leq np+2$. Then 
		\begin{description}
			\item[$\bullet$] the sub-array $\vec{y} = [x_1, \ldots, x_i, x_j, \ldots, x_n]$ has length $T = n + i - j +1 \leq np+2$ and probability mass  $\sum_{\ell=1}^{i} \Psi_{\ell}  + \sum_{\ell=j}^{n} \Psi_{\ell}  \geq p$.
			\item[$\bullet$] the sub-array $\widetilde{\vec{y}} = [x_1, \ldots, x_{i-1}, x_{j+1}, \ldots, x_n]$ has length $T -2  \leq np$ and probability mass  \\ $\sum_{\ell=1}^{i-1} \Psi_{\ell} + \sum_{\ell=j+1}^{n} \Psi_{\ell}\geq p$.
		\end{description}

		\begin{figure}[h!]
			\centering
			\includegraphics[scale=1]{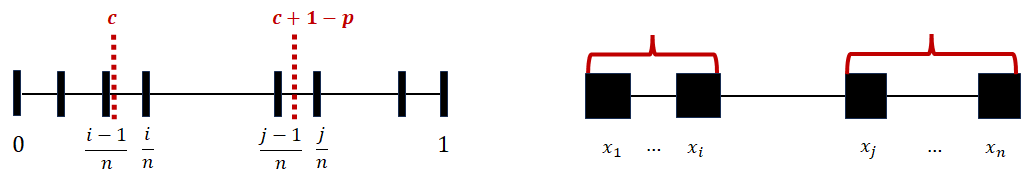}
			\caption{\small{Given  distribution $\Psi = (\Psi_1, \ldots, \Psi_n)$, define  $v: [0,1] \to \mathbb{R}_{\geq 0}$  by $v(x) = n \cdot \Psi_{\ell}$ for all $\ell \in [n]$ and $x \in [(\ell-1)/n, \ell/n]$. The left figure shows point $c$ with probability mass $\int_{0}^{c} v(x)\, dx + \int_{c+1-p}^{1} v(x) \, dx = p$. The right figure shows the queried sub-array consisting of two parts: $\vec{y} = [x_1, \ldots, x_i, x_j, \ldots, x_n]$, of  length $T = n+i-j +1 \leq np+2$ and probability mass $\sum_{\ell=1}^{i} \Psi_{\ell}+ \sum_{\ell=j}^{n} \Psi_{\ell} \geq p$. When $T \geq 2$, the sub-array $\widetilde{\vec{y}} = [x_{1}, \ldots, x_{i-1}, x_{j+1}, \ldots, x_n]$ has length $T-2 \leq np$ and probability mass  $\sum_{\ell=1}^{i-1} \Psi_{\ell}+ \sum_{\ell=j+1}^{n} \Psi_{\ell} \leq  p$.}} \label{fig:illustration_deterministic_on_hard_distribution_case_II}
		\end{figure}
		
		Let $\mathcal{A}$ be the same  $k$-round protocol as in case (a), but where the array $\vec{y}$ is treated as if it were contiguous when making queries:
		\begin{description}
			\item[\emph{Step b.(i)}]  
			\begin{description}
				\item[If $T \leq 2$:] query locations $i$ and $j$ in round $1$. If the element is found, return it.
				\item[Else,   $T \geq 3$.] 
				Let $r =  \lceil p \cdot n^{\frac{1}{k}}  \rceil$. 
				Query in round $1$  locations $i$ and $j$, together with additional locations  $t_1, \ldots, t_r \in \{1, \ldots, i-1, j+1, \ldots, n\}$, set as equally spaced as possible 
				so that 
				for each $\ell \in [r]$, 
				the size of each block  
				$B_{\ell} = [x_{(t_{\ell-1}+1)}, \ldots, x_{t_{\ell}}]$ is at most  $\lceil \frac{T-2}{r} \rceil$. At most one of the blocks may skip over the the indices in  $\{i, \ldots, j\}$.
				If the element is found at one of the queried locations then return it and halt. Else,  go to step \emph{b.(ii)}.
			\end{description}
			\item[\emph{Step b.(ii)}] If  round $1$  indicates that the element is not at one of the indices $\{1, \ldots, i, j, \ldots, n\}$, then halt.
			Otherwise, 
			let   $B_{\ell} = [x_{(t_{\ell-1}+1)}, \ldots, x_{t_{\ell}}]$ be the block identified to contain the element, where location $t_{\ell}$  has  been queried. Run the $(k-1)$-round deterministic protocol from Proposition~\ref{prop:k_round_deterministic_algorithm} on the sub-array $\overline{\vec{y}} = [x_{(t_{\ell-1}+1)}, \ldots, x_{(t_{\ell}-1 ) }]$, which always succeeds and asks at most $(k-1) \cdot \left( \mbox{len}(\overline{\vec{y}}) \right)^{\frac{1}{k-1}}$ queries.
		\end{description}
		
		Next we bound the success probability and expected number of queries when the algorithm executes steps $b.(i)$ and $b.(ii)$.
		
		\medskip 
		\noindent \emph{Success probability.} The algorithm finds the element when its location is one of $[1, \ldots, i, j, \ldots, n]$. Since $\sum_{\ell=1}^{i} \Psi_{\ell} + \sum_{\ell=j}^{n} \Psi_{\ell} \geq p$, the success probability is at least $p$.

		\medskip 
		\noindent \emph{Expected number of queries.}
		The expected number of queries in round $1$ is at most $2 + r =2 + \lceil p n^{\frac{1}{k}} \rceil $, while the number of queries after round $2$ is at most 
		\[
		p \cdot (k-1) \bigl( \left \lceil ({T-2})/{r}  \right \rceil 
		- 1 \bigr)^{\frac{1}{k-1}} \leq p \cdot (k-1) \cdot n^{\frac{1}{k}}\,.
		\]
		Thus the total expected number of queries is at most 
		$
		2 + \lceil p n^{\frac{1}{k}} \rceil +  p \cdot (k-1) \cdot n^{\frac{1}{k}} \leq 2 + k \lceil p n^{\frac{1}{k}} \rceil, 
		$
		which completes the proof.
	\end{proof}

	\subsection{Ordered search lower bounds}
	In this section we prove a lower bound that applies to both randomized algorithms on a worst case input and deterministic algorithms on a worst case input distribution. The lower bound considers the expected query complexity of randomized algorithms on the uniform distribution, which turns out to be the hardest distribution for ordered search.

	\bigskip 
	\noindent \textbf{Proposition~\ref{thm:randomized_lb_k_rounds_ordered_search} (restated).} \emph{Let $k,n \in \mathbb{N}^*$ and $p \in (0,1]$. Then 
		$\mathcal{R}_{1-p}(\mbox{\em ordered}_{n,k})  \geq pk n^{\frac{1}{k}} - 2 pk$ for $k\geq 2$ and $\mathcal{R}_{1-p}(\mbox{\em ordered}_{n,1})  \geq np$.}
	\begin{proof}

		For proving the required lower bound, it  will suffice to assume the input is drawn from the uniform distribution. This means the algorithm is given a bit vector where the location of the unique bit with value $1$ is chosen uniformly at random from $\{1, \ldots, n\}$.  
		If a lower bound  holds for a randomized algorithm  when the input is uniformly distributed, then  by an average argument the same lower bound also holds for a worst case input.
		
		Let  $\mathcal{A}_k$ be  an optimal  $k$-round randomized algorithm that succeeds with probability $p$ when facing the uniform distribution as input. Let $q_k(n,p)$ be the expected number of queries of algorithm $\mathcal{A}_k$ as a function of $n$  and  $p$.
		
		In round $1$, the algorithm has some probability $\delta_m$ of asking $m$ queries, for each $m \in \{0, \ldots, n\}$. Moreover, for  each such $m$, there are different (but finitely many) choices for the positions of the $m$ queries of round $1$. However, since the algorithm is optimal, it suffices to restrict attention to the best way of positioning the queries in round $1$, breaking ties arbitrarily if there are multiple equally good options.
		
		For  each $m \in \{0, \ldots, n\}$, we  define  the following  variables:
		
		\begin{itemize}
			\item $\delta_m$ is the probability that the algorithm asks $m$ queries in round one.
			\begin{figure}[h!] \label{fig:round_1_and_the_blocks_ordered}
				\centering 
				\includegraphics[scale=1]{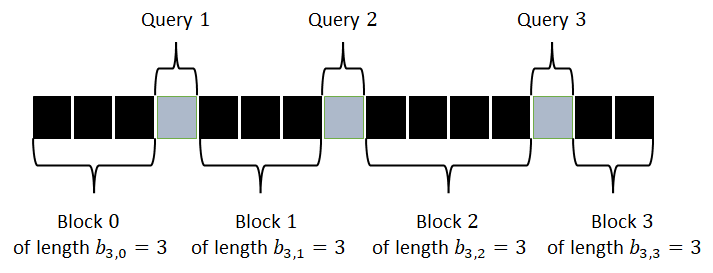}
				\caption{Array with $n=15$ elements. The $m=3$ locations issued in round $1$ are illustrated in gray. The resulting blocks demarcated by these queries are marked, such that the $i$-th block  has length $b_{m,i}$, for $i \in \{0, 1, 2, 3\}$.}
			\end{figure}
			\item $b_{m,i}$ is the size of the $i$-th block demarcated by the indices  queried in round $1$, excluding those indices, counting from left to right, for all  $i \in \{0, \ldots, m\}$. Thus $\sum_{i=0}^{m} b_{m,i} = n-m$. An illustration with an array and the blocks  formed by the queries issued in round $1$ can be found in Figure 1.
			\item  $\alpha_{m,i}$  the success probability of finding the element in the $i$-th block (as demarcated by the indices queried in round $1$),  given that the element is in this block.
		\end{itemize} 
		
		The expected number of queries of the randomized algorithm is 
		
		\begin{align} \label{eq:q_k_n_p}
			q_k(n,p) & = \sum_{m=0}^{n} \delta_m \left[ m + \left( \frac{n-m}{n}\right) \sum_{i=0}^{m}\left( \frac{b_{m,i}}{n-m}\right) \cdot q_{k-1}(b_{m,i}, \alpha_{m,i}) \right] \notag \\
			& = \sum_{m=0}^{n} \delta_m \left[ m + \frac{1}{n} \sum_{i=0}^{m} {b_{m,i}} \cdot q_{k-1}(b_{m,i}, \alpha_{m,i}) \right],
		\end{align}
		where the variables are related by the  following constraints:
		\begin{align}
			&  \sum_{i=0}^{m} b_{m,i} = n-m, \qquad \forall m \in \{0, \ldots, n\} \\
			&  \sum_{m=0}^{n} \delta_m  = 1 \\
			&	p_m = \frac{m}{n} + \frac{n-m}{n}\cdot \sum_{i=0}^{m} \frac{b_{m,i}}{n-m} \cdot \alpha_{m,i} = \frac{m}{n} + \frac{1}{n}\cdot \sum_{i=0}^{m} {b_{m,i}} \cdot \alpha_{m,i}, \qquad \forall m \in \{0, \ldots, n\}   \label{p_m_as_function_sum_of_blocks} \\ 
			& 			p = \sum_{m=0}^n \delta_m \cdot p_m \label{eq:ordered_p_sum_of_p_m} \\ 
			&  b_{m,i} \geq 0, \qquad \forall m \in \{0, \ldots, n\}, i \in \{0, \ldots, m\} \\
			&  0 \leq \alpha_{m,i} \leq 1, \qquad \forall m \in \{0, \ldots, n\}, i \in \{0, \ldots, m\}\\
			& \delta_{m} \geq 0, \qquad \forall m  \in \{0, \ldots, n\} \,.
		\end{align}
		
		Let  $\{\gamma_{\ell}\}_{\ell = 1}^{\infty}$ be the sequence given by  $\gamma_1 = 0$ and $\gamma_{\ell} = 2 \ell$ for $\ell \geq 2$.
		We will prove by induction  on $k$ that 
		\begin{align}  \label{eq:induction_lower_bound_k_ordered}
			q_k(n,p) \geq p \bigl(  k \cdot n^{\frac{1}{k}} -  \gamma_k\bigr) \qquad \forall n,k \geq 1 \mbox{ and }  p \in [0,1] \,.
		\end{align}
		
		\paragraph{Base case.} Proposition \ref{prop:lb_one_round} gives  $q_1(n,p) \geq np$, so  inequality \eqref{eq:induction_lower_bound_k_ordered} holds with $\gamma_1 = 0$ for all $n \geq 1$ and $p \in [0,1]$.
		
		\paragraph{Induction hypothesis.} Suppose  $q_{\ell}(m,s) \geq s \bigl( \ell \cdot  m^{\frac{1}{\ell}} - \gamma_{\ell}\bigr)$ for all $ \ell \in [k-1]$, $m \geq 1$, $s \in [0,1]$.
		
		\paragraph{Induction step.} We will show that \eqref{eq:induction_lower_bound_k_ordered} holds for  $k \geq 2$, where $n \geq 1$ and $p \in [0,1]$. The bound clearly holds when $p = 0$, so we will focus on the scenario  $p > 0$.
		For each $m \in \{0, \ldots, n\}$, define 
		\begin{align} \label{eq:q_k_n_p_m_definition}
			r_k(n,m,p) & =   m +  \frac{1}{n} \cdot  \sum_{i=0}^{m} {b_{m,i}}  \cdot q_{k-1}(b_{m,i}, \alpha_{m,i})  \,.
		\end{align}
		By definition of $q_{k}(n,p)$, 
		\begin{align}
			q_{k}(n,p) = \sum_{m=0}^{n} \delta_m \cdot  r_k(n,m,p)\,. \label{eq:q_k_n_p_as_function_of_r_k_n_p_m}
		\end{align}
		The induction hypothesis implies  $q_{k-1}(b_{m,i}, \alpha_{m,i}) \geq \alpha_{m,i} \cdot \Bigl((k-1)\left(b_{m,i}\right)^{\frac{1}{k-1}}-\gamma_{k-1} \Bigr)$, which substituted in (\ref{eq:q_k_n_p_m_definition}) gives 
		\begin{align} \label{eq:q_k_n_p_m_inequality}
			r_k(n,m,p)   & \geq  m   + \left( \frac{k-1}{n}  \right) \cdot  \sum_{i=0}^{m}   \alpha_{m,i} \cdot \left(b_{m,i}\right)^{\frac{k}{k-1}} - \left( \frac{\gamma_{k-1}}{n}\right)  \cdot \sum_{i=0}^{m}  \alpha_{m,i} \cdot {b_{m,i}}  \,.
		\end{align}
		
		Given a choice of $\alpha_{m,i}, b_{m,i}$ for all $m\in\{0,\ldots, n\}$ and $i \in \{0, \ldots, m\}$, 
		let $i_0, \ldots, i_m \in \{0, \ldots, m\}$ be  such that $0 \leq \alpha_{m,i_0} \leq \ldots \leq \alpha_{m,i_m} \leq 1$. 
		Then we can decompose  $p_m$ using a telescoping sum: 
		\begin{align}  
			p_m & = \frac{m}{n} + \frac{1}{n}\cdot \sum_{i=0}^{m} {b_{m,i}} \cdot \alpha_{m,i} \explain{By \eqref{p_m_as_function_sum_of_blocks}} \\
			& = \alpha_{m,i_0} \cdot \left[ \frac{m}{n} + \frac{1}{n}  \cdot \sum_{\ell=0}^m  b_{m,i_{\ell}}  \right] +  
			\left\{ \sum_{j=1}^{m} (\alpha_{m,i_{j}} - \alpha_{m,i_{j-1}}) \cdot \left[  \frac{m}{n} + \frac{1}{n}  \cdot \sum_{\ell=j}^m  b_{m,i_{\ell}}  \right]\right\} + (1 - \alpha_{m,i_m}) \cdot \frac{m}{n} \,. \label{eq:p_m_decomposed_k_rounds_telescope}
		\end{align}
		We can similarly decompose the right hand side of  inequality \eqref{eq:q_k_n_p_m_inequality}, obtaining:
		\begin{align} 
			r_k(n,m,p)  	&  \geq m + \frac{k-1}{n} \cdot \sum_{i=0}^m \alpha_{m,i} \cdot \left(b_{m,i}\right)^{\frac{k}{k-1}} - \frac{\gamma_{k-1}}{n} \cdot  \sum_{i=0}^m \alpha_{m,i} \cdot b_{m,i} \explain{By \eqref{eq:q_k_n_p_m_inequality}} \\
			& = \alpha_{m,i_0} \cdot \left[  m + \frac{k-1}{n}  \cdot \sum_{\ell=0}^m  \left(b_{m,i_{\ell}}\right)^{\frac{k}{k-1}} - \frac{\gamma_{k-1}}{n}  \cdot \sum_{\ell=0}^m b_{m,i_{\ell}} \right] \notag \\
			& \qquad  + \sum_{j=1}^{m}  (\alpha_{m,i_j} - \alpha_{m,i_{j-1}}) \cdot \left[   m + \frac{k-1}{n}  \cdot \sum_{\ell=j}^m  \left(b_{m,i_{\ell}}\right)^{\frac{k}{k-1}} - \frac{\gamma_{k-1}}{n}  \cdot \sum_{\ell=j}^m b_{m,i_{\ell}} \right] \notag \\
			& \qquad +   (1 - \alpha_{m,i_m}) \cdot m\,.  \label{eq:telescope_objective_m_branch_k_rounds}
		\end{align}

		Let $w_{m,0} = \alpha_{m,i_0}$, $w_{m,j} = \alpha_{m,i_j} - \alpha_{m,i_{j-1}}$ for all $j \in \{1, \ldots, m\}$, and $w_{m,m+1} = 1 - \alpha_{m,i_m}$. Then we can rewrite   \eqref{eq:p_m_decomposed_k_rounds_telescope} and \eqref{eq:telescope_objective_m_branch_k_rounds} as follows:
		
		\begin{align}  \label{eq:inequality_q_k_n_p_m_abstract}
			r_{k}(n,m,p) & \geq \sum_{j=0}^m w_{m, j} \cdot \left[  m + \frac{k-1}{n}  \cdot \sum_{\ell=j}^m  \left(b_{m,i_{\ell}}\right)^{\frac{k}{k-1}} - \frac{\gamma_{k-1}}{n}  \cdot \sum_{\ell=j}^m b_{m,i_{\ell}} \right] + \Bigl(w_{m,m+1} \cdot m \Bigr)  \\
			p_m & =  \sum_{j=0}^m w_{m, j} \cdot \left[  \frac{m}{n} +  \frac{1}{n}  \cdot \sum_{\ell=j}^m b_{m,i_{\ell}} \right] + \Bigl(w_{m,m+1} \cdot \frac{m}{n} \Bigr)\,. \label{eq:identity_p_m_abstract_ordered}
		\end{align}

		For each $m \in \{0, \ldots, n\}$ and $j \in \{0, \ldots, m+1\}$, define   
		\begin{align}
			p_{m,j} & =
			\begin{cases}
				&   \frac{m}{n} +  \frac{1}{n}  \cdot \sum_{\ell=j}^m b_{m,i_{\ell}} \; \; \mbox{ if } j \in \{0, \ldots, m\}\,.
				\\
				& \\ 
				&  \frac{m}{n} \; \; \mbox{ if } j = m+1\,.
			\end{cases} \label{eq:definition_p_m_j_ordered} \\ 
			r_k^j(n,m,p) & =
			\begin{cases}
				m + \frac{k-1}{n}  \cdot \sum_{\ell=j}^m  \left(b_{m,i_{\ell}}\right)^{\frac{k}{k-1}} - \frac{\gamma_{k-1}}{n}  \cdot \sum_{\ell=j}^m   b_{m,i_{\ell}} & \text{ if } j \in \{0, \ldots, m\} \,. \\
				m & \text{ if } j = m+1 \,. 
			\end{cases} \label{eq:definition_q_k_j_n_m_p}
		\end{align}
		Substituting the definition of $r_{k}^{j}(n,m,p)$  in  \eqref{eq:inequality_q_k_n_p_m_abstract}  and that  of  $p_{m,j}$  in \eqref{eq:identity_p_m_abstract_ordered} yields   
		\begin{align} \label{eq:r_k_n_p_m_as_function_of_r_k_j_n_m_p}
			r_k(n,m,p)  \geq  \sum_{j=0}^{m+1} w_{m,j} \cdot r_k^j(n,m,p)  \; \; \mbox{ and } \; \; 
			p_m  = \sum_{j=0}^{m+1} w_{m,j} \cdot p_{m,j}\,.
		\end{align}
		Combining  \eqref{eq:ordered_p_sum_of_p_m}, \eqref{eq:q_k_n_p_as_function_of_r_k_n_p_m}, and \eqref{eq:r_k_n_p_m_as_function_of_r_k_j_n_m_p}, we obtain  
		\begin{align} \label{eq:tying_together_q_k_n_p_function_of_r_k_and_p_as_function_of_p_m}
			& q_k(n,p)  = \sum_{m=0}^n \delta_m \cdot r_k(n,m,p) \geq \sum_{m=0}^n \delta_m \cdot  \left(\sum_{j=0}^{m+1} w_{m,j} \cdot r_k^j(n,m,p) \right)\,. \notag \\
			& p = \sum_{m=0}^n p_m = \sum_{m=0}^n \left( \sum_{j=0}^{m+1} w_{m,j} \cdot p_{m,j}\right) \,.
		\end{align} 
		Let $S_{m,j}  = \sum_{\ell=j}^m b_{m,i_{\ell}}$,  for all $m \in \{0, \ldots, n\}$ and  $j \in \{0, \ldots, m\}$.
		Then for $ j \in \{0, \ldots, m\}$, we have   $n \cdot p_{m,j} = m + S_{m,j}$,  so $S_{m,j} = n \cdot p_{m,j} - m$. Since $S_{m,j} \geq 0$, we have $n \cdot p_{m,j} \geq m$. In summary, 
		\begin{align}
			& \sum_{\ell=j}^{m} b_{m, i_{\ell}} = n \cdot p_{m,j} - m, \qquad  \forall j \in \{0, \ldots, m\} \label{eq:sum_of_j_to_m_blocks_identity_ordered}\\
			& n \cdot p_{m,j} \geq m,  \qquad  \forall j \in \{0, \ldots, m\}  \label{eq:n_times_p_m_j_inequality_at_least_m_ordered}\,.
		\end{align}

		Next we will lower bound $r_k^j(n,m,p)$ and consider two cases, for  $m \geq 1$ and  $m = 0$.
		
		\paragraph{Case $m \geq 1$.} 
		If $j = m+1$, we are in  the scenario where the algorithm asks $m$ queries in round $1$ and no queries in the later rounds. 
		Formally, since  $p_{m,m+1} = m/n$, we have $m = n \cdot p_{m,m+1}$. Using the identity for $r_{k}^{m+1}(n,m,p)$ in \eqref{eq:definition_q_k_j_n_m_p}, we obtain 
		\begin{align} 
			r_k^{m+1}(n,m,p) & = m  =  n \cdot p_{m,m+1} \notag   \\
			& \geq p_{m,m+1} \cdot k n^{\frac{1}{k}} - \gamma_k \cdot p_{m,m+1} \,. \explain{By Corollary~\ref{corr:helper_lemma_1}.}
		\end{align} 
		
		Thus from now on we can assume $j \in \{0, \ldots, m\}$.  Observe that by definition of $p_{m,0}$ in \eqref{eq:definition_p_m_j_ordered}, we have $p_{m,0} = {m}/{n} + \sum\limits_{\ell=0}^{m} b_{m,i_{\ell}}/n = {m}/{n} +{(n-m)}/{n}= 1 \,. $   For all $j \in \{0, \ldots, m\}$, using  \eqref{eq:definition_q_k_j_n_m_p} and Jensen's inequality, we obtain 
		\begin{align}
			r_k^j(n,m,p) & = m + \frac{k-1}{n}  \cdot \sum_{\ell=j}^m  \left(b_{m,i_{\ell}}\right)^{\frac{k}{k-1}} - \frac{\gamma_{k-1}}{n}  \cdot \sum_{\ell=j}^m b_{m,i_{\ell}}  \notag  \\
			&  \geq  m + \frac{(k-1)\left(m-j+1\right)}{n}   \left(\sum_{\ell=j}^m  \frac{b_{m,i_{\ell}}}{m-j+1}\right)^{\frac{k}{k-1}} - \frac{\gamma_{k-1}}{n}   \cdot \sum_{\ell=j}^m b_{m,i_{\ell}}  \,.  \label{eq:case_j_at_least_1_case_a_after_jensen_ordered} 
		\end{align}
		Since $\sum_{\ell = j}^{m} b_{m, i_{\ell}} =   n \cdot p_{m,j} - m$ by \eqref{eq:sum_of_j_to_m_blocks_identity_ordered}, the inequality in \eqref{eq:case_j_at_least_1_case_a_after_jensen_ordered} can be rewritten as 
		\begin{align}
			r_k^j(n,m,p) & 	 \geq   m \left( 1 + \frac{\gamma_{k-1}}{n}\right) + \frac{(k-1)  ( n \cdot p_{m,j} - m)^{\frac{k}{k-1}} }{n \cdot (m-j+1)^{\frac{1}{k-1}}}  - {\gamma_{k-1}} \cdot   p_{m,j}   \,.  \label{eq:inequality_with_r_k_j_containing_mj1_ordered}
		\end{align}
		When $j = 0$, 
		substituting $\sum_{\ell=0}^m b_{m, i_{\ell}} = n - m$ in  \eqref{eq:inequality_with_r_k_j_containing_mj1_ordered}, we obtain  
		\begin{align}
			r_{k}^0(n,m,p) 	&  \geq  m \left( 1 + \frac{\gamma_{k-1}}{n}\right) + \frac{(k-1)({n-m})^{\frac{k}{k-1}}}{n \cdot \left(m+1\right)^{\frac{1}{k-1}}}   - {\gamma_{k-1}}  \notag   \\
			&  \geq k n^{\frac{1}{k}} - \gamma_k  \explain{By Lemma~\ref{lem:helper_lemma_main_2}.}\\
			& = p_{m,0} \cdot  k n^{\frac{1}{k}} - p_{m,0} \cdot \gamma_k \,. \explain{Since $p_{m,0}= 1$}
		\end{align}
		Thus from now on we can assume  $j \in \{ 1, \ldots, m\}$. Using $j \geq 1$ in \eqref{eq:inequality_with_r_k_j_containing_mj1_ordered}, we further get 
		\begin{align} 
			r_k^j(n,m,p)		& \geq m \left( 1 + \frac{\gamma_{k-1}}{n}\right) + \frac{(k-1) ( n \cdot p_{m,j} - m)^{\frac{k}{k-1}}}{n \cdot m^{\frac{1}{k-1}}} - \gamma_{k-1} \cdot p_{m,j}   \,. \label{eq:intermediate_inequality_m_positive_not_m_plus_1_j_nonzero_ordered}
		\end{align}
		In this range of $m$ and $j$, we have $m/n \leq p_{m,j} \leq 1$ and $1/2 < m \leq n \cdot p_{m,j}$ by inequality  \eqref{eq:n_times_p_m_j_inequality_at_least_m_ordered}. Applying Lemma \ref{lem:helper_lemma_main_1} with $c = p_{m,j}$  in \eqref{eq:intermediate_inequality_m_positive_not_m_plus_1_j_nonzero_ordered}, we obtain:
		\begin{align}
			r_k^j(n,m,p)  &\geq    m \left( 1 + \frac{\gamma_{k-1}}{n}\right) + \frac{(k-1) ( n \cdot p_{m,j} - m)^{\frac{k}{k-1}}}{n \cdot m^{\frac{1}{k-1}}} - \gamma_{k-1} \cdot p_{m,j}   \explain{By \eqref{eq:intermediate_inequality_m_positive_not_m_plus_1_j_nonzero_ordered}} \\ 
			& \geq p_{m,j} \cdot  k n^{\frac{1}{k}} - \gamma_k \cdot p_{m,j} \,. \explain{By Lemma~\ref{lem:helper_lemma_main_1}}
		\end{align}
		%

		\paragraph{Case  $m = 0$.} This corresponds to the scenario where the algorithm asks zero queries in round $1$. Since $j \in \{0, \ldots, m+1\}$, it follows that $j = 0$ or $j=1$.

		If $j=0$, then by definition of $p_{m,j}$ we have  $p_{0,0} = 0/n + (1/n) \cdot \sum_{\ell=0}^0 b_{0,i_{\ell}} = 0 + b_{0, i_0}/n$. Since there is only one block, $b_{0,i_0} = n$. Thus $p_{0,0} = 1$.  We get	\begin{align}
			r_k^0(n,0,p) & = 0 + \frac{k-1}{n}  \cdot \sum_{\ell=0}^0  \left(b_{0,i_{\ell}}\right)^{\frac{k}{k-1}} - \frac{\gamma_{k-1}}{n}  \cdot \sum_{\ell=0}^0 b_{0,i_{\ell}}  \notag  \\
			& = ({k-1}) \cdot n^{\frac{1}{k-1}} - \frac{\gamma_{k-1}}{n} \cdot n \explain{Since  $b_{0,i_0} = n$} \\
			& \geq  kn^{\frac{1}{k}} -  \gamma_{k} \explain{By Corollary  \ref{lem:ordered_lb_case_II_0}}\\ 
			& = p_{0,0} \cdot kn^{\frac{1}{k}} - p_{0,0} \cdot \gamma_{k} \,. \explain{Since $p_{0,0} = 1$}
		\end{align}

		If $j=1$, then since $m=0$ we are in the case $j = m+1$. Since $p_{m,m+1} = m/n$, we have $p_{0,1} = 0/n = 0$. Informally, this corresponds to the scenario where the algorithm asks $m=0$ queries in round $1$ and no queries in the later rounds either. Formally, 
		\begin{align}
			r_k^1(n,0,p) & = 0 = p_{0,1} \cdot k n^{\frac{1}{k}} - p_{0,1} \cdot \gamma_k\,.
		\end{align}
		
		\paragraph{Combining cases $m\geq 1$ and $m=0$.} We obtain   
		\begin{align} \label{eq:unified_bound_q_m_j}
			r_k^j(n,m,p) \geq p_{m,j} \cdot kn^{\frac{1}{k}} - p_{m,j} \cdot \gamma_k \; \; \forall m \in \{0, \ldots, n\}, \forall j \in \{0, \ldots, m+1\}
		\end{align}
		Summing inequality \eqref{eq:unified_bound_q_m_j} over all $m \in \{0, \ldots, n\}$ and  $j \in \{0, \ldots, m+1\}$ and using identity \eqref{eq:tying_together_q_k_n_p_function_of_r_k_and_p_as_function_of_p_m} that expresses the total expected number of queries $q_k(n,p)$ as a weighted sum of the $r_k^j(n,m,p)$ terms, we obtain 
		\begin{align}
			q_k(n,p) &  =\sum_{m=0}^n \delta_m \cdot  \left(\sum_{j=0}^{m+1} w_{m,j} \cdot r_k^j(n,m,p) \right) \notag \\
			& \geq \sum_{m=0}^n \delta_m \cdot  \left(\sum_{j=0}^{m+1} w_{m,j} \cdot \left( p_{m,j} \cdot kn^{\frac{1}{k}} - p_{m,j} \cdot \gamma_k \right) \right) \explain{By inequality \eqref{eq:unified_bound_q_m_j}} \\
			& =  \sum_{m=0}^n \delta_m \cdot p_m \cdot \left(kn^{\frac{1}{k}} - \gamma_k\right) \explain{Since $p_m  = \sum_{j=0}^{m+1} w_{m,j} \cdot p_{m,j}$ by \eqref{eq:r_k_n_p_m_as_function_of_r_k_j_n_m_p}}\\
			& = p \cdot \left(kn^{\frac{1}{k}} - \gamma_k\right)\,. \explain{Since $p = \sum_{m=0}^n \delta_m \cdot p_m$ by \eqref{eq:ordered_p_sum_of_p_m}}
		\end{align}
		This completes the induction step and the proof.
	\end{proof}
	
	We  consider separately the case of $k=1$ rounds, giving a lower bound that applies to both ordered and unordered search.
	
	\begin{proposition}\label{prop:lb_one_round}  Let $p \in (0,1]$ and $n \in \mathbb{N}_{\geq 1}$. Then 
		\begin{align}
			\mathcal{R}_{1-p}(\mbox{\em unordered}_{n,1}) \geq  np  \; \; \mbox{and} \; \; \mathcal{R}_{1-p}(\mbox{\em ordered}_{n,1}) \geq   np  \,.
		\end{align} 
	\end{proposition}
	\begin{proof}
		We show the lower bound for randomized algorithms when facing the uniform distribution. For one round, there is no distinction between ordered and unordered search. By an average argument, the lower bound obtained applies to a worst case input.
		
		Let $\mathcal{A}_1$ be a randomized algorithm that runs in one round and  succeeds with probability $p$ when given an input drawn from the uniform distribution. Let $q_1(n,p)$ be the expected number of queries of $\mathcal{A}_1$ as a function of the input size $n$ and the success probability $p$.  Denote by $\delta_m$ the probability that the algorithm issues $m$ queries in round $1$. Since there is no second round and the input distribution is uniform, the location of these queries does not matter. 
		
		The expected number of queries issued by the algorithm on the uniform input distribution can be written as 
		\[
		q_1(n,p) = \sum_{m=0}^n \delta_m \cdot m,
		\]
		where 
		\begin{align}
			& \sum_{m=0}^n \delta_m = 1 \\
			& p = \sum_{m=0}^n \delta_m \cdot \left( \frac{m}{n}\right) \\
			& \delta_m \geq 0 \; \; \forall m \in \{0, \ldots, n\}\,.
		\end{align}
		Thus  we have $q_1(n,p) = \sum_{m=0}^n \delta_m \cdot m = np$.
	\end{proof}

	\subsection{Lemmas for ordered search proofs}	In this section we include the lemmas  used to prove the ordered search upper and lower bounds. 
	
	\begin{lemma} \label{lem:helper_lemma_main_2}
		Let $k \geq 2, n\geq 1$,  and the sequence $\{\gamma_{\ell}\}_{\ell= 1}^{\infty}$ with  $\gamma_1 = 0$ and   $\gamma_{\ell} = 2 \ell$ for all $\ell \geq 2$.
		Then  
		\begin{align} \label{eq:required_inequality_ckn_non-boundary_j_zero}
			x\left(1 + \frac{\gamma_{k-1}}{n}\right) + ({k-1}) \cdot \frac{(n  - x)^{\frac{k}{k-1}}}{n \cdot (x+1)^{\frac{1}{k-1}}} - \gamma_{k-1}  \geq   k n^{\frac{1}{k}} - \gamma_k \; \; \forall x \in (1/2, n]\,. 
		\end{align}
	\end{lemma}
	\begin{proof} 
		Let $t = \left( \frac{n-x}{x+1}\right)^{\frac{1}{k-1}}$. Then $t$ is decreasing in $x$. Since $x \in (1/2, n]$, we have   $0 \leq t < \left( \frac{2n-1}{3}\right)^{\frac{1}{k-1}}$. Expressing $x$ in terms of $t$ we get
		\begin{align} x = \frac{n - t^{k-1}}{t^{k-1}+1}\,. \label{eq:case_j_0_x_in_terms_of_t_ordered}
		\end{align}
		Substituting \eqref{eq:case_j_0_x_in_terms_of_t_ordered} in  \eqref{eq:required_inequality_ckn_non-boundary_j_zero}, we get that \eqref{eq:required_inequality_ckn_non-boundary_j_zero} is equivalent to  
		\begin{align} \label{eq:required_inequality_j_zero_rewritten_in_terms_of_t_ordered}
			&   t^{k} \cdot (k-1)(n+1) - t^{k-1} \cdot \left( kn^{\frac{k+1}{k}} + n + \gamma_{k-1} (n+1) -n \gamma_k   \right) + \left(n^2 +  n \gamma_k  - kn^{\frac{k+1}{k}} \right) \geq 0  \notag \\
			&   \qquad \qquad \qquad \qquad \qquad \qquad \qquad \qquad \qquad \qquad \qquad \qquad \qquad \qquad  \; \; \;   \forall \; 0 \leq  t < \left( \frac{2n-1}{3}\right)^{\frac{1}{k-1}}  \,.
		\end{align}
		
		We consider two cases, for $k=2$ and $k \geq 3$.
		
		\paragraph{Case  $k=2$.}	  Since $\gamma_1=0$ and $\gamma_2 = 4$, inequality \eqref{eq:required_inequality_j_zero_rewritten_in_terms_of_t_ordered}  is equivalent to  
		\begin{align}
			& t^2 \cdot (n+1) - t \cdot \left(2n \sqrt{n} - 3n \right)  + n^2 -2n \sqrt{n} + 4n \geq 0 \qquad \forall \;  0 \leq t  <  \frac{2n-1}{3}  \label{eq:required_inequality_boundary_j_zero_k_2} \,.
		\end{align}
		Inequality~\eqref{eq:required_inequality_boundary_j_zero_k_2} holds by Lemma~\ref{lem:helper_lemma_case1b_2_rounds}.

		\paragraph{Case  $k \geq 3$.}	 
		Since  $\gamma_1 = 0$ and $\gamma_{\ell} = 2 \ell$ for $\ell \geq 2$, inequality \eqref{eq:required_inequality_j_zero_rewritten_in_terms_of_t_ordered} can be simplified to  
		\begin{align}
			& t^k \cdot (k-1)(n+1) - t^{k-1} \cdot \left( k \cdot n^{\frac{k+1}{k}} - n + 2k -2\right) + n^2 + 2kn - k n^{1 + \frac{1}{k}} \geq 0 \notag \\
			& \qquad  \qquad   \qquad \qquad  \qquad   \qquad  \qquad  \qquad   \qquad 
			\qquad  \qquad   \qquad \forall \; 0 \leq t <  \left(\frac{2n-1}{3}\right)^{\frac{1}{k-1}} \,. \label{eq:required_inequality_intermediate_case_1b_k_at_least_3}
		\end{align}
		
		Inequality \eqref{eq:required_inequality_intermediate_case_1b_k_at_least_3} holds by   Lemma~\ref{lem:helper_lemma_case1b_k_more_than_3_rounds} for all $k\geq 3$. 
		This completes the proof. 
	\end{proof}
	
	\begin{lemma} \label{lem:helper_lemma_case1b_2_rounds}
		Let $n \geq 1$. Then for all $t \in \bigl[0, (2n-1)/3 \bigr)$, we have 
		\begin{align}
			& t^2 \cdot (n+1) - t \cdot \left(2n \sqrt{n} - 3n \right)  + n^2 -2n \sqrt{n} + 4n  \geq 0 \,. \label{eq:target_inequality_helper} 
		\end{align}
	\end{lemma}
	\begin{proof}
		Let $f : \mathbb{R} \to \mathbb{R}$ be 
		$	f(t) = t^2 \cdot (n+1) - t \cdot \left(2n \sqrt{n} - 3n \right)  + n^2 -2n \sqrt{n} + 4n  \,.
		$
		Then 
		\begin{align}
			f'(t) &= 2t (n+1) - (2n\sqrt{n} - 3n) \; \; \; \mbox{and} \;\;\; 
			f''(t) = 2(n+1)\,.
		\end{align}
		Thus $f$ is convex and the global minimum is at $t^*$ for which $f'(t^*)=0$, that is, $t^* = \frac{2n\sqrt{n} - 3n}{2n+2}$.
		
		Evaluating $f(t^*)$ gives
		\begin{align}
			f(t^*) & = \left( \frac{2n\sqrt{n} - 3n}{2n+2} \right)^2 \cdot (n+1) - \left(\frac{2n\sqrt{n} - 3n}{2n+2} \right) \cdot \left(2n \sqrt{n} - 3n \right)  + n^2 -2n \sqrt{n} + 4n \notag \\
			& =\frac{ 11n^2  - 8n\sqrt{n} + 16n   + 4n^2 \sqrt{n}}{4n+4}  \notag \\
			& > 0\,. \explain{Since $11n^2 > 8n \sqrt{n}$ for $n \geq 1$}
		\end{align}
		Thus $f(t) \geq f(t^*) > 0$ for all $t \in \mathbb{R}$, which implies the inequality required by the lemma.
	\end{proof}
	
	\begin{lemma} \label{lem:helper_lemma_case1b_k_more_than_3_rounds} Let $n \geq 1$ and $k \geq 3$. Then 
		\begin{align} 
			t^k \cdot (k-1)(n+1) - t^{k-1} \cdot \left( k \cdot n^{\frac{k+1}{k}} - n + 2k -2\right) + n^2 + 2kn - k n^{1 + \frac{1}{k}} \geq 0, \; \; \forall  t \in \left[0, n^{\frac{1}{k-1}}\right) \,. \label{eq:required_inequality_case_1b_k_more_than_3}
		\end{align}
	\end{lemma}
	\begin{proof}		
		Dividing both sides of \eqref{eq:required_inequality_case_1b_k_more_than_3} by $n^2$, we get that \eqref{eq:required_inequality_case_1b_k_more_than_3} holds if and only if   
		\begin{align}
			& \left(\frac{t}{n^{\frac{1}{k-1}}}\right)^{k} 
			\cdot (k-1) \left(1 + \frac{1}{n}\right) \cdot n^{\frac{1}{k-1}} - \left(\frac{t}{n^{\frac{1}{k-1}}}\right)^{k-1} \cdot \left( k \cdot n^{\frac{1}{k}} - 1 + \frac{2k -2}{n} \right) + 1 + \frac{2k}{n} - \frac{k}{n^{\frac{k-1}{k}}}   \geq 0   \notag \\ 
			& \qquad \qquad \qquad \qquad \qquad \qquad  \qquad \qquad \qquad \qquad \qquad \qquad \qquad \qquad \qquad  \qquad 
			\forall  t \in \left[0, n^{\frac{1}{k-1}}\right)\,. \label{eq:required_inequality_case_1b_k_more_than_3_intermediate_1}
		\end{align}
		If $t=0$ then \eqref{eq:required_inequality_case_1b_k_more_than_3_intermediate_1}  is equivalent to  $n + 2k - k n^{\frac{1}{k}} \geq 0$, which  holds by Corollary~\ref{corr:helper_lemma_1}. 
		
		For $t > 0$, let $x = {n^{\frac{1}{k-1}}}/{t}$. Since $0 < t < n^{\frac{1}{k-1}}$, we have $x > 0$. Substituting $t$ by $x$ we obtain that \eqref{eq:required_inequality_case_1b_k_more_than_3_intermediate_1} is equivalent to 
		%
		\begin{align}
			&   x^k \cdot \left(1 + \frac{2k}{n} - \frac{k}{n^{\frac{k-1}{k}}} \right) - x \cdot  \left(k \cdot n^{\frac{1}{k}} - 1 + \frac{2k -2}{n}\right) + (k-1)\left(1+\frac{1}{n}\right) \cdot n^{\frac{1}{k-1}} \geq 0, \;\; \forall x > 1 \,. \label{eq:def_f_equation_case_1b_k_3}
		\end{align}
		Define the function $f : (0, \infty) \to \mathbb{R}$, where $f(x)$ is given by the left hand side of \eqref{eq:def_f_equation_case_1b_k_3}. Then 
		\begin{align}
			f'(x) & = k \left(1 + \frac{2k}{n} - \frac{k}{n^{\frac{k-1}{k}}} \right) x^{k-1} - \left(k \cdot n^{\frac{1}{k}} - 1 + \frac{2k -2}{n}\right) \notag \\
			f''(x) &= k(k-1)\left(1 + \frac{2k}{n} - \frac{k}{n^{\frac{k-1}{k}}} \right) x^{k-2}\,.
		\end{align}
		By  Corollary~\ref{corr:helper_lemma_1}, we have $1 + {2k}/{n} - {k}/{n^{\frac{k-1}{k}}}  > 0$ for all $n \geq 1$ and $k \geq 3$. Thus 
		$f''(x) > 0$ for all $x > 0$, so  $f$ is convex on  $(0, \infty)$. Observe that $k \cdot n^{\frac{1}{k}} - 1 + \frac{2k -2}{n} > 0$ for all $n \geq 1$ and $k \geq 3$.
		Then there is a  global minimum of $f$ at a point $\overline{x} \in (0, \infty)$ with $ f'(\overline{x}) = 0 $, or equivalently,
		\begin{align}
			&   \overline{x} = \left( \frac{ k \cdot n^{\frac{1}{k}} - 1 + \frac{2k -2}{n}}{  k \left(1 + \frac{2k}{n} - \frac{k}{n^{\frac{k-1}{k}}} \right)} \right)^{\frac{1}{k-1}} \,.
		\end{align}
		Evaluating $f$ at $\overline{x}$ and rearranging terms gives
		\begin{align}
			f(\overline{x}) & = \overline{x} \left[  \left(1 + \frac{2k}{n} -  \frac{k}{n^{\frac{k-1}{k}}} \right) \overline{x}^{k-1} - \left(k \cdot n^{\frac{1}{k}} - 1 + \frac{2k -2}{n}\right)   \right] +  (k-1)\left(1+\frac{1}{n}\right) \cdot n^{\frac{1}{k-1}}  \notag \\
			& = \overline{x} \left(\frac{1}{k}-1\right) \left(k \cdot n^{\frac{1}{k}} - 1 + \frac{2k -2}{n}\right)  + (k-1)\left(1+\frac{1}{n}\right) \cdot n^{\frac{1}{k-1}} \notag \\
			& =  (k-1)\left(1+\frac{1}{n}\right) \cdot n^{\frac{1}{k-1}}  - (k-1) \left({  n^{\frac{1}{k}} -\frac{1}{k} + \frac{2}{n} - \frac{2}{kn} }\right)^{\frac{k}{k-1}} \left(\frac{n}{n+2k-{k}{n^{\frac{1}{k}}}}\right)^{\frac{1}{k-1}} \,.
		\end{align}
		Thus  $ f(x) > 0 \; \; \forall x > 1$ whenever the next two properties are met
		\begin{enumerate}
			\item $f(\overline{x}) > 0$ when $\overline{x} > 1$. This is equivalent to 
			
			\begin{align}
				(n+2k-{k}n^{\frac{1}{k}})\left(1 + \frac{1}{n}\right)^{{k-1}} & >   \left({  n^{\frac{1}{k}} -\frac{1}{k} + \frac{2}{n} - \frac{2}{kn} }\right)^{k} \label{eq:first_property_under_next_condition_helper_lemma_ordered} \\ 
				& \mbox{ whenever }    n^{\frac{1}{k}} - \frac{1}{k} + \frac{2}{n} - \frac{2}{kn}   >  1 + \frac{2k}{n}  - \frac{k}{n^{\frac{k-1}{k}}} \,.\label{eq:first_property_helper_lemma_case1b_k_more_than_3_rounds}
			\end{align}
			Lemma~\ref{lem:helper_lemma_case1b_k_more_than_3_rounds_II} implies that   inequality   \eqref{eq:first_property_under_next_condition_helper_lemma_ordered} holds under condition \eqref{eq:first_property_helper_lemma_case1b_k_more_than_3_rounds}.
			\item $f(1) \geq  0 $  when $\overline{x} < 1$. To show this, observe that for all $n \geq 1$ and $k \geq 3$, we have  
			\begin{align}
				f(1) & = \left(1 + \frac{1}{n}\right) \left( 2 - k \cdot n^{\frac{1}{k}} + (k-1)  \cdot n^{\frac{1}{k-1}} \right) \notag \\
				& \geq 0 \,. \explain{By Lemma~\ref{lem:case_1b_case_2_f_1_positive}}
			\end{align}
			Thus $f(1) \geq 0$ for all $n \geq 1$ and $k \geq 3$, which completes property 2. 
		\end{enumerate} 
		Since both properties $1$ and $2$ hold, we have that  $f(x) > 0$ for all $x > 1$,  so \eqref{eq:def_f_equation_case_1b_k_3} holds. Equivalently, \eqref{eq:required_inequality_case_1b_k_more_than_3} holds as required by the lemma.
	\end{proof}

	\begin{lemma} \label{lem:case_1b_case_2_f_1_positive}
		Let $n \geq 1$ and $k \geq 3$, where $k,n \in \mathbb{N}$. Then 
		$
		2  - k \cdot  n^{\frac{1}{k}} + (k-1)\cdot  n^{\frac{1}{k-1}} \geq 0 \,.
		$
	\end{lemma}
	\begin{proof}
		Consider the function $f : [2, \infty)\to \mathbb{R}$ given by $f(x) = x n^{\frac{1}{x}}$. We first show an upper bound on $f'$ and then use it to upper bound $f(k)- f(k-1)$.
		We have 
		\[
		f'(x) = n^{\frac{1}{x}}\left(1  -    \frac{\ln(n)}{x} \right)\,.
		\]
		Let $y = n^{\frac{1}{x}}$. Then $\ln(y) = \frac{1}{x} \ln(n)$. Since $x \geq 2$, we have $y \in (1, \sqrt{n}]$. Then 
		\begin{align}
			n^{\frac{1}{x}}\left(1  -    \frac{\ln(n)}{x} \right) = y \left( 1 - \ln(y)\right) \,. \label{eq:identity_derivative_x_g_y}
		\end{align}
		The function $g : (1, \infty) \to \mathbb{R}$ given by $g(y) = y \left(1 - \ln(y)\right)$ has $g'(y) = -\ln(y) < 0$. Therefore $g(y) < g(1) = 1$ for all $y > 1$.  Using the identity in \eqref{eq:identity_derivative_x_g_y}, we get that $f'(x) < 1$ for all $x \geq 2$. 
		
		Then  for all $k \geq 3$,
		\begin{align}
			k \cdot  n^{\frac{1}{k}} =	f(k) \leq f(k-1) + \max_{x \in [2, \infty)} f'(x) < f(k-1) + 1 = (k-1)\cdot  n^{\frac{1}{k-1}} + 1\,. \label{eq:required_even_stronger_f_k_minus_f_k_1}
		\end{align}
		Inequality \eqref{eq:required_even_stronger_f_k_minus_f_k_1} implies the lemma statement.
	\end{proof}
	
	\begin{corollary} \label{lem:ordered_lb_case_II_0}
		Let $n \geq 1$ and $k \geq 2$. Suppose $\{\gamma_{\ell}\}_{\ell=1}^{\infty}$ is the sequence given by $\gamma_1=0$ and $\gamma_{\ell} = 2\ell$ for  $\ell \geq 2$.  Then
		$({k-1}) \cdot n^{\frac{1}{k-1}} - {\gamma_{k-1}} 
		\geq   kn^{\frac{1}{k}} -  \gamma_{k} \,.$
	\end{corollary}
	\begin{proof}
		If $k=2$, the required inequality is   
		$
		n \geq 2 \sqrt{n} - 4
		$, or  $\left(\sqrt{n} -1 \right)^2 + 3 \geq 0$. The latter   holds for all $n \geq 1$.
		If $k \geq 3$, the required inequality is  
		$(k-1) \cdot n^{\frac{1}{k-1}} +2 \geq kn^{\frac{1}{k}}$, which 
		holds by Lemma~\ref{lem:case_1b_case_2_f_1_positive}.
	\end{proof}
	
	\begin{lemma} \label{lem:helper_lemma_case1b_k_more_than_3_rounds_II}
		Let $n \geq 1$ and $k \geq 3$, where $k,n \in \mathbb{N}$. Then 
		\begin{align}
			(n+2k-{k}n^{\frac{1}{k}})\left(1 + \frac{1}{n}\right)^{{k-1}} & >   \left({  n^{\frac{1}{k}} -\frac{1}{k} + \frac{2}{n} - \frac{2}{kn} }\right)^{k} \label{eq:required_inequality_helper_lemma_case1b_k_more_than_3_rounds_restated} \\ 
			& \mbox{ whenever }   \;  n^{\frac{1}{k}} - \frac{1}{k} + \frac{2}{n} - \frac{2}{kn}   >  1 + \frac{2k}{n}  - \frac{k}{n^{\frac{k-1}{k}}} \,. \label{eq:first_property_helper_lemma_case1b_k_more_than_3_rounds_restated}
		\end{align} 
	\end{lemma}
	\begin{proof}		
		If $n=1$ then the condition   in  \eqref{eq:first_property_helper_lemma_case1b_k_more_than_3_rounds_restated}  is equivalent to $1 - 1/k + 2 -2/k > 1 + 2k - k$, which holds if and only if $2 - 3/k > k$ $(\dag)$. Since $k\geq 3$,  inequality $(\dag)$ does not hold so condition \eqref{eq:first_property_helper_lemma_case1b_k_more_than_3_rounds_restated}  is not met either. 
		
		Thus from now on we can assume $n \geq 2$.  By Lemma~\ref{lem:ruling_out_range_k_at_least_n_ordered}, condition \eqref{eq:first_property_helper_lemma_case1b_k_more_than_3_rounds_restated}  implies  $k < n$.
		We show  \eqref{eq:required_inequality_helper_lemma_case1b_k_more_than_3_rounds_restated} holds when  $n \geq 2$ and $k < n$ by considering  separately a few ranges of $k$. 	
		
		\paragraph{Case I: $n/2 <  k < n$ and $k \geq 3$.} Then $k < n < 2k$. 
		When  $n=2k-1$ inequality \eqref{eq:required_inequality_helper_lemma_case1b_k_more_than_3_rounds_restated} holds by Lemma~\ref{lem:case_II_ordered_search_analysis_inequality_for_n_equal_to_2k_minus_1}.
		
		Thus from now on we can assume $k < n \leq 2k-2$. 
		To show inequality  \eqref{eq:required_inequality_helper_lemma_case1b_k_more_than_3_rounds_restated},  we will first bound separately several of the terms in the inequality and then combine the bounds. 
		
		For  $k \geq 3$, we have $k \leq 2^{k-1}$. Moreover, since $ n <  2k$, we have 
		$n^{\frac{1}{k}} < (2k)^{\frac{1}{k}} \leq 2$, and so $2k > kn^{\frac{1}{k}}$. Thus  
		$n+2k - k n^{\frac{1}{k}} > n$, which implies  
		\begin{align}
			& \left(n+2k - k n^{\frac{1}{k}} \right) \left(1 + \frac{1}{n}\right)^{k-1} > n \left(1 + \frac{1}{n}\right)^{k-1}  \,.  \label{eq:n_plus_2k_minus_k_times_n_to_power_1_over_k_at_least_n_times_extra_terms_on_both_sides_ordered}
		\end{align}
		Moreover, since $n \geq 2$, we have $2k \leq k n \cdot n^{\frac{1}{k}}$, and so $2k - 2 - n <  k n \cdot n^{\frac{1}{k}}$. Since $n \leq 2k-2$, we also have $2k-2-n \geq 0$, and so   
		\begin{align}
			0 \leq   \frac{2k - 2 - n}{k n \cdot n^{\frac{1}{k}}} < 1\,. \label{eq:2k_minus_2_minus_n_over_k_times_n_times_n_to_power_1_over_k_ordered}
		\end{align}
		Let $r = ({2k-2-n})/({kn \cdot n^{\frac{1}{k}}})$. Inequality    \eqref{eq:2k_minus_2_minus_n_over_k_times_n_times_n_to_power_1_over_k_ordered} gives $0 \leq r < 1$. We consider two sub-cases:
		\begin{itemize}
			\item If  $n=2k-2$ then $r=0$. We have  
			\begin{align} 
				\left(1 + \frac{2k-2-n}{kn\cdot n^{\frac{1}{k}}}\right)^k = (1+r)^k = 1 = e^{0} =  e^{\left(\frac{2k-2-n}{n \cdot n^{\frac{1}{k}}} \right)}\,. \label{eq:part_1_upper_bound_e_to_involved_exponent_ordered}
			\end{align}
			\item Else  $k < n < 2k-2$. Then $0 < r < 1$. We have  
			\begin{align}
				\left(1 + \frac{2k-2-n}{kn \cdot n^{\frac{1}{k}}}\right)^k & = (1+r)^k \explain{By definition of $r$} \\
				& = \left[\left(1 + r\right)^{\frac{1}{r}} \right]^{kr} \notag \\
				& \leq e^{kr} \explain{Since $(1+r)^{\frac{1}{r}} \leq e$ for  $r \in (0,1)$.} \\ 
				& = e^{\left( \frac{2k-2-n}{n \cdot n^{\frac{1}{k}}} \right)}  \,. \label{eq:part_2_upper_bound_e_to_involved_exponent_ordered}
			\end{align}
		\end{itemize}
		Combining  inequalities \eqref{eq:part_1_upper_bound_e_to_involved_exponent_ordered} and \eqref{eq:part_2_upper_bound_e_to_involved_exponent_ordered} from the two sub-cases, we obtain \begin{align}
			\left(1 + \frac{2k-2-n}{kn \cdot n^{\frac{1}{k}}}\right)^k \leq e^{\left( \frac{2k-2-n}{n \cdot n^{\frac{1}{k}}} \right)} \qquad \forall n \in \mathbb{N} \; \mbox{with} \; k < n \leq 2k-2\,. \label{eq:both_parts_upper_bound_e_to_involved_exponent_ordered}
		\end{align}
		
		Using \eqref{eq:both_parts_upper_bound_e_to_involved_exponent_ordered}, we can upper bound the right hand side of  inequality  \eqref{eq:required_inequality_helper_lemma_case1b_k_more_than_3_rounds_restated} as follows:  
		\begin{align}
			\left(n^{\frac{1}{k}} - \frac{1}{k} + \frac{2}{n} - \frac{2}{kn}\right)^k & = \left(n^{\frac{1}{k}} + \frac{2k-2-n}{kn} \right)^k \notag \\
			& = n \left(1 + \frac{2k-2-n}{kn \cdot n^{\frac{1}{k}}}\right)^k \notag \\
			& \leq n  e^{\left( \frac{2k-2-n}{n \cdot n^{\frac{1}{k}}} \right)}\,. \label{eq:RHS_of_target_at_most_e_to_power_big_fraction_ordered}
		\end{align}
		
		By Lemma~\ref{lem:1_plus_1_over_n_all_to_k_minus_1_at_least_exp_of_big_fraction_with_2k_minus_2_minus_n_ordered}, we have 
		\begin{align}  e^{\left( \frac{2k-2-n}{n \cdot n^{\frac{1}{k}}} \right)} \leq  \left(1+\frac{1}{n}\right)^{k-1}  \,. \label{eq:1_plus_1_over_n_to_power_k_minus_1_less_than_e_to_big_fraction_ordered}
		\end{align} 
		Combining \eqref{eq:n_plus_2k_minus_k_times_n_to_power_1_over_k_at_least_n_times_extra_terms_on_both_sides_ordered}, \eqref{eq:RHS_of_target_at_most_e_to_power_big_fraction_ordered}, and \eqref{eq:1_plus_1_over_n_to_power_k_minus_1_less_than_e_to_big_fraction_ordered},  gives:
		\begin{align}
			\left(n^{\frac{1}{k}} - \frac{1}{k} + \frac{2}{n} - \frac{2}{kn}\right)^k & \leq n  e^{\left( \frac{2k-2-n}{n \cdot n^{\frac{1}{k}}} \right)} \explain{By \eqref{eq:RHS_of_target_at_most_e_to_power_big_fraction_ordered}} \\
			& \leq n \left(1 + \frac{1}{n}\right)^{k-1} \explain{By \eqref{eq:1_plus_1_over_n_to_power_k_minus_1_less_than_e_to_big_fraction_ordered}} \\
			& < \left(n+2k - k n^{\frac{1}{k}} \right) \left(1 + \frac{1}{n}\right)^{k-1}   \,.  \explain{By \eqref{eq:n_plus_2k_minus_k_times_n_to_power_1_over_k_at_least_n_times_extra_terms_on_both_sides_ordered}}
		\end{align}
		
		In summary,   
		\[ 
		\left(n^{\frac{1}{k}} - \frac{1}{k} + \frac{2}{n} - \frac{2}{kn}\right)^k \leq \left(n+2k - k n^{\frac{1}{k}} \right) \left(1 + \frac{1}{n}\right)^{k-1} \qquad \forall n \in \mathbb{N} \; \mbox{with} \; k < n \leq 2k-2\,.
		\]
		This is the required inequality \eqref{eq:required_inequality_helper_lemma_case1b_k_more_than_3_rounds_restated}, which completes case I.
		
		\paragraph{Case II: $3 \leq k \leq n/2$.} Then $n \geq 2k$ and $k \geq 3$. Then 
		$t^k \ge 2k$.
		For $k=3$, the 
		required inequality \eqref{eq:required_inequality_helper_lemma_case1b_k_more_than_3_rounds_restated} is equivalent to   
		\[
		\left(x + 6 - 3 x^{\frac{1}{3}}\right) \left(1 + \frac{1}{x} \right)^{2} - \left( x^{\frac{1}{3}} - \frac{1}{3} + \frac{2}{x} - \frac{2}{3x}\right)^3 > 0 \qquad \forall x \geq 6^{\frac{1}{3}},
		\]
		which can be easily checked to hold  (see, e.g., \cite{wolfram}). 
		
		Thus from now on we can assume $k \ge 4$ with $k \in \mathbb{N}$.
		Let $f : (0, \infty) \to \mathbb{R}$ be 
		\[ f(x) = \left(x+2k-{k}x^{\frac{1}{k}} \right)\left(1 + \frac{1}{x}\right)^{{k-1}} -   \left({  x^{\frac{1}{k}} -\frac{1}{k} + \frac{2}{x} - \frac{2}{kx} }\right)^{k} \,.
		\] 
		
		Using Bernoulli's inequality gives 
		\begin{align}
			\left(1 + \frac{1}{n}\right)^{k-1} \geq 1 + \frac{k-1}{n}\,.  \label{eq:bernoulli_applied_to_1_plus_1_over_x_all_to_k_minus_1_ordered}
		\end{align}
		Since $n \geq 2k$, we have ${2}/{n} \leq 1/k$. Thus 
		\begin{align}
			-\frac{1}{k}+ \frac{2k-2}{kn} \leq  -\frac{1}{k}+ \frac{k-1}{k^2} = - \frac{1}{k^2} \,. \label{eq:2k_minus_2_all_over_kn_less_than_k_minus_1_over_k_squared_ordered}
		\end{align}
		Using \eqref{eq:bernoulli_applied_to_1_plus_1_over_x_all_to_k_minus_1_ordered} and \eqref{eq:2k_minus_2_all_over_kn_less_than_k_minus_1_over_k_squared_ordered}, we can lower bound $f(n)$ as follows: 
		\begin{align}
			f(n) & = \left(n+2k-{k}n^{\frac{1}{k}} \right)\left(1 + \frac{1}{n}\right)^{{k-1}} -   \left({  n^{\frac{1}{k}} -\frac{1}{k} + \frac{2k-2}{kn}  }\right)^{k}  \notag \\
			& \geq \left(n+2k-{k}n^{\frac{1}{k}} \right)\left(1 + \frac{k-1}{n}\right) - \left({  n^{\frac{1}{k}} - \frac{1}{k^2}  }\right)^{k} \explain{By \eqref{eq:bernoulli_applied_to_1_plus_1_over_x_all_to_k_minus_1_ordered} and \eqref{eq:2k_minus_2_all_over_kn_less_than_k_minus_1_over_k_squared_ordered}} \\
			& \geq  \left(2k-{k}n^{\frac{1}{k}} \right)\left(1 + \frac{k-1}{n}\right) +n - \left({  n^{\frac{1}{k}} - \frac{1}{k^2}  }\right)^{k} \,. \label{eq:f_n_lb_2k_minus_k_times_n_to_power_1_over_k_times_bernoulli_result_plus_n_minus_smaller_than_n_ordered}
		\end{align}
		
		If $n^{\frac{1}{k}} \leq 2$, then $2k - kn^{\frac{1}{k}} \geq 0$, which together with \eqref{eq:f_n_lb_2k_minus_k_times_n_to_power_1_over_k_times_bernoulli_result_plus_n_minus_smaller_than_n_ordered} yields $
		f(n) \geq n - \left({  n^{\frac{1}{k}} - \frac{1}{k^2}  }\right)^{k} \geq 0$, as required.
		
		Thus from now on we will assume  $n^{\frac{1}{k}} > 2$, that is, $n > 2^k$. Then  $2k - kn^{\frac{1}{k}} < 0$. Together with $n \geq 2k$, this implies  
		\begin{align}
			\left(2k-{k}n^{\frac{1}{k}} \right) \left(1 + \frac{k-1}{n} \right) > \left(2k-{k}n^{\frac{1}{k}} \right) \left(1 + \frac{k-1}{2k} \right)\,. \label{eq:first_part_of_f_simplified_bounded_by_multiple_of_1_plus_k_minus_1_over_2k_ordered}
		\end{align}
		Inequality \eqref{eq:first_part_of_f_simplified_bounded_by_multiple_of_1_plus_k_minus_1_over_2k_ordered} together with $(k-1)/2k < 1/2$ yields 
		\begin{align} 
			\left(2k-{k}n^{\frac{1}{k}} \right) \left(1 + \frac{k-1}{n} \right) > 1.5 \cdot \left(2k-{k}n^{\frac{1}{k}} \right)\,. \label{eq:first_part_of_f_simplified_bound_at_least_1_point_5_times_2k_minus_k_times_n_to_power_1_over_k_ordered}
		\end{align}
		Combining  
		\eqref{eq:f_n_lb_2k_minus_k_times_n_to_power_1_over_k_times_bernoulli_result_plus_n_minus_smaller_than_n_ordered}  and \eqref{eq:first_part_of_f_simplified_bound_at_least_1_point_5_times_2k_minus_k_times_n_to_power_1_over_k_ordered}
		gives 
		\begin{align}
			f(n) \geq 1.5 \cdot \left(2k-{k}n^{\frac{1}{k}}  \right) +n - \left({  n^{\frac{1}{k}} - \frac{1}{k^2}  }\right)^{k} \,. \label{eq:f_n_at_least_1_point_5_times_2k_minus_k_times_n_to_1_over_k_plus_n_minus_n_to_1_over_k_minus_small_term_all_to_k_ordered}
		\end{align}
		
		Next we expand  and truncate $\bigl(n^{\frac{1}{k}} - \frac{1}{k^2} \bigr)^k$ via Lemma~\ref{lem:expanding}, yielding
		\begin{align}
			-\left(n^{\frac{1}{k}} - \frac{1}{k^2}\right)^k \ge -n + \frac{n^{1-\frac{1}{k}}}{k} - \frac{(k-1)n^{1-\frac{2}{k}}}{2k^3} \,.  \label{eq:expanding_and_truncating_ordered}
		\end{align}
		Using \eqref{eq:expanding_and_truncating_ordered}, we can further bound $f(n)$ by  
		\begin{align}
			f(n) & \ge 1.5 \cdot \left(2k-{k}n^{\frac{1}{k}}  \right) + \frac{n^{1-\frac{1}{k}}}{k} - \frac{(k-1)n^{1-\frac{2}{k}}}{2k^3} \explain{Combining \eqref{eq:f_n_at_least_1_point_5_times_2k_minus_k_times_n_to_1_over_k_plus_n_minus_n_to_1_over_k_minus_small_term_all_to_k_ordered} and \eqref{eq:expanding_and_truncating_ordered}} \\
			& \geq 0  \,. \explain{By Lemma~\ref{lem:bound_polynomial_after_truncation_ordered}}
		\end{align}
		Thus $f(n) > 0$, as required. This completes the analysis for the range $n > 2^k$ and case II.

		\paragraph{Wrapping up.} We obtain that inequality  \eqref{eq:required_inequality_helper_lemma_case1b_k_more_than_3_rounds_restated} holds under condition \eqref{eq:first_property_helper_lemma_case1b_k_more_than_3_rounds_restated}, as required.
	\end{proof}

	\begin{lemma} \label{lem:expanding}
		Let  $k, n \in \mathbb{N}$ with $n \ge 1$ and $k \ge 3$.
		Then
		\begin{align}
			\left(n^{\frac{1}{k}} - \frac{1}{k^2} \right)^k \le n - \frac{n^{1 - \frac{1}{k}}}{k} + \frac{(k-1) \cdot  n^{1 - \frac{2}{k}}}{2k^3}\,. \label{eq:expanding_target_inequality_ordered}
		\end{align}
	\end{lemma}
	\begin{proof}
		Let $t = n^{\frac{1}{k}}$. Then $t \geq 1$. The required inequality \eqref{eq:expanding_target_inequality_ordered} is equivalent to  
		\begin{align}
			\left(t - \frac{1}{k^2} \right)^k \le t^k - \frac{t^{k-1}}{k} + \frac{(k-1) \cdot t^{k-2}}{2k^3}\,. \label{eq:expanding_target_inequality_ordered_in_terms_of_t}
		\end{align}
		
		For $i \in [k+1]$ let $c_i$ be the $i$-th term in the binomial expansion of $(t - 1/k^2)^k$. In particular, 
		\begin{align}
			c_1 = t^k ; \qquad c_2 = - \frac{t^{k-1}}{k} ; \qquad  c_3 = \frac{(k-1)t^{k-2}}{2k^3} ; \qquad c_{k+1} = (-1)^k \frac{1}{k^{2k}} \,.
		\end{align}
		
		Let us bound the ratio $|c_i / c_{i+1}|$ for $i \in [k]$:
		\begin{align}
			\left| \frac{c_i}{c_{i+1}} \right| = 
			\frac{k^2 \cdot t  i}{k-i+1} \geq  tk > 1 \,. \label{eq:ratio_c_i_over_c_i_plus_1_in_absolute_value_ordered}
		\end{align}
		Since $c_{2i} < 0$ and $c_{2i+1} > 0$ for all $i$, inequality \eqref{eq:ratio_c_i_over_c_i_plus_1_in_absolute_value_ordered} implies  
		\begin{align}
			c_i + c_{i+1} \le 0 \qquad \forall i \in [k] \; \mbox{with} \; i \in 2 \mathbb{N} \,. \label{eq:ci_plus_ci_plus_1_negative}
		\end{align}
		We bound the term $\left(t- \frac{1}{k^2}\right)^k$ by considering two cases.
		If $k$ is even, then
		\begin{align}
			\left(t- \frac{1}{k^2} \right)^k &=  c_1 + c_2 + c_3 + \sum_{i=2}^{k/2} \left( c_{2i} + c_{2i+1} \right) \explain{By definition of $c_i$.} \\
			&\le c_1 + c_2 + c_3 \,. \explain{Since $c_{2i} + c_{2i+1} < 0$ by \eqref{eq:ci_plus_ci_plus_1_negative}} 
		\end{align}
		If  $k$ is odd, then 
		\begin{align}
			\left(t- \frac{1}{k^2} \right)^k &= \sum_{i=1}^{k+1} c_i \explain{By definition of $c_i$.}\\
			&< \sum_{i=1}^k c_i 
			= c_1 + c_2 + c_3 + \sum_{i=2}^{(k-1)/2} \left( c_{2i} + c_{2i+1} \right) \explain{Since $c_{k+1} < 0$.}\\
			&\le c_1 + c_2 + c_3 \,. \explain{Since $c_{2i} + c_{2i+1} < 0$ by \eqref{eq:ci_plus_ci_plus_1_negative}}
		\end{align}
		
		Thus for both odd and even $k$, we have 
		\begin{align}
			\left(t- \frac{1}{k^2} \right)^k \leq c_1 + c_2 + c_3 =   t^k - \frac{t^{k-1}}{k} + \frac{(k-1)t^{k-2}}{2k^3}\,.
		\end{align}
		Thus in both cases \eqref{eq:expanding_target_inequality_ordered_in_terms_of_t} holds, as required.
	\end{proof}
	
	\begin{lemma} \label{lem:bound_polynomial_after_truncation_ordered}
		Let $k,n \in \mathbb{N}$ with  $n \geq 2$ and $k \geq 4$. Then 
		\begin{align}
			1.5  \left(2k-{k}n^{\frac{1}{k}}  \right) + \frac{n^{1-\frac{1}{k}}}{k} - \frac{(k-1) \cdot n^{1-\frac{2}{k}}}{2k^3} \geq 0 \,. \label{eq:target_bound_polynomial_after_truncation_ordered}
		\end{align}
	\end{lemma}
	\begin{proof}
		Let $t = n^{\frac{1}{k}}$. Then $t > 1 $ since $n \geq 2$. 
		For $k=4$, inequality \eqref{eq:target_bound_polynomial_after_truncation_ordered} with $n$  substituted by $t^4$ is equivalent to 
		$1.5(8 - 4t) + {t^{3}}/{4} - {3 t^2}/{128} \geq 0, $
		which holds for all $t > 1$.
		
		Thus from now on we can assume $k \geq 5$. The left hand side of \eqref{eq:target_bound_polynomial_after_truncation_ordered}, where  $n$ is substituted by $t^k$, can be bounded as follows:
		\begin{align}
			1.5  \left(2k-{k} t  \right) + \frac{t^{k-1}}{k} - \frac{(k-1) \cdot  t^{k-2} }{2k^3} &\geq  1.5(2k - kt) + \frac{t^{k-2}}{2k^2} \bigl(2kt - 1 \bigr)  \explain{Since $k-1 < k$}  \\
			& \geq  1.5 \left(2k - kt + \frac{t^{k-1}}{2k}\right)\,.  \explain{Since  $ t > 1$ and $k \geq 4$} \\
			\label{eq:target_inequality_at_least_g_t_ordered_1.5_times_stuff}
		\end{align}
		Let $g: (0, \infty) \to \mathbb{R}$ be  $g(t) = 2k - kt + \frac{t^{k-1}}{2k}$.
		We will show that $g(t) \geq 0$ for all $t > 1$. We have 
		\begin{align}
			g'(t) = -k + \frac{k-1}{2k} \cdot t^{k-2} \qquad \mbox{and} \qquad g''(t) = \frac{(k-1)(k-2)}{2k} \cdot t^{k-3}\,.
		\end{align}
		Thus $g$ is convex on $(0, \infty)$. The global minimum is $t^*$ with $g'(t^*) = 0$, so 
		$
		t^* = \left(\frac{2k^2}{k-1}\right)^{\frac{1}{k-2}}\,.
		$
		Then 
		\begin{align}
			g(t) \geq g(t^*) & = 2k - k \cdot t^* + t^* \cdot \frac{(t^*)^{k-2}}{2k} = 2k -  \left(\frac{2k^2}{k-1}\right)^{\frac{1}{k-2}} \left(k - \frac{k}{k-1}\right) \notag \\
			&= k \left( 2 -  \left(\frac{2k^2}{k-1}\right)^{\frac{1}{k-2}}\cdot \left( \frac{k-2}{k-1}\right) \right)\,. \label{eq:g_t_at_least_g_t_star_at_least_k_times_2_minus_product_of_two_terms_ordered}
		\end{align}
		Since $2^{k-2} > \left(\frac{2k^2}{k-1} \right)$ for $k \geq 5$, we get 
		\begin{align}
			2 > \left( \frac{2k^2}{k-1} \right)^{\frac{1}{k-2}} > \left(  \frac{2k^2}{k-1} \right)^{\frac{1}{k-2}} \cdot \left( \frac{k-2}{k-1} \right) \qquad \forall k \geq 5\,. \label{eq:2_greater_than_2_k_squared_over_k_minus_1_all_to_1_over_k_minus_2_times_k_minus_2_over_k_minus_1_ordered}
		\end{align}
		Using \eqref{eq:2_greater_than_2_k_squared_over_k_minus_1_all_to_1_over_k_minus_2_times_k_minus_2_over_k_minus_1_ordered} in \eqref{eq:g_t_at_least_g_t_star_at_least_k_times_2_minus_product_of_two_terms_ordered}, we obtain 
		\begin{align}
			g(t) \geq k \left( 2 -  \left(\frac{2k^2}{k-1}\right)^{\frac{1}{k-2}}\cdot \left( \frac{k-2}{k-1}\right) \right) > 0 \qquad \forall t > 1, k \geq 5\,. \label{eq:g_t_positive_ordered}
		\end{align}
		Combining \eqref{eq:target_inequality_at_least_g_t_ordered_1.5_times_stuff} and  \eqref{eq:g_t_positive_ordered}, we obtain 
		$
		1.5  \left(2k-{k} t  \right) + {t^{k-1}}/{k} - {(k-1) \cdot  t^{k-2} }/{(2k^3)} \geq 0 
		$ for all $t > 1 $ and $k \geq 5$. This completes the proof.
	\end{proof}

	\begin{lemma} \label{lem:case_II_ordered_search_analysis_inequality_for_n_equal_to_2k_minus_1}
		Let $k \in \mathbb{N}$ with $k \geq 3$. Then 
		\begin{align}
			\left(4k-1-{k}\cdot (2k-1)^{\frac{1}{k}}\right)\left(1 + \frac{1}{2k-1}\right)^{{k-1}} >   \left({  (2k-1)^{\frac{1}{k}} -\frac{1}{k \cdot (2k-1)}   }\right)^{k}  \,.
		\end{align}
	\end{lemma}
	\begin{proof}
		Since $k \geq 3$, we have $2k-1 \leq 2^k$, so $(2k-1)^{\frac{1}{k}} \leq 2$. Then 
		\begin{align}
			4k-1 - k \cdot (2k-1)^{\frac{1}{k}} \geq  2k-1\,. \label{eq:lhs_of_target_first_part_at_least_2k_minus_1_ordered}
		\end{align}
		Meanwhile, 
		\begin{align}
			\left({  (2k-1)^{\frac{1}{k}} -\frac{1}{k \cdot (2k-1)}   }\right)^{k} = (2k-1) \cdot \left(1 - \frac{1}{k \cdot (2k-1)^{1 + \frac{1}{k}}} \right)^k < 2k-1\,. \label{eq:rhs_of_target_is_at_most_2k_minus_1_ordered}
		\end{align}
		Combining \eqref{eq:lhs_of_target_first_part_at_least_2k_minus_1_ordered} and \eqref{eq:rhs_of_target_is_at_most_2k_minus_1_ordered}, we obtain  
		\begin{align}
			\left(4k-1-{k}\cdot (2k-1)^{\frac{1}{k}}\right)\left(1 + \frac{1}{2k-1}\right)^{{k-1}} & > \left(4k-1-{k}\cdot (2k-1)^{\frac{1}{k}}\right) \explain{Since $1 + {1}/{(2k-1)}  > 1$} \\
			& \geq 2k-1 \explain{By \eqref{eq:lhs_of_target_first_part_at_least_2k_minus_1_ordered}} \\
			& > \left({  (2k-1)^{\frac{1}{k}} -\frac{1}{k \cdot (2k-1)}   }\right)^{k} \,. \explain{By \eqref{eq:rhs_of_target_is_at_most_2k_minus_1_ordered}}
		\end{align}
		Thus the required inequality \eqref{eq:lhs_of_target_first_part_at_least_2k_minus_1_ordered} holds, which completes the proof.
	\end{proof}

	\begin{lemma} \label{lem:ruling_out_range_k_at_least_n_ordered}
		Let $n \geq 2$ and $k \geq 3$, where $k,n \in \mathbb{N}$. Suppose 
		\begin{align}
			n^{\frac{1}{k}} - \frac{1}{k} + \frac{2}{n} - \frac{2}{kn}   >  1 + \frac{2k}{n}  - \frac{k}{n^{\frac{k-1}{k}}} \,. \label{eq:first_property_helper_lemma_case1b_k_more_than_3_rounds_restated_in_separate_lemma}
		\end{align}
		Then $k < n$.
	\end{lemma}
	\begin{proof}
		We will show the constraint in \eqref{eq:first_property_helper_lemma_case1b_k_more_than_3_rounds_restated_in_separate_lemma}  is incompatible with the range $k \geq n$. 
		Let $t = n^{\frac{1}{k}}$. Then $n = t^k$. Since $k \geq  n$, we have $t =n^{\frac{1}{k}} \leq  k^{\frac{1}{k}}$. We have 
		\begin{align}
			&  n^{\frac{1}{k}} - \frac{1}{k} + \frac{2}{n} - \frac{2}{kn} >  1 + \frac{2k}{n} - \frac{k n^{\frac{1}{k}}}{n}, \; \;  \forall k \geq  n \iff \label{eq:target_inequality_in_terms_of_t_case_k_more_than_n_ordered_part_0} \\
			&  kt^{k+1} - t^k (k+1) + k^2t -2k^2 +2k - 2 > 0, \;  \forall t \in \bigl(0, k^{\frac{1}{k}}\bigr], \label{eq:target_inequality_in_terms_of_t_case_k_more_than_n_ordered}
		\end{align}
		where \eqref{eq:target_inequality_in_terms_of_t_case_k_more_than_n_ordered} is obtained from \eqref{eq:target_inequality_in_terms_of_t_case_k_more_than_n_ordered_part_0} by multiplying both sides by $kn$,  substituting $n = t^k$, and rearranging.
		
		In order to upper bound the left hand side of \eqref{eq:target_inequality_in_terms_of_t_case_k_more_than_n_ordered}, we  define a function   $f : [0, \infty) \to \mathbb{R}$ by 
		\[ f(x) =  x^k \left( k^{1+\frac{1}{k}} -k - 1\right) + k^2 x - 2k^2 + 2k -2\,.
		\] 
		For $0 \leq  t \leq  k^{\frac{1}{k}}$, we have $ kt^{k+1} \leq kt^{k} k^{\frac{1}{k}}$, so  the left hand side of \eqref{eq:target_inequality_in_terms_of_t_case_k_more_than_n_ordered} can be upper bounded as follows:
		\begin{align}
			kt^{k+1} - t^k (k+1) + k^2t -2k^2 +2k - 2 & \leq  t^k \left( k^{1+\frac{1}{k}} -k - 1\right) + k^2 t - 2k^2 + 2k -2 \notag \\
			& = f(t)\,. \label{eq:simpler_ub_expression_t_case_k_more_than_n_ordered}
		\end{align}
		Observe   
		$f'(x) = kx^{k-1}\left( k^{1+\frac{1}{k}} -k - 1\right) + k^2$ and $ f''(x) = k(k-1) x^{k-2} \left( k^{1+\frac{1}{k}} -k - 1\right)$.
		By Lemma~\ref{lem:helper_k_to_power_1_plus_1_over_k} the function $f$   is convex for all  $k \geq 3 $. Thus $f$ has a global maximum on the interval $[0, k^{\frac{1}{k}}]$ which is attained at one of the endpoints. We check the value of the function is  negative at both endpoints of $[0, k^{\frac{1}{k}}]$:
		\begin{itemize}
			\item $f(0) = -2k^2 + 2k - 2 = -k^2 - (k-1)^2 -1 < 0\,. $
			\item $f(k^{\frac{1}{k}})  = 2 k^2 k^{\frac{1}{k}} -3k^2 + k -2 
			< 0 $ by Lemma \ref{lem:f_of_k_to_1_over_k_is_negative_when_k_at_least_3}.
		\end{itemize}
		By convexity, it follows that $f(t) < 0$ for all $t \in [0, k^{\frac{1}{k}}]$.
		Combining this fact with  \eqref{eq:simpler_ub_expression_t_case_k_more_than_n_ordered}, we get 
		\begin{align}
			kt^{k+1} - t^k (k+1) + k^2t -2k^2 +2k - 2  \leq  f(t) < 0 \qquad \forall t \in [0, k^{\frac{1}{k}}],
		\end{align}
		which implies   
		\eqref{eq:target_inequality_in_terms_of_t_case_k_more_than_n_ordered}  cannot hold when $k \geq n$. Thus condition \eqref{eq:first_property_helper_lemma_case1b_k_more_than_3_rounds_restated_in_separate_lemma} in the lemma statement rules out the range  $k \geq  n$. This completes the proof.
	\end{proof}
	
	\begin{lemma} \label{lem:1_plus_1_over_n_all_to_k_minus_1_at_least_exp_of_big_fraction_with_2k_minus_2_minus_n_ordered}
		Let $k,n \in \mathbb{N}$ such that $k \geq 3$ and $k < n \leq 2k-2$. Then 
		\begin{align}
			\left(1+\frac{1}{n}\right)^{k-1} \geq  e^{\left( \frac{2k-2-n}{n \cdot n^{\frac{1}{k}}} \right)}\,. \label{eq:lem_target_inequality_1_plus_1_over_n_to_power_k_minus_1_at_least_e_to_big_fraction_exponent_ordered}
		\end{align} 
	\end{lemma}
	\begin{proof}
		Taking log on both sides of \eqref{eq:lem_target_inequality_1_plus_1_over_n_to_power_k_minus_1_at_least_e_to_big_fraction_exponent_ordered}, the required inequality is equivalent to 
		\begin{align}
			\ln\left(1 + \frac{1}{n}\right) \geq 
			\frac{2}{n^{1+\frac{1}{k}}} - \frac{1}{(k-1) \cdot n^{\frac{1}{k}}}  \,. \label{eq:required_ln_of_1_plus_1_over_n_at_least_2_over_big_term_minus_1_over_k_minus_1_times_n_to_1_over_k_ordered}
		\end{align}
		We first show several independent inequalities and then combine them to obtain the inequality required by the lemma.
		Recall that  $\ln(1+x) \geq \frac{x}{1+x}$ for all  $x > -1$ (see, e.g., \cite{bernoulli}). Taking $x = 1/n$ yields  
		\begin{align}
			\ln\left(1 + \frac{1}{n}\right) \geq \frac{1/n}{1 + 1/n} = \frac{1}{n+1}\,. \label{eq:ln_1_plus_1_over_n_at_least_1_over_1_plus_n_plus_1_ordered}
		\end{align}
		Since $n > k \geq 3$, we get $n \geq 3$.
		By Lemma~\ref{lem:helper_k_to_power_1_plus_1_over_k}, we obtain $n^{\left(1 + \frac{1}{n}\right)} \geq n+1$. Since $k < n$, we get 
		\begin{align}
			n^{\left(1 + \frac{1}{k}\right)} \geq n^{\left(1 + \frac{1}{n}\right)}  \geq n+1\,. \label{eq:n_to_power_1_plus_1_over_k_at_least_n_plus_1_ordered}
		\end{align}
		Next we will show that 
		\begin{align}
			\frac{1}{n+1} \geq \frac{2}{n^{1+\frac{1}{k}}} - \frac{1}{(k-1) \cdot n^{\frac{1}{k}}}, \label{eq:intermediate_1_over_n_plus_1_at_least_2_over_n_to_1_plus_1_over_k_minus_little_stuff_like_k_minus_1_times_n_to_1_over_k_ordered}
		\end{align}
		which is equivalent to 
		\begin{align}
			(k-1)\cdot  n^{1+\frac{1}{k}} \geq 2(k-1)(n+1) - n(n+1)  \,. \label{eq:intermediate_k_minus_1_times_all_n_to_1_plus_1_over_k_at_least_bunch_of_little_terms_ordered}
		\end{align}
		By \eqref{eq:n_to_power_1_plus_1_over_k_at_least_n_plus_1_ordered} we have 
		\begin{align} 
			(k-1) n^{\left(1+\frac{1}{k} \right)} \geq (k-1)(n+1)\,. \label{eq:k_minus_1_all_times_n_to_power_1_plus_1_over_k_ordered}
		\end{align}
		Since $n > k-1$, we have $k-1 > 2(k-1)-n$, which multiplied by $n+1$ on both sides gives 
		\begin{align}
			(k-1)(n+1) \geq 2(k-1)(n+1) - n(n+1) \,. \label{eq:k_minus_1_all_times_n_plus_1_at_least_2_times_k_minus_1_times_n_plus_1_minus_n_times_n_plus_1_ordered}
		\end{align}
		Combining \eqref{eq:k_minus_1_all_times_n_to_power_1_plus_1_over_k_ordered} and \eqref{eq:k_minus_1_all_times_n_plus_1_at_least_2_times_k_minus_1_times_n_plus_1_minus_n_times_n_plus_1_ordered} yields 
		\begin{align}
			(k-1) n^{\left(1+\frac{1}{k} \right)} & \geq  (k-1)(n+1) \explain{By \eqref{eq:k_minus_1_all_times_n_to_power_1_plus_1_over_k_ordered}} \\
			& \geq 2(k-1)(n+1) - n(n+1) \explain{By \eqref{eq:k_minus_1_all_times_n_plus_1_at_least_2_times_k_minus_1_times_n_plus_1_minus_n_times_n_plus_1_ordered}} \,.
		\end{align}
		Thus \eqref{eq:intermediate_k_minus_1_times_all_n_to_1_plus_1_over_k_at_least_bunch_of_little_terms_ordered} holds, so \eqref{eq:intermediate_1_over_n_plus_1_at_least_2_over_n_to_1_plus_1_over_k_minus_little_stuff_like_k_minus_1_times_n_to_1_over_k_ordered} holds as well. Combining \eqref{eq:ln_1_plus_1_over_n_at_least_1_over_1_plus_n_plus_1_ordered} and \eqref{eq:intermediate_1_over_n_plus_1_at_least_2_over_n_to_1_plus_1_over_k_minus_little_stuff_like_k_minus_1_times_n_to_1_over_k_ordered} yields 
		\begin{align}
			\ln\left(1 + \frac{1}{n}\right) & \geq \frac{1}{n+1} \explain{By \eqref{eq:ln_1_plus_1_over_n_at_least_1_over_1_plus_n_plus_1_ordered}} \\
			& \geq \frac{2}{n^{1+\frac{1}{k}}} - \frac{1}{(k-1) \cdot n^{\frac{1}{k}}} \explain{By \eqref{eq:intermediate_1_over_n_plus_1_at_least_2_over_n_to_1_plus_1_over_k_minus_little_stuff_like_k_minus_1_times_n_to_1_over_k_ordered}}
		\end{align}
		Thus \eqref{eq:required_ln_of_1_plus_1_over_n_at_least_2_over_big_term_minus_1_over_k_minus_1_times_n_to_1_over_k_ordered} holds, which is equivalent to the required inequality \eqref{eq:lem_target_inequality_1_plus_1_over_n_to_power_k_minus_1_at_least_e_to_big_fraction_exponent_ordered}.  This completes the proof.
	\end{proof}
	
	\begin{lemma} \label{lem:f_of_k_to_1_over_k_is_negative_when_k_at_least_3}
		Let  $k \in \mathbb{N}$ such that $k \geq 3$. Then
		$2k^2 \cdot k^{\frac{1}{k}} - 3k^2 + k - 2 < 0\,. $
	\end{lemma}
	\begin{proof}
		Let $f : (0, \infty) \to \mathbb{R}$ be 
		$f(x) = 2x^2 \cdot x^{\frac{1}{x}} - 3x^2 + x - 2$. We check separately for $k \in \{3, 4, 5, 6\}$:
		\begin{itemize}
			\item $f(3) = 18 \cdot 3^{\frac{1}{3}} - 26 < -0.01 < 0$ and  $f(4) = 32 \cdot 4^{\frac{1}{4}} - 46 < -0.7 < 0$.
			\item $f(5) = 50 \cdot 5^{\frac{1}{5}} - 72 < -3 < 0 $ and $f(6) = 72 \cdot 6^{\frac{1}{6}} - 104 < -6 < 0$.
		\end{itemize}
		Thus it remains to show  the required inequality when $k \geq 7$. 
		The function $x^{\frac{1}{x}}$ has a global maximum at $e^{\frac{1}{e}}$ (see, e.g., Wolfram Alpha \cite{wolfram}). Then 
		\begin{align}
			f(x) = 2x^2 \cdot x^{\frac{1}{x}} - 3x^2 + x - 2  & \leq 2x^2 \cdot e^{\frac{1}{e}} - 3x^2 + x-2 \notag \\
			& < -0.11 x^2 + x -2  \notag \\
			& < 0 \qquad \forall x \geq 7\,.
		\end{align}
		Thus $f(k) < 0$ for all $k \geq 3, k \in \mathbb{N}$, as required.
	\end{proof}
	
	\begin{lemma} \label{lem:helper_lemma_main_1}
		Let $k \geq 2, n\geq 1$,  $c \in [1/n,1]$, and the  sequence $\{\gamma_{\ell}\}_{\ell= 1}^{\infty}$ with  $\gamma_1 = 0$ and  $\gamma_{\ell} = 2\ell$ for  $\ell \geq 2$.
		Then  
		\begin{align} \label{eq:required_inequality_ckn_non-boundary}
			x\left(1 + \frac{\gamma_{k-1}}{n}\right) + ({k-1}) \cdot \frac{(n c - x)^{\frac{k}{k-1}}}{n \cdot x^{\frac{1}{k-1}}} - \gamma_{k-1} \cdot  c \geq  c \cdot k n^{\frac{1}{k}} - \gamma_k \cdot c, \qquad \forall x \in (1/2, nc]\,.
		\end{align}
	\end{lemma}
	\begin{proof}
		When $x = nc$, inequality \eqref{eq:required_inequality_ckn_non-boundary} is equivalent to 
		\begin{align}  \label{eq:required_boundary_case_x=nc}
			nc\left(1 + \frac{\gamma_{k-1}}{n}\right) - \gamma_{k-1} \cdot c & \geq c \cdot k n^{\frac{1}{k}} - \gamma_k \cdot  c t\iff  \\
			n - k n^{\frac{1}{k}} & \geq - \gamma_k \explain{Dividing both sides by $c$ and re-arranging terms.}\,.
		\end{align}
		Since $n  -k n^{\frac{1}{k}} \geq 1 - k$ for all $n \geq 1$, it follows that \eqref{eq:required_boundary_case_x=nc} holds if   $\gamma_k \geq k - 1$, which is the case since $\gamma_k = 2k$.
		Thus \eqref{eq:required_inequality_ckn_non-boundary} holds when $x = nc$.
		
		From now on we can assume $x \in (1/2, nc)$. 
		Let $ t = \left(nc/x - 1 \right)^{\frac{1}{k-1}}$.  Then $0 < t < (2nc-1)^{\frac{1}{k-1}}$. Equivalently, $x = \frac{nc}{1 + t^{k-1}}$, which substituted in 
		\eqref{eq:required_inequality_ckn_non-boundary} gives  
		\begin{align}
			& x\left(1 + \frac{\gamma_{k-1}}{n}\right) + ({k-1}) \cdot \frac{(n c - x)^{\frac{k}{k-1}}}{n \cdot x^{\frac{1}{k-1}}} - \gamma_{k-1} \cdot  c \geq c k n^{\frac{1}{k}} - \gamma_k \cdot c \iff \notag \\
			& \left( \frac{nc}{1 + t^{k-1}}\right) \left(1 + \frac{\gamma_{k-1}}{n}\right) + ({k-1}) \cdot \frac{\left(n c - \frac{nc}{1 + t^{k-1}}\right)^{\frac{k}{k-1}}}{n \cdot \left( \frac{nc}{1 + t^{k-1}}\right)^{\frac{1}{k-1}}} - \gamma_{k-1} \cdot  c \geq c k n^{\frac{1}{k}} - \gamma_k \cdot c \,. \label{eq:changing_from_x_to_t_variable_ordered}  
		\end{align}
		
		Multiplying both sides of \eqref{eq:changing_from_x_to_t_variable_ordered} by $\left({1+t^{k-1}}\right)/{c}$ and simplifying, we see   \eqref{eq:changing_from_x_to_t_variable_ordered} is equivalent to 
		\begin{align}
			(k-1) \cdot t^k - t^{k-1} \cdot \left(k n^{\frac{1}{k}} - \gamma_k + \gamma_{k-1}\right) + \left(n - k n^{\frac{1}{k}} + \gamma_k\right) \geq 0\,.\label{eq:rephrasing_t_variable}
		\end{align}
		We will show that \eqref{eq:rephrasing_t_variable}  holds, which will imply  inequality \eqref{eq:required_inequality_ckn_non-boundary} for all $x \in (1/2, nc)$.
		We consider two cases, depending on whether $k=2$ or $k\geq 3$.
		
		\paragraph{Case $k =2$.} Since $\gamma_1 = 0$ and $\gamma_2 = 4$,  inequality \eqref{eq:rephrasing_t_variable} is equivalent to 
		\begin{align}
			& t^2 - t \left(2 \sqrt{n} - 4 + 0 \right) + \left(n - 2 \sqrt{n} + 4 \right) \geq 0 \iff  \label{eq:simple_inequality_raw_sum_of_squares_k_2_ordered_search}\\
			& \left( t - \left(\sqrt{n} - 2 \right) \right)^2 + 2 \sqrt{n} \geq 0, \label{eq:simple_inequality_sum_of_squares_k_2_ordered_search}
		\end{align}
		where  \eqref{eq:simple_inequality_sum_of_squares_k_2_ordered_search} was obtained from \eqref{eq:simple_inequality_raw_sum_of_squares_k_2_ordered_search} by re-arranging terms. 
		Inequality \eqref{eq:simple_inequality_sum_of_squares_k_2_ordered_search} clearly holds, which implies \eqref{eq:rephrasing_t_variable}  and completes the analysis for $k=2$.
		
		\paragraph{Case $k \geq 3$.} We define a function 
		$h : \left(0, \infty \right) \to \mathbb{R}$ to capture the left hand side of \eqref{eq:rephrasing_t_variable}. Then we will show $h$ is non-negative on the entire domain, which will imply \eqref{eq:rephrasing_t_variable}. Let 
		\begin{align}
			h(t) = (k-1) \cdot t^k - t^{k-1} \cdot \left(k n^{\frac{1}{k}} - \gamma_k + \gamma_{k-1}\right) + \left(n - k n^{\frac{1}{k}} + \gamma_k\right)\,.
		\end{align}
		
		The first and second derivatives of $h$ are 
		\begin{align}
			h'(t) & = t^{k-1} \cdot k(k-1) - t^{k-2} \cdot (k-1) \left(k n^{\frac{1}{k}} - \gamma_k + \gamma_{k-1}\right) \notag \\
			h''(t) & = t^{k-2} \cdot k(k-1)^2 - t^{k-3} \cdot (k-1)(k-2)\left(k n^{\frac{1}{k}} - \gamma_k + \gamma_{k-1}\right)\,.
		\end{align}
		On  $\left(0, \infty\right)$ we have:
		\begin{itemize}
			\item the function $h'$ has a unique root at $t_1 = n^{\frac{1}{k}} + \frac{\gamma_{k-1}-\gamma_k}{k}$;
			\item the function   $h''$ has a unique root at $t_2 = \left(\frac{k-2}{k-1}\right)\left(n^{\frac{1}{k}} + \frac{\gamma_{k-1}-\gamma_k}{k}\right) = \left(\frac{k-2}{k-1}\right) t_1$.
		\end{itemize}  
		
		Clearly $t_2 < t_1$.
		Since $n \geq 1$ and $\gamma_k - \gamma_{k-1} = 2$ when $k \geq 3$, we have 
		\[ 
		n^{\frac{1}{k}} \geq 1 > \frac{2}{k} = \frac{\gamma_{k} - \gamma_{k-1}}{k} \geq 0\,.
		\]
		Thus $n^{\frac{1}{k}} +  ({\gamma_{k-1}-\gamma_k})/{k} > 0$,  so $t_2 > 0$.
		We obtain   $ 0 < t_2 < t_1$. 
		Moreover,  $h'(t) < 0$ for $t < t_1$ and $h'(t) > 0$ for $t > t_1$; similarly $h''(t) < 0$ for $t < t_2$ and $h''(t) > 0$ for $t > t_2$.
		Thus  $h$ is
		\begin{itemize}
			\item concave and decreasing on $(0,t_2)$;
			\item convex and decreasing on $(t_2, t_1)$;
			\item convex and increasing on $(t_1, \infty)$. 
		\end{itemize}  
		Thus $h$ has a unique global minimum at $t_1$, so   the required inequality \eqref{eq:rephrasing_t_variable} holds if    $h(t_1) \geq 0$.
		We have 
		\begin{align}
			h(t_1) & = (k-1) \cdot t_1^k - t_1^{k-1} \cdot \left(k n^{\frac{1}{k}} - \gamma_k + \gamma_{k-1}\right) + \left(n - k n^{\frac{1}{k}} + \gamma_k\right)  \notag \\
			& = n - k n^{\frac{1}{k}} + \gamma_k - \left(n^{\frac{1}{k}} + \frac{\gamma_{k-1}-\gamma_k}{k} \right)^k\,. \label{eq:h_t1}
		\end{align}

		Since $\gamma_{k} = 2k$ and $\gamma_{k-1} = 2(k-1)$ for $k\geq 3$, we have 
		
		\begin{align}
			n^{\frac{1}{k}} + \frac{\gamma_{k-1} - \gamma_k}{k} = n^{\frac{1}{k}} - \frac{2}{k}   \geq 1 - \frac{2}{k} > 0\,. \label{eq:simple_ineq_n_1_per_k_minus_2_per_k}
		\end{align}
		Using \eqref{eq:simple_ineq_n_1_per_k_minus_2_per_k} in \eqref{eq:h_t1} gives 
		\begin{align}
			h(t_1) & = n - k n^{\frac{1}{k}} + 2k - \left(n^{\frac{1}{k}} - \frac{2}{k} \right)^k \notag \\
			& > 0 \,. \explain{By Lemma~\ref{eq:helper_lemma_1}.}
		\end{align}
		Thus $h(t_1) \geq 0$, and  so inequality \eqref{eq:rephrasing_t_variable} also holds in the case $k \geq 3$.
		
		\paragraph{Combining the cases.}		In both cases $k=2$ and $k \geq 3$, inequality  \eqref{eq:rephrasing_t_variable} holds, which implies \eqref{eq:required_inequality_ckn_non-boundary} for all $x \in (1/2, nc)$. This completes the proof.
	\end{proof}
	
	\begin{lemma} \label{eq:helper_lemma_1}
		Let $n  \geq 1$ and $k \geq 3$, where $k, n \in \mathbb{N}$. Then  
		$
		n - k n^{\frac{1}{k}} + 2(k-1) -  \left(n^{\frac{1}{k}} - \frac{2}{k}\right)^k \geq 0\,.
		$
	\end{lemma}
	\begin{proof}
		Define  $f : \left[1 - \frac{2}{k}, \infty\right) \to \mathbb{R}$ as 
		\begin{align}  \label{eq:def_f_show_positive_ordered}
			f(x) = \left(x + \frac{2}{k}\right)^k - k \left(x + \frac{2}{k}\right) + 2(k-1) -  x^k = \left(x + \frac{2}{k}\right)^k - kx - x^k + 2k -4  \,.
		\end{align}
		The lemma statement requires showing $f\left(n^{\frac{1}{k}}-\frac{2}{k}\right) \geq 0$. We will show that  $f(x) \geq 0$ for all $x \geq 1 - 2/k$, which will imply the required inequality.
		We divide the range of $x$ in two parts and analyze each separately.
		
		\paragraph{Case $x \in \left[1 - \frac{2}{k},1\right]$.}  We consider a few sub-cases depending on the value of $k$:
		\begin{itemize} 
			\item If $k = 3$, then 
			$ f(x) = \left(x + \frac{2}{3}\right)^3 - 3x - x^3 +2 \cdot 3 - 4 = \frac{1}{27} \left( 54x^2 - 45x  + {62}  \right)\,.$
			Then $\Delta < 0$, so $f(x) > 0$ for all $x \in \mathbb{R}$.
			\item If $k \geq 4$, then  using the inequalities $1-2/k \leq x \leq 1$ in the definition of $f$ from \eqref{eq:def_f_show_positive_ordered} gives   
			\begin{align}
				f(x) \geq 1^k - k \cdot 1 - 1^k + 2k -4 = k - 4   \geq 0 \,.
			\end{align}
		\end{itemize}
		\paragraph{Case $x > 1$.} Then   
		\begin{align} \label{eq:lemma_simple_inequality_f}
			f(x) & \geq x^k + \binom{k}{1} \cdot x^{k-1} \cdot \frac{2}{k} + \binom{k}{2} \cdot x^{k-2} \cdot \left( \frac{2}{k} \right)^2 - kx - x^k + 2k -4 
			\notag \\
			& = 2 x^{k-1} +\frac{2(k-1)}{k} \cdot x^{k-2} - kx + 2k - 4 \,.
		\end{align}
		When $k = 3$, using  inequality  \eqref{eq:lemma_simple_inequality_f}, we obtain  
		$f(x) \geq 2x^2 - {5}x/3 + 2 \geq 0 \; \; \forall x \in [1, \infty)\,.$ 
		
		\noindent Thus from now on we can assume  $k \geq 4$.  
		Using $x > 1$ and $ k \geq 4$, we obtain 
		\begin{align} 
			f(x) & \geq 2 x^{k-1} +\frac{2(k-1)}{k} \cdot x^{k-2} - kx + 2k - 4 \explain{By  \eqref{eq:lemma_simple_inequality_f}} \\ 
			& > 2 x^{k-2}  - k x + k, \label{eq:simple_inequality_f_simpler}
		\end{align}
		
		Let $f_1: (0, \infty) \to \mathbb{R}$ be $f_1(x) = 2x^{k-2}-kx+k$. The derivatives are $f_1'(x) = 2(k-2)x^{k-3}-k$ and $f_1''(x)=2(k-2)(k-3)x^{k-4}$. Since we are in the case $k > 3$, we have $f_1''(x) > 0$ for $x > 0$. Thus the function $f_1$ is convex and has a unique global minimum at the point $x^*$ for which $f_1'(x^*)=0$, that is, at  
		\begin{align}
			x^* = \left( \frac{k}{2(k-2)}\right)^{\frac{1}{k-3}}\,.
		\end{align} Since $k \geq 4$, we have $\frac{k}{2(k-2)} \leq 1$, and so 
		\begin{align} \label{eq:inequality_simple_f1}
			f_1(x^*) & = 2  \left( \frac{k}{2(k-2)}\right)^{\frac{k-2}{k-3}} - k  \left( \frac{k}{2(k-2)}\right)^{\frac{1}{k-3}} + k  \geq 2  \left( \frac{k}{2(k-2)}\right)^{\frac{k-2}{k-3}}  - k \cdot 1 + k > 0 \,.
		\end{align}
		Combining \eqref{eq:simple_inequality_f_simpler} and  \eqref{eq:inequality_simple_f1} gives  
		$f(x) \geq f_1(x) \geq f_1(x^*) > 0 \; \; \forall x > 0\,.$ 
		In particular, the required inequality  holds for all $x > 1$, which completes the case.
		
		Combining the cases, we obtain $f(x) \geq 0$ for all $x \geq 1-2/k$, and so $f\left(n^{\frac{1}{k}} - 2/k \right) \geq 0$. This  completes the proof of the lemma.
	\end{proof}
	
	\begin{corollary}
		\label{corr:helper_lemma_1} 
		For each $n  \geq 1$ and $k \geq 3$, we have: 
		$
		n  + 2k >  k n^{\frac{1}{k}} + 2  \,.
		$
	\end{corollary}
	\begin{proof}  
		Lemma~\ref{eq:helper_lemma_1} yields  
		$
		n  + 2k \geq   k n^{\frac{1}{k}} + 2 + \left(n^{\frac{1}{k}} - \frac{2}{k}\right)^k \,.
		$
		Since $k \geq 3$, we also have $n^{\frac{1}{k}} \geq 1 > \frac{2}{k}$, so $\left(n^{\frac{1}{k}} - \frac{2}{k}\right)^k  > 0$. Thus $n + 2k > k n^{\frac{1}{k}} + 2 $, as required.
	\end{proof}

	\section{Appendix: Unordered search} \label{app:unordered_search}
	
	In this section we include the omitted proofs for unordered search. 
	
	\subsection{Unordered search upper bounds}
	
	Here we give  the optimal randomized algorithms on a worst case input and deterministic algorithms for any input distribution for unordered search.
	
	\paragraph{Deterministic algorithms for a worst case input.}
	We start with a simple observation, namely that the optimal $k$-round deterministic algorithm in the worst case just queries $n/k$ locations in each round.

	\begin{observation} \label{obs:k_round_deterministic_algorithm_unordered}
		For each $k \in \{1, \ldots, n\}$, 
		there is a deterministic $k$-round algorithm for ordered search that  always succeeds and asks at most $n$ queries in the worst case:  
		\begin{itemize} \item In each round $j \in [k]$, issue   $\lfloor n/k \rfloor$ or $\lceil n/k \rceil$ at locations not previously queried. When the item is found, return it and halt. 
		\end{itemize}
	\end{observation}
	\begin{proof}
		This algorithm queries $n$ locations in the worst case, and so always finds the element using at most $n$ queries.
	\end{proof}

	\paragraph{Randomized algorithms for a worst case input.}  The optimal randomized algorithm is described next.
	
	\bigskip 
	
	\noindent \textbf{Proposition~\ref{thm:select_ub_rand} (restated).}  
	\emph{Let $p \in (0,1]$ and $k,n \in \mathbb{N}_{\geq 1}$. Then}
	\[ 
	\mathcal{R}_{1-p}(\mbox{unordered}_{n,k}) \leq np \cdot \frac{k+1}{2k} + p + \frac{p}{n} \,.
	\] 
	\begin{proof}
		Consider the following algorithm, which has an all-or-nothing structure.
		\begin{description}
			\item[$\diamond$ With probability $1-p$:] do nothing.
			\item[$\diamond$ With probability $p$:] run the following protocol: 
			\begin{itemize} 
				\item Choose a uniform random permutation $\vec{\pi} = (\pi_1, \ldots, \pi_n)$ of $[n]$. For each $j \in [k]$, define  
				\[ m_j =  \left \lceil n \cdot {j}/{k} \right \rceil 
				\qquad \mbox{and} \qquad 
				S_j = \{\pi_1, \ldots, \pi_{m_j}\}\,.
				\] 
				\item In each round $j \in [k]$: query all the locations in $S_j$ that have not been queried yet. Whenever the element is found, return its location and halt immediately.
			\end{itemize}
		\end{description}
		We bound the success probability and the expected number of queries of the algorithm.
		\paragraph{Success probability.}
		If the algorithm finishes execution in exactly $j \geq 1$ rounds, then the number of queries issued is $|S_j| = \lceil n j /k \rceil$. By the end of the $k$-th round, the number of queries issued would be $\lceil n  k/k\rceil = n$. 
		Thus if the algorithm enters round $1$ then  it doesn't stop until finding  where the element is, so the success probability is exactly $p$.

		\paragraph{Expected number of queries.}
		Let $A_j$ be the event that the algorithm halts exactly at the end of round $j$.
		On event $A_j$, the algorithm issues  $\lceil n  {j}/{k} \rceil$ queries.
		The probability of event $A_j$ is
		\begin{align} \label{eq:pr_A_j_unordered_randomized_ub}
			\Pr(A_j) = p \cdot   \frac{\lceil n \cdot \frac{j}{k} \rceil - \lceil n \cdot \frac{j-1}{k}\rceil}{n}  \,.
		\end{align}
		Then the expected number of queries issued by the algorithm is $ q_k = \sum_{j=1}^k \Pr(A_j) \cdot \lceil n  {j}/{k}\rceil$.  Using \eqref{eq:pr_A_j_unordered_randomized_ub}, we can rewrite this as 
		\begin{align}
			q_k & = 
			\sum_{j=1}^k p \cdot \frac{\lceil n \cdot \frac{j}{k} \rceil - \lceil n \cdot \frac{j-1}{k}\rceil}{n} \cdot \left \lceil n \cdot \frac{j}{k} \right \rceil \notag \\
			&= np + \frac{p}{n} \cdot \sum_{j=1}^{k-1} \left( \left  \lceil n \cdot \frac{j}{k} \right \rceil - \left \lceil n \cdot \frac{j+1}{k} \right \rceil \right)  \left  \lceil n \cdot \frac{j}{k} \right \rceil \,. \label{eq:rewriting_q_k_to_fit_lemma_for_unordered_both_det_and_randomized}
		\end{align}
		Applying  Lemma~\ref{lem:unordered_upper_bound_shared} with $x = n$ to bound the  expression in  \eqref{eq:rewriting_q_k_to_fit_lemma_for_unordered_both_det_and_randomized} yields 
		\begin{align}
			q_k &\le np + \frac{p}{n}\cdot  \left( -\frac{n^2(k-1)}{2k} + \lceil n \rceil + 1 \right) = np \cdot \frac{k+1}{2k} + p + \frac{p}{n}\,.
		\end{align}
		This completes the proof.
	\end{proof}
	
	\paragraph{Deterministic algorithms for a random input.}  Given an input distribution $\Psi = (\Psi_1, \ldots, \Psi_n)$, we next design  an optimal deterministic algorithm  for it.
	
	\bigskip 
	
	\noindent \textbf{Proposition~\ref{thm:select_ub} (restated).} \emph{ Let $p \in (0,1]$ and $k,n \in \mathbb{N}_{\geq 1}$. 
		Then }
	\[\mathcal{D}_{1-p}(\mbox{unordered}_{n,k}) \leq 
	np \Bigl( 1 - \frac{k-1}{2k} \cdot p\Bigr) + 1+ p + \frac{2}{n}\,.
	\]
	\begin{proof}
		Suppose the input distribution is $\Psi = (\Psi_1, \ldots, \Psi_n)$.
		Let $\pi$ be a permutation of $[n]$ such that $\Psi_{\pi_1} \geq \ldots \geq \Psi_{\pi_n}$. For each $j \in [k]$, let  $S_j \subseteq [n]$ be  the top $\lceil np \cdot \frac{j}{k} \rceil$ array positions in the ordering given by $\pi$, that is:
		\[
		S_j = \left \{\pi_1, \ldots, \pi_{m_j} \right \}, \mbox{ where } m_j = \left \lceil np \cdot \frac{j}{k} \right \rceil \,.
		\]
		Consider the following algorithm.
		\begin{quote}
			\emph{In each round $j \in [k]$:
				Query the locations in $S_j$ that have not been queried in the previous $j-1$ rounds. 
				Once the element is found, return its location and halt immediately.}
		\end{quote}
		\paragraph{Success probability.} To bound the success probability of the algorithm, observe that the subsets $S_j$ are nested, that is: 
		$ S_1 \subseteq  \ldots \subseteq S_k \,.$
		By the end of round $k$, the algorithm has only queried locations from $S_k$ and either found the element or exhausted $S_k$. 
		
		For all $j \in [k]$, denote the probability that the sought element is in $S_j$ by 
		\[ 
		\phi_j = \sum_{\ell \in S_j} \Psi_{\ell}\,.
		\] 
		Lemma~\ref{eq:simple_lemma_sorted_array} gives  $\phi_j \geq {|S_j|}/{n}$. Then the success probability is  $
		\widetilde{p}  = \sum_{\ell \in S_k} \Psi_{\ell} 
		\geq \frac{|S_k|}{n} = \frac{\lceil np \cdot \frac{k}{k} \rceil}{n} \geq p\,.
		$ 
		\paragraph{Expected number of queries.} Next we bound the expected number of queries. For each $j \in [k]$, let $A_j$ be the event that the algorithm halts exactly at the end of round $j$. On event $A_j$, the algorithm issues a total of  $|S_j|$ queries.  Moreover,  the probability of event $A_j$ is 
		\begin{align}  \label{eq:probability_a_j_unordered_det_lb}
			\Pr(A_j) = 
			\begin{cases}
				\phi_j - \phi_{j-1} & \text{ if } 1 \leq j \leq k-1, \mbox{ where } \phi_0 = 0\,. \\
				1 - \phi_{k-1} & \text{ if } j = k\,.
			\end{cases}
		\end{align}
		Let $S_0 = \emptyset$. 
		For each $j \in \{0, \ldots, k\}$, define  
		\begin{align}
			\eta_j = 
			\phi_j - \frac{|S_j|}{n} \,.
		\end{align}
		We have $\eta_j \geq 0$  since $\phi_j \geq |S_j|/n$.
		
		Then the expected number $q_k$ of queries issued  on input distribution $\Psi$ can be bounded by:
		\begin{align}
			q_k & = \sum_{j=1}^{k} \Pr(A_j) \cdot |S_j|   = \left(1 - \phi_{k-1} \right) \cdot |S_k| + \sum_{j=1}^{k-1} \left( \phi_j - \phi_{j-1}\right) \cdot |S_j| \explain{By \eqref{eq:probability_a_j_unordered_det_lb}}   \\
			& = \left(1 - \frac{|S_{k-1}|}{n} - \eta_{k-1} \right) \cdot |S_k| + \sum_{j=1}^{k-1} \left( \frac{|S_j|}{n} + \eta_j -\frac{|S_{j-1}|}{n} - \eta_{j-1} \right) \cdot |S_j| \explain{By definition of $\eta_j$} \\
			& = \left[ \left(1 - \frac{|S_{k-1}|}{n}  \right) \cdot |S_k| + \sum_{j=1}^{k-1} \left( \frac{|S_j| - |S_{j-1}|}{n}  \right) \cdot |S_j|  \right] - \eta_{k-1} \cdot |S_k| + \sum_{j=1}^{k-1} \left( \eta_{j} - \eta_{j-1} \right) \cdot |S_j|\,. \label{eq:decomposing_q_j_unordered_deterministic_lb}
		\end{align}

		We have $0 = |S_0| \leq |S_1| \leq \ldots \leq |S_{k}|$, and so 
		$     \sum_{j=1}^{k-1} \left( \eta_{j} - \eta_{j-1} \right) \cdot |S_j| \leq \eta_{k-1} \cdot |S_{k-1}| \,.$ Thus  
		\begin{align} \label{eq:bound_telescopic_eta_j_up_to_k_minus_one_unordered}
			- \eta_{k-1} \cdot |S_k| + \sum_{j=1}^{k-1} \left( \eta_{j} - \eta_{j-1} \right) \cdot |S_j|  \leq - \eta_{k-1} \cdot |S_k| + \eta_{k-1} \cdot |S_{k-1}|  \leq 0 \,.
		\end{align}
		Using \eqref{eq:bound_telescopic_eta_j_up_to_k_minus_one_unordered} in \eqref{eq:decomposing_q_j_unordered_deterministic_lb} gives 
		\begin{align} \label{eq:intermediate_bound_q_k_deterministic_on_random_input_distribution_geq_sum_of_stuff_unordered}
			q_k \leq \left(1 - \frac{|S_{k-1}|}{n}  \right) \cdot |S_k| + \sum_{j=1}^{k-1} \left( \frac{|S_j| - |S_{j-1}|}{n}  \right) \cdot |S_j|  \,.
		\end{align}
		
		We observe that  
		\begin{align}
			\left(1 - \frac{|S_{k-1}|}{n} \right) \cdot |S_k| - \sum_{j=1}^{k-1} \frac{|S_{j-1}|}{n} \cdot |S_j| = |S_k| - \sum_{j=1}^{k-1} \frac{|S_{j+1}|}{n} \cdot |S_j| \,. \label{eq:rewriting_sum_s_j_times_s_j_plus_one_unordered}
		\end{align}
		Adding $\sum_{j=1}^{k-1} {|S_j|^2}/{n}$ to both sides of \eqref{eq:rewriting_sum_s_j_times_s_j_plus_one_unordered}, we obtain 
		\begin{align}
			\left(1 - \frac{|S_{k-1}|}{n} \right) \cdot |S_k| + \sum_{j=1}^{k-1} \left( \frac{|S_j| - |S_{j-1}|}{n}\right) \cdot |S_j| = |S_k| + \sum_{j=1}^{k-1} \left( \frac{|S_j| - |S_{j+1}|}{n} \right) \cdot |S_j|    \,.  \label{eq:result_of_rewrite_s_k_plus_sum_of_s_j_times_s_j_minus_s_j_plus_one_unordered}
		\end{align}
		Substituting \eqref{eq:result_of_rewrite_s_k_plus_sum_of_s_j_times_s_j_minus_s_j_plus_one_unordered} in \eqref{eq:intermediate_bound_q_k_deterministic_on_random_input_distribution_geq_sum_of_stuff_unordered} and using the identity $|S_j| = \lceil np \cdot j/k \rceil$ gives 
		\begin{align}
			q_k & \leq |S_k| + \sum_{j=1}^{k-1} \left( \frac{|S_j| - |S_{j+1}|}{n} \right) \cdot |S_j|   = \lceil np \rceil + \sum_{j=1}^{k-1} \left( \frac{ \lceil np \cdot \frac{j}{k} \rceil  - \lceil np \cdot \frac{j+1}{k} \rceil }{n} \right) \cdot \left \lceil np \cdot \frac{j}{k} \right \rceil  \,.  \label{eq:bound_so_far_for_q_k_unordered}
		\end{align}
		
		Applying Lemma~\ref{lem:unordered_upper_bound_shared} with $x = np$ gives 
		\begin{align}
			\sum_{j=1}^{k-1} \left( \left \lceil np \cdot \frac{j}{k} \right  \rceil - \left \lceil np \cdot \frac{j+1}{k} \right  \rceil \right) \cdot \left  \lceil np \cdot \frac{j}{k} \right  \rceil
			\le -\frac{(np)^2(k-1)}{2k} + \lceil np \rceil +1 \,.  \label{eq:result_of_applying_lemma_with_x_equal_np_unordered_deterministic}
		\end{align}
		Combining    \eqref{eq:bound_so_far_for_q_k_unordered}  and \eqref{eq:result_of_applying_lemma_with_x_equal_np_unordered_deterministic}  gives: 
		\begin{align}
			q_k &\le \lceil np \rceil + \frac{1}{n}\left( -\frac{(np)^2(k-1)}{2k} + \lceil np \rceil + 1 \right)
			= np \left( 1 - \frac{p(k-1)}{2k} \right) +  \frac{\lceil np \rceil }{n} + \frac{1}{n}  + \lceil np \rceil - np \notag \\
			& \leq np \left( 1 - \frac{p(k-1)}{2k} \right) + p + \frac{2}{n} + \lceil np \rceil - np \,.
		\end{align}
		This completes the proof.
	\end{proof}

	\subsection{Unordered search lower bounds}

	In this section we include the  unordered search lower bounds. 
	
	\bigskip 
	
	\noindent \textbf{Proposition~\ref{thm:randomized_lb_k_rounds_unordered_search} (restated).}
	\emph{Let $p \in (0,1]$ and $k,n \in \mathbb{N}_{\geq 1}$. Then $\mathcal{R}_{1-p}(\mbox{unordered}_{n,k}) \geq np \cdot \frac{k+1}{2k} $.}
	\begin{proof}
		For proving the required lower bound, it  will suffice to assume the input is drawn from the uniform distribution. 
		By an average argument, such a lower bound will also hold for a worst case input.
		
		Let  $\mathcal{A}_k$ be  a   $k$-round randomized algorithm that succeeds with probability $p$ when facing the uniform distribution as input and denote by  $q_k(n,p)$ the expected number of queries asked by  $\mathcal{A}_k$  on the uniform distribution.
		
		In round $1$, the algorithm has some probability $\delta_m$ of asking $m$ queries, for each $m \in \{0, \ldots, n\}$. Moreover, for  each such $m$, there are different (but finitely many) choices for the positions of the $m$ queries of round $1$. However, since the goal is to minimize the number of queries,  it suffices to restrict attention to the best way of positioning the queries in round $1$, breaking ties arbitrarily between different equally good options. For unordered search, each queried location is equivalent to any other since a query only reveals whether the element is there or not.
		
		For  each $m \in \{0, \ldots, n\}$, we  define  the following  variables:
		
		\begin{itemize}
			\item $\delta_m$ is the probability that the algorithm asks $m$ queries in the first round.
			\item  $\alpha_m$ is the probability that the algorithm finds the element in one of the rounds in $\{2, \ldots, k\}$, given that it didn't find it in the first round.
		\end{itemize} 
		
		The probability of finding the element in the first round is $m/n$, so the probability that the algorithm may need to continue to one of the rounds in $\{2, \ldots, k\}$ is $(n-m)/n$. The expected number of queries of $\mathcal{A}_k$ on the uniform distribution is 
		\begin{align} \label{eq:q_k_n_p_unordered}
			q_k(n,p) & = \sum_{m=0}^{n} \delta_m \left( m + \left( \frac{n-m}{n}\right) \cdot  q_{k-1}(n-m, \alpha_{m}) \right), 
		\end{align}
		where the variables are related by the  following constraints:
		\begin{align}
			&  \sum_{m=0}^{n} \delta_m  = 1 \\
			&	p_m = \frac{m}{n} + \left( \frac{n-m}{n} \right) \cdot \alpha_m, \qquad \forall m \in \{0, \ldots, n\}   \label{p_m_as_function_sum_of_blocks_unordered} \\ 
			& 			p = \sum_{m=0}^n \delta_m \cdot p_m \label{eq:unordered_p_sum_of_p_m} \\ 
			&  0 \leq \alpha_{m} \leq 1, \qquad \forall m \in \{0, \ldots, n\} \\
			& \delta_{m} \geq 0, \qquad \forall m  \in \{0, \ldots, n\} \,.
		\end{align}

		\paragraph{Base case.} Proposition~\ref{prop:lb_one_round} gives  $q_1(n,p) \geq np$, as required.
		
		\paragraph{Induction hypothesis.} Suppose $q_{\ell}(v,s) \geq v s \cdot \frac{\ell+1}{2\ell}$ for all $\ell \in [k-1]$, $v \in \mathbb{N}$, and $s \in [0,1]$.
		
		\paragraph{Induction step.} Using the induction hypothesis in \eqref{eq:q_k_n_p_unordered} gives
		\begin{align}
			q_{k}(n, p) & = \sum_{m=0}^{n} \delta_m \left( m + \frac{n-m}{n} \cdot q_{k-1}(n-m, \alpha_m) \right)  \notag \\
			& \geq \sum_{m=0}^{n} \delta_m \left( m + \frac{(n-m)^2}{n} \cdot \alpha_m \cdot \frac{k}{2k-2} \right)  \,.  \label{eq:intermediate_lb_unordered_randomized_expected}
		\end{align}
		Substituting $\alpha_m = ({n \cdot p_m - m})/({n-m})$ from \eqref{p_m_as_function_sum_of_blocks_unordered} in  \eqref{eq:intermediate_lb_unordered_randomized_expected} gives 
		\begin{align}
			q_{k}(n, p) & \geq  \sum_{m=0}^{n} \delta_m \left( m + \frac{(n-m)(n \cdot p_m - m)}{n}  \cdot \frac{k}{2k-2} \right) \,. \label{eq:unordered_q_k_lb_substituting_ih}
		\end{align}
		Lemma~\ref{lem:crucial_unordered_lb_lemma} gives 
		\begin{align}
			m + \frac{(n-m)(n \cdot p_m - m)}{n}  \cdot \frac{k}{2k-2} \geq n \cdot p_m \left( \frac{k+1}{2k} \right)\,. \label{eq:result_of_crucial_lemma_unordered_lb_randomized}
		\end{align}
		Using \eqref{eq:result_of_crucial_lemma_unordered_lb_randomized} in \eqref{eq:unordered_q_k_lb_substituting_ih} gives  
		\begin{align}
			q_{k}(n, p) & \geq \sum_{m=0}^n \delta_m \left( n \cdot p_m \cdot  \frac{k+1}{2k}  \right) = n \cdot \left( \frac{k+1}{2k} \right) \cdot \sum_{m=0}^n \delta_m \cdot p_m \notag  \\
			& = n p \cdot \left( \frac{k+1}{2k} \right)  \,.  \explain{Since $p = \sum_{m=0}^{n} \delta_m \cdot p_m$ by \eqref{eq:unordered_p_sum_of_p_m}}
		\end{align}
	\end{proof}
	
	Next we give the lower bound on the distributional complexity.
	
	\bigskip

	\noindent \textbf{Proposition~\ref{thm:det_lb_k_rounds_unordered_search} (restated).}
	\emph{Let $p \in (0,1]$ and $k,n \in \mathbb{N}_{\geq 1}$. Then $\mathcal{D}_{1-p}(\mbox{unordered}_{n,k}) \geq  np\left(1-\frac{k-1}{2k}p\right) $.}
	\begin{proof}
		For each $\ell \in \mathbb{N}$, let 
		$\mathcal{A}_{\ell}$ be  an optimal  $\ell$-round randomized algorithm that succeeds with probability $p$ when facing the uniform distribution as input. Let $q_{\ell}(n,p)$ be the expected number of queries of algorithm $\mathcal{A}_{\ell}$ when given an array of length $n$.
		
		Since $\mathcal{A}_k$ is deterministic,  it asks a fixed number $m$ of queries in round $1$. Moreover, since the input is drawn from the uniform distribution,  each location is equally likely to contain the answer, and so the actual locations do not matter, but rather only their number. Thus the probability of finding the answer in round $1$ is $m/n$.  Let $\alpha$ be the  probability that the algorithm finds the element in one of the later rounds in $\{2, \ldots, k \}$, given that the element was not found in the first round.
		
		Given these observations, the expected number of queries of the deterministic algorithm can be written as 
		\begin{align} \label{eq:objective_recursion_unordered}
			q_{k}(n,p) = m + \left( \frac{n-m}{n} \right) \cdot q_{k-1}(n-m, \alpha),
		\end{align}
		where the variables are related by the following constraints:
		\begin{align} \label{eq:p_constraints_recursion_unordered} 
			\begin{cases}
				&	p = \frac{m}{n} + \left( \frac{n-m}{n} \right) \cdot \alpha    \\ 
				&  0 \leq \alpha \leq 1\,.  
			\end{cases}
		\end{align} 
		
		We prove by induction on $k$ that that 
		\begin{align} \label{eq:target_unordered_lb_uniform_distribution}
			q_k(n,p) \geq np \left(1 - \frac{k-1}{2k}\cdot  p\right) \,.
		\end{align} 
		
		\paragraph{Base case.} Proposition~\ref{prop:lb_one_round} shows that $q_1(n,p) \geq np$.
		
		\paragraph{Induction hypothesis.} Suppose $q_{\ell}(v,s) \geq v s \left(1 - \frac{\ell-1}{2\ell}\cdot  s\right)$ for all $\ell \in [k-1]$, $v \in \mathbb{N}$, and $s \in [0,1]$.
		
		\paragraph{Induction step.} We prove   \eqref{eq:target_unordered_lb_uniform_distribution} holds for $k$ and all $n \in \mathbb{N}, p \in [0,1]$. The induction hypothesis gives 
		\begin{align} \label{eq:q_k_minus_1_unordered_deterministic_on_uniform_by_ih}
			q_{k-1}(n-m, \alpha) \geq (n-m) \alpha \left(1 - \frac{k-2}{2k-2} \cdot \alpha\right),
		\end{align}
		which substituted in \eqref{eq:objective_recursion_unordered} yields  
		\begin{align}
			q_{k}(n,p) & = m + \left( \frac{n-m}{n} \right) \cdot q_{k-1}(n-m, \alpha)   \geq m +  \frac{\alpha (n-m)^2}{n}    \left(1 - \frac{k-2}{2k-2} \cdot \alpha\right) \,.  
		\end{align}
		Since $\alpha = ({np-m})/({n-m})$ by \eqref{eq:p_constraints_recursion_unordered}, we obtain 
		\begin{align}
			q_{k}(n,p) & \geq 
			m + \frac{\left(np-m\right)\left(n-m\right)}{n} \left(1 - \frac{k-2}{2k-2} \cdot \left( \frac{np-m}{n-m}\right) \right) \notag  \\
			& \geq  np \left( 1 - \frac{k-1}{2k} \cdot p \right) \,. \explain{By Lemma~\ref{lem:lower_bound_inequality_unordered_deterministic_on_uniform_distribution}}
		\end{align}
		This completes the induction step and the proof. 
	\end{proof}
	
	\subsection{Lemmas for unordered search}
	
	In this section we include the lemmas used to prove the unordered search bounds.

	\begin{lemma} \label{lem:unordered_upper_bound_shared}
		Let $x \in \mathbb{R}$ and $k \in \mathbb{N}$, where $x, k > 0$. Then  
		\begin{align}
			\sum_{j=1}^{k-1} \left( \left \lceil x \cdot \frac{j}{k} \right  \rceil - \left \lceil x \cdot \frac{j+1}{k} \right  \rceil \right)  \left  \lceil x \cdot \frac{j}{k} \right  \rceil
			\le -\frac{x^2(k-1)}{2k} + \lceil x \rceil+1 \,.  \label{eq:target_ineq_lemma_with_bj_s_unordered}
		\end{align}
	\end{lemma}
	\begin{proof}
		For every $j \in [k]$, let $b_j = \lceil x  {j}/{k} \rceil - x  {j}/{k}$.
		The left hand side of \eqref{eq:target_ineq_lemma_with_bj_s_unordered} can be rewritten as 
		\begin{align}
			\sum_{j=1}^{k-1} \left( \left \lceil x \cdot \frac{j}{k} \right \rceil - \left  \lceil x \cdot \frac{j+1}{k} \right \rceil \right)  \left  \lceil x \cdot \frac{j}{k} \right \rceil
			&= \sum_{j=1}^{k-1} \left( x \cdot \frac{j}{k} + b_j - x \cdot \frac{j+1}{k} - b_{j+1}\right)  \left( x \cdot \frac{j}{k} + b_j \right)\\
			&= \sum_{j=1}^{k-1}  \left( -\frac{x^2j}{k^2} + b_j \left(b_j-b_{j+1} \right) + \frac{x}{k}  \Bigl( j(b_j - b_{j+1}) - b_j \Bigr) \right) \,. \label{eq:qk_sum_expanded}
		\end{align}
		
		The last term of the sum in \eqref{eq:qk_sum_expanded} almost entirely cancels:
		\begin{align} \label{eq:qk_sum_term_cancels}
			\sum_{j=1}^{k-1}  \frac{x}{k} \cdot \Bigl( j \left(b_j - b_{j+1} \right) - b_j \Bigr) = -b_k \cdot \frac{x(k-1)}{k}  \leq 0 \,.
		\end{align}
		Combining \eqref{eq:qk_sum_expanded} with \eqref{eq:qk_sum_term_cancels}, we get
		\begin{align} 
			\sum_{j=1}^{k-1} \left( \left \lceil x \cdot \frac{j}{k} \right \rceil - \left  \lceil x \cdot \frac{j+1}{k} \right \rceil \right)  \left  \lceil x \cdot \frac{j}{k} \right \rceil
			& \le \sum_{j=1}^{k-1} \left(  -\frac{x^2j}{k^2} + b_j \left(b_j-b_{j+1} \right) \right) \notag \\
			& =  -\frac{x^2(k-1)}{2k} + \sum_{j=1}^{k-1}   b_j \left(b_j-b_{j+1} \right)\,. \label{eq:qk_sum_expanded_3}
		\end{align}

		Next we bound the summation term  in \eqref{eq:qk_sum_expanded_3}. If  $b_{j} \ge b_{j+1} \ge b_{j+2}$ for some  $j \in [k-2]$, then
		\begin{align}
			b_{j}(b_{j} - b_{j+1}) + b_{j+1}(b_{j+1} - b_{j+2}) \leq b_{j}(b_{j} - b_{j+2}) \,. \label{eq:qk_bj_collapse_lemma}
		\end{align}
		Thus if there is a (weakly) decreasing sequence  $b_j \geq b_{j+1} \geq \ldots \geq b_{j+t}$ for some $t \geq 2$ and $j \in [k-t]$, then  applying inequality  \eqref{eq:qk_bj_collapse_lemma} iteratively gives 
		\begin{align} 
			\sum_{i = j}^{j+t-1} b_i \left( b_i - b_{i+1}\right)  \leq b_j \left(b_j - b_{j+t} \right)\,. \label{eq:ineq_decreasing_subsequence_of_length_t_unordered}
		\end{align}
		
		We will use inequality \eqref{eq:qk_bj_collapse_lemma}  to collapse some of the terms in the sum  $\sum_{j=1}^{k-1} b_j (b_j-b_{j+1})$.
		
		Towards this end, let $G = ([k], E)$ be a line graph where the vertices are $\{1, \ldots, k\}$ and the edges $E = \{(j, j+1) \mid j \in [k-1]\}$. 
		For each $j \in [k-1]$, if $b_j \geq  b_{j+1}$ then edge $(j,j+1)$ is colored with black and depicted as oriented down, and otherwise it is colored with yellow and oriented up.  
		
		We also give each vertex $j \in [k]$ a color $c_j \in \{R,B\}$, such that $c_1 = c_k = R$. Furthermore, for each $j \in [k-1]$, if $b_j < b_{j+1}$ then  both endpoints of the edge are colored red:   $c_j = c_{j+1} = R$. All other vertices are colored blue ($B$). See Figure~\ref{fig:b_j_sequence} for an illustration. 
		
		\begin{figure}[h!]
			\centering 
			\includegraphics[scale=0.8]{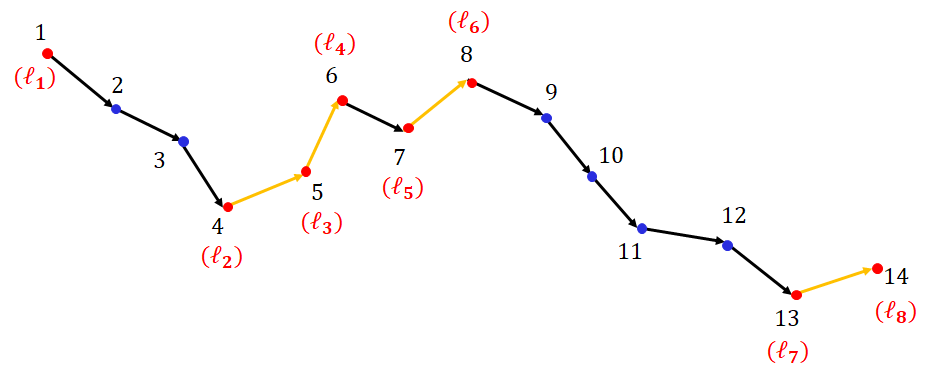}
			\caption{\small{Given $k \geq 2$ and numbers $b_1, \ldots, b_k \in [0,1)$, we construct a graph with edges  $(j, j+1)$ for each $j \in [k-1]$. For each $j \in [k]$, if $b_j \geq  b_{j+1}$, the edge from $j$ to $j+1$ is oriented downwards and is colored with black. If $b_j < b_{j+1}$, the edge from $j$ to $j+1$ is oriented upwards and is colored with yellow. The endpoints of all the yellow edges are added to the set $L$, together with special vertices $1$ and $k$. All the vertices in $L$ are colored red and the vertices in $[k] \setminus L$ are colored blue. For the graph in the picture we have  $k = 14$ and $L =
					\{1, 4, 5,6,7,8,13,14\}$. Each element $\ell_j$ of  $L$ is marked in red near the corresponding node.}}
			\label{fig:b_j_sequence}
		\end{figure}
		
		Let $\ell_1 = 1 < \ldots < \ell_m = k$ be the   red vertices in $G$ and $L = \{\ell_1, \ldots, \ell_m\}$. For all $i \in [m-1]$:
		\begin{itemize}
			\item  if the path from $\ell_i$ to $\ell_{i+1}$ has black edges, then $b_{\ell_i} \geq \ldots \geq b_{\ell_{i+1}}$ and so   inequality \eqref{eq:qk_bj_collapse_lemma} gives 
			\begin{align}
				\sum_{j = \ell_i}^{\ell_{i+1}-1} b_j (b_j - b_{j+1}) \leq b_{\ell_i} \left( b_{\ell_i} - b_{\ell_{(i+1)}} \right) \,. \label{eq:black_path_inequality}
			\end{align}
			\item else, the path from $\ell_i$ to $\ell_{i+1}$ has no black edges. Then $\ell_{i+1} = \ell_i + 1$, and so the next inequality trivially holds:
			\begin{align}
				b_{\ell_i} \left( b_{\ell_i} - b_{\ell_{(i+1)}} \right) \leq b_{\ell_i} \left( b_{\ell_i} - b_{\ell_{(i+1)}} \right) \,. \label{eq:consecutive_yellow_edge_inequality}
			\end{align}
		\end{itemize}
		Combining \eqref{eq:black_path_inequality} and \eqref{eq:consecutive_yellow_edge_inequality},  we can bound the sum of all $b_j$'s as follows:
		\begin{align}
			\sum_{j=1}^{k-1} b_j \left(b_j - b_{j+1}\right) \leq \sum_{i=1}^{m-1} b_{\ell_i} \left( b_{\ell_i} - b_{\ell_{(i+1)}} \right)\,.  \label{eq:qk_bj_collapse_pre}
		\end{align}


		%
		%
		%

		Since $b_j \in [0, 1)$ for all $j$, we have   \mbox{$b_{\ell_i}\left(b_{\ell_{i}} - b_{\ell_{i+1}} \right) \le (1-b_{\ell_{i+1}})$} and \mbox{$b_{\ell_{i+1}}\left(b_{\ell_{i+1}} - b_{\ell_{i+2}} \right) \le b_{\ell_{i+1}}$}.
		Thus adjacent terms in \eqref{eq:qk_bj_collapse_pre} sum to at most $1$. Then 
		\begin{itemize}
			\item If $m-1$ is even, then  
			\begin{align}
				\label{eq:even_m_minus_one_even_unordered_det_ub}
				\sum_{i = 1}^{m-1} b_{\ell_i} \left(b_{\ell_i} - b_{\ell_{(i+1)}} \right) \leq \frac{m-1}{2} < \frac{m}{2} \,. 
			\end{align}
			\item If $m-1$ is odd, then 
			\begin{align}
				\label{eq:odd_m_minus_one_even_unordered_det_ub}
				\sum_{i = 1}^{m-1} b_{\ell_i} \left(b_{\ell_i} - b_{\ell_{(i+1)}} \right) & \leq \left  \lfloor \frac{m-1}{2}\right \rfloor + b_{\ell_i} \left(b_{\ell_i} - b_{\ell_{(i+1)}} \right)   \leq  \left  \lfloor \frac{m-1}{2}\right \rfloor +  1  = \frac{m}{2}\,.
			\end{align}
		\end{itemize}
		Combining \eqref{eq:even_m_minus_one_even_unordered_det_ub} and \eqref{eq:odd_m_minus_one_even_unordered_det_ub}, we obtain 
		\begin{align} \label{eq:sum_of_l_i_s_at_most_m_over_2_unordered_det_ub}
			\sum_{i = 1}^{m-1} b_{\ell_i} \left(b_{\ell_i} - b_{\ell_{(i+1)}} \right) \leq \frac{m}{2} \,.
		\end{align}
		Combining \eqref{eq:qk_bj_collapse_pre} with \eqref{eq:sum_of_l_i_s_at_most_m_over_2_unordered_det_ub} while summing over all $j \in [k-1]$ gives 
		\begin{align} \label{eq:qk_bj_sum_le_m_over_2}
			\sum_{j=1}^{k-1} b_j \left(b_j-b_{j+1} \right) &\le \frac{m}{2} \,. 
		\end{align}
		
		Let $D = \{j \in [k-1] \mid  b_j < b_{j+1}\}$ and  $\Delta = |D|$.
		Since $b_j = \lceil x j/k \rceil - xj/k$, 
		we have $b_j \le b_{j+1}$ if and only if $\lceil x j/k \rceil - xj/k \leq   \lceil x (j+1)/k \rceil - x(j+1)/k$ $(\dag)$. Since $x/k > 0$, inequality  $(\dag)$ implies 
		\begin{align}
			\lceil x  {j}/{k} \rceil + 1 &\le \lceil x  {(j+1)}/{k} \rceil   \qquad \forall j \in D\,.  \label{eq:bj_ascending}
		\end{align}
		Consider the elements of $D$ in sorted order: $d_1 < \ldots < d_{\Delta}$.
		We will show by induction that 
		\begin{align} 
			\lceil x \cdot {d_i}/{k} \rceil\ge i \; \mbox{ for all } i \in [\Delta]\,. \label{eq:induct_hypo_x_times_d_i_over_k_unordered}
		\end{align}
		
		The base case is $i=1$. Indeed $\lceil x \cdot {d_1}/{k} \rceil \ge 1$ since  $ x \cdot {d_1}/{k} > 0$.
		We assume inequality \eqref{eq:induct_hypo_x_times_d_i_over_k_unordered} holds for $i$ and  show this implies the inequality for $i+1$.
		We have
		\begin{align}
			\left \lceil x \cdot \frac{d_{i+1}}{k} \right \rceil
			&\ge \left \lceil x \cdot \frac{d_i + 1}{k}\right \rceil \explain{Since $d_{i+1} > d_i$ and $d_i, d_{i+1} \in \mathbb{N}$.}\\
			&\ge \left \lceil x \cdot \frac{d_i}{k} \right \rceil + 1 \explain{By \eqref{eq:bj_ascending}.}\\
			&\ge i+1\,. \explain{By the inductive hypothesis.}
		\end{align}
		This completes the induction, so \eqref{eq:induct_hypo_x_times_d_i_over_k_unordered} holds.
		Now  we can bound the size of $D$. Since $d_{\Delta} \in [k-1]$, we have $ \left \lceil  x \cdot {k}/{k} \right \rceil  \geq \left \lceil  x \cdot {d_{\Delta}}/{k} \right \rceil$. By \eqref{eq:induct_hypo_x_times_d_i_over_k_unordered}, we have $\left \lceil  x \cdot {d_{\Delta}}/{k} \right \rceil \geq \Delta$. Thus 
		\begin{align}
			\lceil x \rceil =  \left \lceil  x \cdot \frac{k}{k} \right \rceil  & \geq \left \lceil  x \cdot \frac{d_{\Delta}}{k} \right \rceil 
			\geq \Delta \,.   
			\label{eq:D_bound}
		\end{align}
		
		Observe that $\Delta$  is equal to the number of yellow edges in the graph $G$, since $j \in D$  if and only if the edge $(j, j+1)$ is yellow. Thus the number of endpoints of yellow edges in $G$ is at most $2\Delta$. Since $|L| = m$ and $L$ consists precisely of all the endpoints of yellow edges together with vertices $1$ and $k$, we have 
		$m = |L| \leq | \{1, k\}| + 2\Delta = 2 + 2\Delta \,.$  Since $\Delta \leq \lceil x \rceil$ by  \eqref{eq:D_bound}, we obtain 
		\begin{align}
			m \leq 2 + 2 \Delta \leq 2 + 2 \lceil x \rceil \,.  \label{eq:m_bound}
		\end{align}
		%
		Combining \eqref{eq:m_bound} with \eqref{eq:qk_bj_sum_le_m_over_2}, we get
		\begin{align}
			\sum_{j=1}^{k-1} b_j \left(b_j-b_{j+1} \right) &\le \lceil x \rceil +1 \,. \label{eq:qk_bj_collapse_post}
		\end{align}
		Combining  \eqref{eq:qk_bj_collapse_post} with \eqref{eq:qk_sum_expanded_3} gives 
		\begin{align}
			\sum_{j=1}^{k-1} \left( \left \lceil x \cdot \frac{j}{k} \right \rceil - \left \lceil x \cdot \frac{j+1}{k}\right  \rceil \right)  \left \lceil x \cdot \frac{j}{k} \right \rceil
			& \leq -\frac{x^2(k-1)}{2k} + \sum_{j=1}^{k-1}   b_j \left(b_j-b_{j+1} \right) \leq -\frac{x^2(k-1)}{2k} + \lceil x \rceil + 1\,. \notag 
		\end{align}
		This completes the proof.
	\end{proof}
	
	\begin{lemma} \label{lem:lower_bound_inequality_unordered_deterministic_on_uniform_distribution}
		Let $k, n \in \mathbb{N}$, $x \in [0,n]$, and $p \in [0,1]$.  Suppose $k \geq 2$. Then 
		\begin{align} \label{eq:unordered_deterministic_uniform_lb_main_inequality}
			x + \frac{(np-x)\left(n-x \right)}{n} \left( 1 - \frac{k-2}{2k-2} \cdot \frac{np-x}{n-x} \right) \geq np \left( 1 - \frac{k-1}{2k} \cdot p \right)\,.
		\end{align}
	\end{lemma}
	\begin{proof}
		Let $f: \mathbb{R} \to \mathbb{R}$ be 
		\begin{align}
			f(x) = x + \frac{(np-x)\left(n-x \right)}{n}  \left( 1 - \frac{k-2}{2k-2} \cdot \frac{np-x}{n-x} \right) - np \left( 1 - \frac{k-1}{2k} \cdot p \right)\,.
		\end{align}
		Then the required inequality \eqref{eq:unordered_deterministic_uniform_lb_main_inequality} is equivalent to showing $f(x) \geq 0$ for all $x \in [0,n]\,.$
		
		Expanding the terms in the expression for $f(x)$, we get \begin{align}
			& f(x) \geq 0 \iff \notag \\
			& x + np - x - \frac{x(np-x)}{n} - \left( \frac{k-2}{2k-2} \right) \frac{(np-x)^2}{n}
			-np + np^2 \left( \frac{k-1}{2k} \right) \geq 0, \label{eq:desired_inequality_equivalent_to_f_x_non_negative_unordered_lb}
		\end{align}
		which after simplification is equivalent to 
		\begin{align} 
			& x^2 k^2 -x \cdot 2knp + n^2 p^2 \geq 0\,. \label{eq:quadratic_simplified_lb_unordered_deterministic_on_uniform}
		\end{align}
		The quadratic equation in \eqref{eq:quadratic_simplified_lb_unordered_deterministic_on_uniform} has a unique global minimum at $x^* = np/k$, with $f(x^*) = 0$. Thus \eqref{eq:quadratic_simplified_lb_unordered_deterministic_on_uniform} holds, so \eqref{eq:desired_inequality_equivalent_to_f_x_non_negative_unordered_lb} holds and so $f(x) \geq f(x^*) = 0 \, \, \forall x \in [0,n]$ as required.
	\end{proof}
	
	\begin{lemma} \label{lem:crucial_unordered_lb_lemma}
		Let $k, n, m \in \mathbb{N}$, where $k \geq 2$, $n \geq 1$, and $m \in \{0, \ldots, n\}$. Let $\gamma \in \left[ {m}/{n}, 1\right]$. Then 
		\begin{align}
			m + \frac{(n-m)(n \cdot \gamma - m)}{n}  \cdot \frac{k}{2k-2} - n \cdot \gamma \left( \frac{k+1}{2k} \right) \geq 0 \,. \label{eq:target_inequality_crucial_lemma_unordered_randomized_lb}
		\end{align}
	\end{lemma}
	\begin{proof}
		Inequality \eqref{eq:target_inequality_crucial_lemma_unordered_randomized_lb} is equivalent to  
		\begin{align}
			& m + n  \gamma \cdot \frac{k}{2k-2} - m \cdot \frac{k}{2k-2} - m \gamma \cdot \frac{k}{2k-2} + \frac{m^2}{n} \cdot \frac{k}{2k-2} - n \cdot \gamma \cdot  \frac{k+1}{2k} \geq 0 \,.  \label{eq:rephrased_target_inequality_crucial_lemma_unordered_randomized_lb}
		\end{align}
		Multiplying both sides of \eqref{eq:rephrased_target_inequality_crucial_lemma_unordered_randomized_lb} by $2nk(k-1)$ and rearranging, we get that \eqref{eq:target_inequality_crucial_lemma_unordered_randomized_lb} is equivalent to 
		\begin{align}
			\gamma  n (n - mk^2) + mn (k^2-2k) +  m^2 k^2  \geq 0\,.   \label{eq:simpler_target_unordered_randomized_lb_rephrased}
		\end{align}
		If $n \geq m  k^2$ then \eqref{eq:simpler_target_unordered_randomized_lb_rephrased} clearly holds. Else, assume $n < m k^2$.  Since $m/n \leq \gamma \leq 1$, we have
		\begin{align}
			\gamma n (n - mk^2) \geq n (n - mk^2) \,. \label{eq:unordered_simple_but_crucial_inequality_gamma_product_randomized_lb}
		\end{align}
		Using \eqref{eq:unordered_simple_but_crucial_inequality_gamma_product_randomized_lb}, we can bound the left hand side of \eqref{eq:simpler_target_unordered_randomized_lb_rephrased} as follows:
		\begin{align}
			\gamma n (n - mk^2 ) + mn (k^2-2k) +  m^2 k^2 & \geq n( n - mk^2 ) +  mn (k^2-2k) +  m^2 k^2 = (n - mk)^2 \geq 0\,. \notag 
		\end{align}
		Thus \eqref{eq:simpler_target_unordered_randomized_lb_rephrased} holds when $n < mk^2$ as well, which implies  \eqref{eq:target_inequality_crucial_lemma_unordered_randomized_lb} holds in all cases, as required.
	\end{proof}

	\section{Appendix: Cake cutting and sorting in rounds} \label{app:cake_cutting_and_sorting}
	
	In this section we study cake cutting in rounds and discuss the connection between sorting with rank queries and proportional cake cutting. We first introduce the cake cutting model.
	
	\paragraph{Cake cutting model.} The resource (cake) is represented as the interval $[0,1]$. There is a set of players $N=\{1,\ldots,n\}$, such that each player $i\in N$ is endowed with a private \emph{valuation function} $V_i$ that assigns a value to every subinterval of $[0,1]$. These values are induced by a non-negative integrable \emph{value density function} $v_i$, so that for an interval $I$, $V_i(I)=\int_{x\in I} v_i(x)\ dx$. The valuations are additive, so $V_i\left(\bigcup_{j=1}^m I_j\right) = \sum_{j=1}^m V_i(I_j)$ for any disjoint intervals $I_1, \ldots, I_m \subseteq [0,1]$. The value densities are non-atomic, and sets of measure zero are worth zero to a player.
	W.l.o.g., the valuations are normalized to $V_i([0,1]) = 1$, for all $i = 1 \ldots n$. 
	
	A \emph{piece of cake} is a finite union of disjoint intervals. A piece is \emph{connected} (or contiguous) if it consists of a single interval.
	An \emph{allocation} $A = (A_1, \ldots, A_n)$ is a partition of the cake among the players, such that each player $i$ receives the piece $A_i$, the pieces are disjoint, and $\bigcup_{i \in N} A_i = [0,1]$.
	An allocation $A$ is said to be \emph{proportional} if $V_i(A_i) \geq 1/n$ for all $i \in N$.
	
	\paragraph{Query complexity of cake cutting.} All the discrete cake cutting protocols operate in a query model known as the Robertson-Webb model (see, e.g., the book of \cite{RW98}), which was explicitly stated by~\cite{woeginger2007complexity}. In this model, the protocol communicates with the players using the following types of queries:
	\begin{itemize}
		\item{}$\emph{\textbf{Cut}}_i(\alpha)$: Player $i$ cuts the cake at a point $y$ where $V_{i}([0,y]) = \alpha$, where $\alpha \in [0,1]$ is chosen arbitrarily by the center \footnote{Ties are resolved deterministically, using for example the leftmost point with this property.}. The point $y$ becomes a {\em cut point}.
		\item{}$\emph{\textbf{Eval}}_i(y)$:  Player $i$ returns $V_{i}([0,y])$, where $y$ is a previously made cut point.
	\end{itemize}
	
	An RW protocol asks the players a sequence of cut and evaluate queries, at the end of which it outputs an allocation demarcated by cut points from its execution (i.e. cuts discovered through queries).
	Note that the value of a piece $[x,y]$ can be determined with two Eval queries, $\textrm{Eval}_i(x)$ and $\textrm{Eval}_i(y)$.
	
	When a protocol runs in $k$ rounds, then multiple RW queries (to the same or different agents) can be issued at once in each round. Note the choice of queries submitted in round $j$ cannot depend on the results of queries from the same or later rounds (i.e. $j, j+1, \ldots, k$).
	
	\subsection{Upper bounds}
	
	We will devise a protocol that finds a proportional allocation of the cake in $k$ rounds of interaction, which will also give a protocol for sorting with rank queries. For the special case of one round, a proportional protocol was studied in~\cite{BBKP14,MO12}.
	Our high level approach is to iteratively divide the cake into subcakes and assign agents to each subcake.
	

	\begin{prop} \label{prop_ub}
		There is an algorithm that runs in $k$ rounds and computes a proportional allocation with a total of $O(k n^{1+1/k})$ RW queries.
	\end{prop}
	
	We first describe the algorithm, and then prove Proposition~\ref{prop_ub}.
	The idea behind the algorithm is to partition the cake into $n^{1/k}$ subcakes and assign $n^{1 - 1/k}$ agents to each section, such that every agent believes that if they ultimately get a proportional share of their subcake, then they will have a proportional slice overall.
	Then all that remains is to recurse on each subcake in parallel in the successive rounds.
	
	One complication is that our only method of asking agents to cut a subcake, the $Cut$ query, requires that we know the values of the boundary of the subcake to that agent.
	However, the boundaries of the subcakes are known only with respect to one agent (possibly different agents for each boundary).
	We circumvent this difficulty by instead asking each agent to divide a further subset of their subcake whose boundary values for their valuation are known.
	In Algorithm 1, this further subset for each agent $i$ is the interval $[\textrm{Cut}_i(a_i), \textrm{Cut}_i(b_i))$.

	\paragraph{Algorithm 1.} \label{alg:algorithm_1}
	Input:
	\begin{itemize}
		\item Cake interval $[x,y]$ to be divided.
		\item Agent set $A$ among whom the cake is to be allocated.
		\item For each agent $i \in A$, values $a_i$ and $b_i$ in $[0,1]$.
		\item Number of remaining rounds $k$.
	\end{itemize}
	Procedure:
	\begin{enumerate}
		\item If $|A| = 1$, allocate the whole interval to the sole agent. Otherwise, continue.
		\item Define $z = \lceil |A|^{1/k} \rceil$ and define $m_j = \lceil |A| \cdot \frac{j}{z} \rceil - \lceil |A| \cdot \frac{j-1}{z} \rceil$ for each $j \in [z]$.
		\item Query $\textrm{Cut}_i\left( a_i + (b_i-a_i) \cdot \frac{1}{|A|} \sum_{\ell = 1}^j m_\ell \right)$ for all agents $i \in A$ and all $j \in [z-1]$.
		\item For $j = 1, 2, \ldots, z-1$:
		\begin{enumerate}
			\item Select $S_j$ to be the $m_j$ agents $i$ among $A \setminus \left( \bigcup_{\ell=1}^{j-1} S_\ell \right)$ with the smallest values for $\textrm{Cut}_i\left( a_i + (b_i-a_i) \cdot \frac{1}{|A|} \sum_{\ell = 1}^j m_\ell \right)$.
			\item Set $c_j$ to be the $m_j$th smallest value for $\textrm{Cut}_i\left( a_i + (b_i-a_i) \cdot \frac{1}{|A|} \sum_{\ell = 1}^j m_\ell \right)$ among all $i \in A \setminus \left( \bigcup_{\ell=1}^{j-1} S_\ell \right)$.
		\end{enumerate}
		\item Set $S_z = A \setminus \left( \bigcup_{\ell=1}^{z-1} S_\ell \right)$, $c_0 = 0$, and $c_z = 1$.
		\item In parallel in the following rounds, recurse on the the following instance for each $j \in [z]$:
		\begin{itemize}
			\item The cake interval to be divided is $[c_{j-1}, c_j]$.
			\item The set of agents is $S_j$.
			\item For each agent $i \in S_j$, set $new(a_i) =  a_i + (b_i-a_i) \cdot \frac{1}{|A|} \sum_{\ell = 1}^{j-1} m_\ell$.
			\item For each agent $i \in S_j$, set $new(b_i) =  a_i + (b_i-a_i) \cdot \frac{1}{|A|} \sum_{\ell = 1}^j m_\ell$.
			\item The number of remaining rounds is $k-1$.
		\end{itemize}
	\end{enumerate}
	
	To initially run Algorithm $1$, use as input the following parameters.
	The cake interval to be divided is $[0,1]$. The set of agents is $[n]$.
	For each agent $i \in N$, set $a_i = 0$ and $b_i = 1$.
	The number of (remaining) rounds is $k$.

	\paragraph{Example of running Algorithm 1.}
	
	\begin{example}
		Let $n = 4$ and $k =2$.
		Let the agents' value densities be as shown in Figure \ref{fig:cake_example}.
		After the first round we will have 
		\begin{itemize}
			\item $\textrm{Cut}(a_1) = 0.65, \textrm{Cut}(b_1) = 1, \textrm{Cut}(a_2) = 0.5, \textrm{Cut}(b_2) = 1, \textrm{Cut}(a_3) = 0, \textrm{Cut}(b_3) = 0.45$,\\$ \textrm{Cut}(a_4) = 0, \textrm{Cut}(b_4) = 0.4$.
		\end{itemize}
		The dividing line between the two subcakes, i.e. $c_1$, will be $\textrm{Cut}(a_3) = 0.5$.
		
		\begin{figure}[h!]
			\centering
			\includegraphics[scale=0.8]{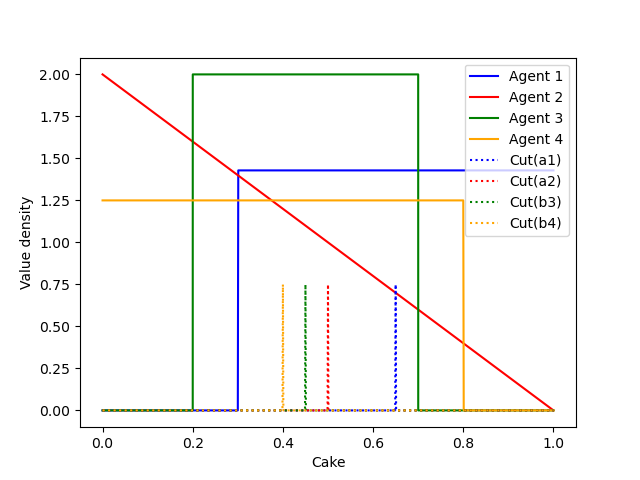}
			\caption{A potential value distribution for four agents from Example 1.
				When running Algorithm 1 on the shown value density functions with $k=2$, after the first round we will have $a_1 = a_2 = 0.5$, $b_1 = b_2 = 1$, $a_3 = a_4 = 0$, and $b_3 = b_3 = 0.5$.
				This leads to the $\textrm{Cut}$ values shown.
				In the second round, Algorithm 1 will recurse on the subcake $[0,\textrm{Cut}(a_3)) = [0, 0.5]$ with agents $3$ and $4$ and on the subcake $[0.5,1]$ with agents $1$ and $2$.
			}
			\label{fig:cake_example}
		\end{figure}
		
	\end{example}
	
	\begin{example}
		Let $n=1000$ and $k=3$. Algorithm 1 works as follows in each round:
		\begin{enumerate} 
			\item Round 1: everyone is asked to mark their $\frac{1}{10}, \frac{2}{10}, \ldots, \frac{9}{10}$ points.
			These are used to separate the agents into $10$ subcakes, each containing $100$ agents.
			\item Round 2: Everyone is asked to mark their $\frac{1}{10}, \frac{2}{10}, \ldots, \frac{9}{10}$ points within their respective value interval $[a_i, b_i]$.
			For example, for the second subcake each agent marks their $\frac{11}{100}, \frac{12}{100}, \ldots, \frac{19}{100}$ points.
			Again these are used to separate each subcake further into $10$ subcakes, each containing $10$ agents.
			\item Round 3: Everyone is asked to mark their $1/10, \ldots, 9/10$ points within their respective value interval.
			This time when assigning agents to subcakes, the algorithm assigns only $1$ to each, so we're done. 
		\end{enumerate}
	\end{example}
	

	Next we prove that the algorithm correctly computes a proportional allocation of the cake in $k$ rounds. 
	
	\begin{proof}[{Proof of Proposition} \ref{prop_ub}]
		
		Consider Algorithm 1.
		We claim that after $j$ rounds each subcake contains at most $n^{1-j/k}$ agents.
		In the base case, after $0$ rounds, the sole subcake contains all $n$ agents.
		In the inductive case, we assume that after $j$ rounds each subcake contains at most $n^{1-j/k}$ agents.
		Consider an arbitrary subcake containing $m$ agents and an arbitrary $\ell$.
		Then
		\begin{align}
			m_\ell = \lceil m \cdot \frac{j}{\lceil n^{1/k} \rceil} \rceil - \lceil m \cdot \frac{j-1}{\lceil n^{1/k} \rceil} \rceil
			\le m^{1 - 1/(k-j)} \le n^{1 - \frac{j+1}{k}}
		\end{align}
		This concludes the induction.
		Then after $k$ rounds each subcake contains at most $n^{1-k/k} = 1$ agents.
		Thus the algorithm generates an allocation in $k$ rounds.
		
		Next we claim inductively that at the start of every call to Algorithm 1, for all $i \in N$ we have $x \le \textrm{Cut}_i(a_i)$ and $\textrm{Cut}_i(b_i) \le y$.
		In the initial call to Algorithm 1 this is true since $0 \le \textrm{Cut}_i(0)$ and $\textrm{Cut}_i(1) \le 1$.
		In the recursive call in step 6, consider an arbitrary $j \in [z]$ and an arbitrary agent $i \in S_j$.
		If $j=1$, then $c_{j-1} = x \le a_i = a_i + a_i + (b_i-a_i) \cdot \frac{1}{|A|} \sum_{\ell = 1}^{j-1} m_\ell = new(a_i)$ by inductive assumption.
		If instead $j>1$, then because $i \notin S_{j-1}$, we know by definition of $c_j$ that
		\begin{align}
			c_{j-1} \le \textrm{Cut}_i\left( a_i + (b_i-a_i) \cdot \frac{1}{|A|} \sum_{\ell = 1}^j m_\ell \right) = \textrm{Cut}_i(new(a_i))\,.
		\end{align}
		Either way, in the recursive call we have $x = c_{j-1} \le \textrm{Cut}_i(new(a_i))$.
		If $j=z$ then $c_j = y \ge b_i = a_i + (b_i-a_i) \cdot \frac{1}{|A|} \sum_{\ell = 1}^j m_\ell = new(b_i)$ by inductive assumption.
		And if instead $j < z$, then because $i \in S_j$, we know by definition of $c_j$ that
		\begin{align}
			c_j \ge \textrm{Cut}_i(a_i + (b_i-a_i) \cdot \frac{1}{|A|} \sum_{\ell = 1}^j m_\ell) = \textrm{Cut}_i(new(b_i))\,.
		\end{align}
		Therefore $x \le \textrm{Cut}_i(a_i)$ and $\textrm{Cut}_i(b_i) \le y$ at the start of every call to Algorithm 1.
		
		To argue that every agent receives value at least $1/n$, we proceed by induction on $k$.
		In particular, we claim that at the start of every call to Algorithm 1, every agent $i$ has $b_i - a_i \ge |A|/n$.
		In the initial call to Algorithm 1 this is true since $1-0 = n/n$.
		In the recursive call in step 6, consider an arbitrary $j \in [z]$ and an arbitrary agent $i$.
		Because $x \le \textrm{Cut}_i(a_i)$ and $\textrm{Cut}_i(b_i) \le y$ at the start of each call to Algorithm 1, we have that agent $i$ values $[x,y]$ as at least $b_i-a_i$.
		By definition of $new(a_i)$ and $new(b_i)$ in step 6, we have by inductive assumption
		\begin{align}
			new(b_i) - new(a_i) = (b_i - a_i) \cdot \frac{m_j}{|A|} \ge \frac{m_j}{n}
		\end{align}
		This completes the induction.
		When Algorithm 1 returns, it gives each agent $i$ the interval $[x,y]$. Since $x \le \textrm{Cut}_i(a_i)$ and $\textrm{Cut}_i(b_i) \le y$, this has value at least $b_i - a_i \ge \frac{1}{n}$ to agent $i$.
		
		To argue the bound on the number of queries, we proceed by induction on $k$.
		For $k=1$, the bound is $n^2$, which is satisfied since we issue $n-1$ queries for each of $n$ agents.
		In the inductive case, in the first round we issue $\lceil n^{1/k} \rceil - 1 \le n^{1/k}$ queries for every agent, for a total of at most $n^{1+1/k}$.
		By the inductive assumption, the remaining number of queries is
		\begin{align}
			\sum_{j=1}^{\lceil n^{1/k} \rceil} (k-1) m_j^{1+1/(k-1)}  \le (k-1) \Bigl(  \sum_{j=1}^{\lceil n^{1/k} \rceil} m_j  \Bigr)^{1+1/k} = (k-1)n^{1+1/k}
		\end{align}
		Combining, we get at most $kn^{1+1/k}$ queries in total.
	\end{proof}

	A key step in connecting cake cutting with sorting will be the following reduction, which reduces sorting a vector of $n$ elements with rank queries to proportional (contiguous) cake cutting with $n$ agents.  Rank queries have the form \emph{``How is $rank(x_j)$ compared to $k$?''}, where the answer can be ``$<$'', ``$=$'', or ``$>$''. 
	
	\medskip 
	
	\begin{prop} \label{sorting_to_prop_reduction} There exists a polynomial time reduction from sorting $n$ elements with rank queries to proportional cake cutting with $n$ agents. The reduction holds for any number of rounds.
	\end{prop}
	The reduction from sorting to cake cutting was essentially done in the work of Woeginger and Sgall \cite{woeginger2007complexity}, but appears implicitly. We formalize the connection to rank queries and note the reduction is round-preserving.
	The proof of Proposition~\ref{sorting_to_prop_reduction} is in section \ref{app:sorting_reduction}.
	
	
	
	\begin{prop}
		There is a deterministic sorting algorithm in the rank query model that runs in $k$ rounds and asks a total of $O(k n^{1+1/k})$ queries.
	\end{prop}
	\begin{proof}
		By the reduction from sorting to cake cutting in Proposition~\ref{sorting_to_prop_reduction}, the upper bound follows from Proposition~\ref{prop_ub}.
		
		The sketch of the resulting deterministic sorting algorithm is as follows.
		In the first round, for each $x$ in the array,
		query comparing $rank(x)$ to $\lceil n^{1-1/k} \rceil, \lceil 2n^{1-1/k} \rceil, \ldots, \lceil n - n^{1-1/k} \rceil$.
		This divides the array into $\lceil n^{1/k} \rceil$ blocks of indices of the form $(\lceil (i-1) n^{1-1/k} \rceil, \lceil i n^{1-1/k} \rceil)$ for $i = 1, 2, \ldots, \lceil n^{1/k} \rceil$.
		Each element either has its exact rank revealed, or is found to belong to a particular block.
		Then recursively call the sorting algorithm in each block.
	\end{proof}
	

	\subsection{Lower bound}
	
	In this section we first show a lower bound for sorting in the rank query model; for deterministic algorithms this bound improves upon the bound in ~\cite{AlonAzar88} by a constant factor   and the proof is simpler (see Appendix~\ref{app:sorting_lb}). Deterministic algorithms are relevant specifically for fair division, since some studies find that it is preferable to avoid randomness in the allocations if possible when dealing with human agents.

	\begin{proposition} \label{thm:sorting_lower_bound}
		Let $c(k,n)$ be the minimum total number of queries required to sort $n$ elements in the rank query model by the best deterministic algorithm in $k$ rounds. Then
		$c(k,n) \ge \frac{k}{2e} n^{1+1/k} -kn\,.$
	\end{proposition}

	Alon and Azar~\cite{AlonAzar88} show a lower bound of $\Omega\left(kn^{1+1/k}\right)$ for randomized sorting with rank queries, which together with the reduction in Proposition~\ref{sorting_to_prop_reduction} implies the next corollary.
	
	\begin{corollary}
		Let $\mathcal{A}$ be an algorithm that runs in $k$ rounds for solving proportional cake cutting with contiguous pieces for $n$ agents. If $\mathcal{A}$ succeeds with constant probability, then it issues $\Omega(k n^{1+1/k})$ queries in expectation.
	\end{corollary}
	
	The proof of Proposition~\ref{thm:sorting_lower_bound} is given in section \ref{app:sorting_lb}.

	\subsection{Sorting to cake cutting reduction} \label{app:sorting_reduction}
	
	Here we prove the reduction of sorting to proportional cake cutting where the sorting is not with comparisons, but rather with queries that, given an item $p$ and index $i$ return whether the rank of $p$ is less than, equal to, or greater than $i$.
	The bulk of the work has already been done by Woeginger and Sgall \cite{woeginger2007complexity} through the introduction of a set of cake valuations and an adversary protocol. 
	We present again their valuations and adversary protocol without proving the relevant lemmas; we would refer the reader to their paper for the proofs.
	Then we perform the last few steps to prove the reduction.
	

	
	\begin{definition}\cite{woeginger2007complexity}
		Let the $\alpha$-point of an agent $p$ be the infimum of all numbers $x$ such that $\mu_p([0,x]) = \alpha$.
		In other words, $\textrm{Cut}_p(\alpha) = x$.
	\end{definition}
	
	We fix $0 < \epsilon < 1/n^4$. The choice is not important.
	
	\begin{definition}\cite{woeginger2007complexity}
		For $i=1,\ldots,n$ let $X_i \subset [0,1]$ be the set consisting of the $n$ points $i/(n+1) + k\epsilon$ with integer $1 \le k \le n$.
		Further let $X = \bigcup_{1 \le i \le n} X_i$
	\end{definition}
	
	By definition every agent's $0$-point is at $0$. The positions of the $i/n$-points with $1 \le i \le n$ are fixed by the adversary during the execution of the protocol. In particular, the $i/n$-points of all agents are distinct elements of $X_i$. Note that this implies that all $i/n$-points are left of all $(i+1)/n$ points.
	
	\medskip
	
	\begin{definition}\cite{woeginger2007complexity}
		Let $\mathcal{I}_{p,i}$ be a tiny interval of length $\epsilon$ centered around the $i/n$-point of agent $p$.
	\end{definition}
	
	We place all the value of each agent $p$ in her $\mathcal{I}_{p,i}$ for $i=0,\ldots,n$. More precisely, for $i=0,\ldots,n$ she has a sharp peak of value $i/(n^2+n)$ immediately to the left of her $i/n$ point and a sharp peak of value $(n-i)/(n^2+n)$ immediately to the right of her $i/n$ point. Note that the measure between the $i/n$ and $(i+1)/n$ points is indeed $1/n$. Further note that the value $\mu_p(\mathcal{I}_{p,i}) = 1/(n+1)$. Also note that the $\mathcal{I}_{p,i}$ are all disjoint except for the $\mathcal{I}_{p,0}$, which are identical. Finally note that every $\alpha$-point of an agent $p$ lies inside one of that agent's $\mathcal{I}_{p,i}$s.
	
	\begin{definition}\cite{woeginger2007complexity}
		If $x \in \mathcal{I}_{p,i}$, then let $c_p(x)$ be the corresponding $i/n$-point of agent $p$.
	\end{definition}
	
	
	
	\begin{definition}\cite{woeginger2007complexity}
		We call a protocol \emph{primitive} iff in all of its cut operations $\textrm{Cut}_p(\alpha)$ the value of $\alpha$ is of the form $i/n$ with integer $0 \le i \le n$.
	\end{definition}
	
	\begin{lemma}\cite{woeginger2007complexity}
		For every protocol $P$ there exists a primitive protocol $P'$ such that for every cake cutting instance of the restricted form described above,
		\begin{enumerate}
			\item $P$ and $P'$ make the same number of cuts.
			\item if $P$ assigns to agent $p$ a piece $\mathcal{J}$ of measure $\mu_p(\mathcal{J}) \ge 1/n$, then also $P'$ assigns to agent $p$ a piece $\mathcal{J}'$ of measure $\mu_p(\mathcal{J}') \ge 1/n$.
		\end{enumerate}
	\end{lemma}
	
	It is also true that given $P$, protocol $P'$ can be quickly constructed. This follows directly from Woeginger and Sgall's constructive proof of the above lemma. This implies that we, the adversary, may assume w.l.o.g. that the protocol is primitive. We can now define the adversary's strategy. Fix a permutation $\pi$ on $[n]$. Suppose at some point the protocol asks $\textrm{Cut}_p(i/n)$. With multiple queries in the same round, answer the queries in an arbitrary order.
	\begin{enumerate}
		\item If $\pi(p) < i$, then the adversary assigns the $i/n$ point of agent $p$ to the smallest point in the set $X_i$ that has not been used before.
		\item If $\pi(p) > i$, then the adversary assigns the $i/n$ point of agent $p$ to the largest point in the set $X_i$ that has not been used before.
		\item If $\pi(p) = i$, then the adversary assigns the $i/n$ point of agent $p$ to the $i$th smallest point in the set $X_i$.
	\end{enumerate}
	
	This strategy immediately precipitates the following lemma.
	
	\begin{lemma}\cite{woeginger2007complexity}
		If $\pi(p) \le i \le \pi(q)$ and $p \ne q$, then the $i/n$ point of agent $p$ strictly precedes the $i/n$ point of agent $q$
	\end{lemma}
	
	At the end, the protocol must assign intervals to agents. Let $y_0, y_1, \ldots, y_n$ be the boundaries of these slices; i.e. $y_0 = 0$, $y_n = 1$, and all other $y_j$ are cuts performed. Then there is a permutation $\phi$ of $[n]$ such that for $i=1,\ldots,n$ the interval $[y_{i-1},y_i)$ goes to agent $\phi(i)$.
	
	\medskip
	
	\begin{figure}[h!]
		\centering
		\includegraphics[scale=1.6]{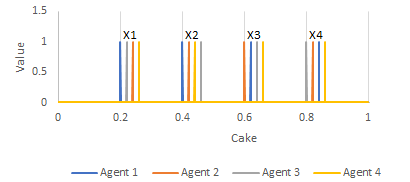}
		\caption{A potential value distribution for four agents.
			Each agent receives a spike in value in each of $X_0, X_1, X_2, X_3, X_4$ ($X_0$ is not shown).
			Each spike has total value $1/5$, so to get the required $1/4$ value an agent's slice must include parts of multiple $X_i$.
			Note that Agent 1 receives the first slot in $X_1$, Agent 2 receives the second slot in $X_2$, etc. Further note that slot 1 is allocated to Agent 1 in $X_2$ and slot 4 is allocated to Agent 4 in $X_3$.
			This implies that the slices must be allocated to agents $1, 2, 3, 4$ in order.
		}
		\label{fig:value_distribution}
	\end{figure}
	
	\begin{lemma}\label{lem:yi_in_xi}\cite{woeginger2007complexity}
		If the primitive protocol $P'$ is fair, then $y_i \in X_i$ for $1 \le i \le n-1$ and the interval $[y_{i-1},y_i]$ contains the $(i-1)/n$-point and the $i/n$-point of agent $\phi(i)$.
	\end{lemma}
	
	\begin{lemma}\label{lem:permuation_inversion}\cite{woeginger2007complexity}
		For any permutation $\sigma \ne id$ of $[n]$, there exists some $i$ with $$\sigma(i+1) \le i \le \sigma(i) \, .$$
	\end{lemma}
	
	We can now claim that $\phi = \pi^{-1}$.
	To prove this, suppose for sake of contradiction $\phi \ne \pi^{-1}$; then $\pi \circ \phi \ne id$ and by Lemma~\ref{lem:permuation_inversion}, there exists an $i$ such that
	\begin{align}
		\pi(\phi(i+1)) \le i \le \pi(\phi(i))
	\end{align}
	Then let $p = \phi(i+1)$ and $q = \phi(i)$. Further let $z_p$ be the $i/n$ point of agent $p$ and $z_q$ be the $i/n$ point of agent $q$. By Lemma~\ref{lem:permuation_inversion}, we have $z_p < z_q$. By Lemma~\ref{lem:yi_in_xi}, we have $z_p \in [y_i,y_{i+1}]$ and $z_q \in [y_{i-1}, y_i]$. But this implies $z_p \ge y_i \ge z_q$, in contradiction with $z_p < z_q$.
	Therefore $\phi = \pi^{-1}$.
	With this preliminary work out of the way, we are finally ready to state and prove the reduction.
	
	\bigskip
	
	\noindent \textbf{Proposition} \ref{sorting_to_prop_reduction} (restated): 
	\emph{There exists a polynomial time reduction from sorting an $n$ element with rank queries to proportional cake cutting with $n$ agents. The reduction holds for any number of rounds.}
	\begin{proof}
		After an evaluation query $\textrm{Eval}_p(x)$, where $x = \textrm{Cut}_{p'}(i/n)$ and $p \ne p'$, there are only two possible answers: $i/(n+1)$ and $(i+1)/(n+1)$. This reveals whether the $i/n$ point of $p$ is left or right of that of $p'$. This only reveals new information if $\pi(p') = i$. In this case, the information is whether $\pi(p) < i$ or $\pi(p) > i$. 
		
		After a cut query $\textrm{Cut}_p(i/n)$, there are only three answers. These correspond exactly to $\pi(p) < i$, $\pi(p) = i$, and $\pi(p) > i$. Thus w.l.o.g., all queries are cut queries. 
		
		Then given a sorting problem with rank queries, we can construct a proportional cake cutting instance such that any solution assigns slices according to the inverse permutation of the unsorted elements of the original sorting problem. The sorting problem can then be solved without any additional queries. Furthermore, each query in the cake cutting instance can be answered using at most one query in the sorting instance. This completes the reduction. Because of the one-to-one correspondence between queries, it immediately follows that the reduction holds for any number of rounds.
	\end{proof}

	\subsection{Sorting lower bound} \label{app:sorting_lb}
	
	Our approach for the lower bound builds on the work in~\cite{AZV86}. To show that this sorting problem is hard, we find a division between two regions of the array such that one must be solved in future rounds while the other still needs to be solved in the current round. 
	
	\bigskip 
	
	\noindent \textbf{Proposition} \ref{thm:sorting_lower_bound} (restated): 
	\emph{Let $c(k,n)$ be the minimum total number of queries required to sort $n$ elements in $k$ rounds in the rank query model. Then
		$c(k,n) \ge \frac{k}{2e} n^{1+1/k} -kn$.}
	\begin{proof}
		We proceed by induction on $k$.
		For $k=1$, note that if any two items $p_j, p_k$ have no query for indices $i, i+1$ then the adversary can assign those positions to those items and the solver will be unable to determine their true order. Thus for $i=2,4,\ldots n$ at least $n-1$ queries are necessary, for a total of $\lfloor n/2 \rfloor (n-1)$. Then
		\begin{align}
			\lfloor n/2 \rfloor (n-1) \ge (n/2-1/2)(n-1) = n^2/2 - n + 1/2 > n^2/(2e) - n\,. \notag
		\end{align}
		For $k>1$, 
		assume the claim holds for all pairs $(k',n')$ where either $(k' < k)$ or $(k' = k$ and $n' < n)$.
		If $n^{1/k} \le 2e$, then
		\begin{align}
			\frac{n^{1/k}}{2e} - 1 \le 0 \iff 
			\frac{k}{2e} n^{1+1/k} -kn \le 0 \notag 
		\end{align}
		so the bound is non-positive, and is thus trivially satisfied.
		Thus we may assume $n^{1/k} > 2e$.
		
		\smallskip
		
		If there are no queries in the first round, then we have
		\begin{align}
			c(k,n) \ge c(k-1,n) \ge \frac{(k-1)}{2e} n^{1+\frac{1}{k-1}} -(k-1)n
			= \frac{k}{2e} n^{1+1/k}  \Bigl[ \left(1-\frac{1}{k}\right)n^{\frac{1}{k^2-k}} + \frac{2e}{kn^{1/k}} \Bigr] - kn \notag 
		\end{align}
		From here it suffices to show
		$(1-\frac{1}{k})n^{1/(k^2-k)} + \frac{2e}{kn^{1/k}} \ge 1$.  \\
		
		\smallskip
		
		Recall the AM–GM inequality
		$ \alpha a + \beta b \ge a^\alpha b^\beta $ with $a,b,\alpha,\beta > 0$ and $\alpha+\beta=1$. Taking $\alpha = 1-{1}/{k}$, $\beta = {1}/{k}$, $a=n^{1/k^2-1/k}$, and $b={2e}/{n^{1/k}}$, we get
		\begin{align}
			\left(1-\frac{1}{k}\right)n^{1/(k^2-k)} + \frac{2e}{kn^{1/k}} \ge (2e)^{1/k} \ge 1
		\end{align}
		so we may assume there is at least one query in the first round.
		
		Take any $k$-round algorithm for sorting a set $V$ of $n$ elements using rank queries.
		Let $x$ be the maximum integer such that there exist $x$ items with no queries in $[1,x]$ but there do not exist $x+1$ items with no queries in $[1,x+1]$.
		Note that since there is at least one query, it follows that $x < n$.
		Let $S$ be one such set of $x$ items.
		Then at least $n-x$ items have a query in $[1,x+1]$. 
		At this point the adversary announces that every element of $S$ precedes every element of $V-S$.
		The adversary also announces the item at position $x+1$. We call this item $p_{mid}$.
		None of the $n-x$ queries help to sort the items in $S$ since they are either at $x+1$ or for an item not in $S$, so we also need $c(k-1,x)$ queries to sort $S$.
		Additionally, none of the $n-x$ queries help to sort the items in $V-S-\{p_{mid}\}$, so we also need an additional $c(k,n-x-1)$ queries to sort $V-S-\{p_{mid}\}$.
		This implies the following inequality.
		\begin{align} \label{recurrent_inequality}
			c(k,n) \ge c(k, n-x-1) + (n-x) + c(k-1, x)
		\end{align}
		We consider two cases.
		
		\begin{description}
			\item[Case $x \ge k/\ln{2}$.] By the inductive assumption,
			\begin{align}
				\begin{split}
					c(k,n) &\ge \frac{k}{2e} (n-x-1)^{1+1/k} -k(n-x-1) + (n-x) + \frac{k-1}{2e} x^{1+1/(k-1)} - (k-1)x\\
					&= \frac{k}{2e} n^{1+1/k}  \Bigl[  \left(1-\frac{x+1}{n}\right)^{1+1/k}
					+ \left(1-\frac{1}{k}\right) \frac{ x^{1+1/(k-1)} }{ n^{1+1/k} }
					+ \frac{2e}{kn^{1/k}}  \Bigr] 
					-kn+k    
				\end{split}
			\end{align}
			
			In the AM-GM inequality $\alpha a + \beta b \ge a^\alpha b^\beta$, taking 
			\begin{align} \alpha = 1-{1}/{k}, \beta = {1}/{k}, a = \frac{ x^{1+1/(k-1)} }{ n^{1+1/k} }, \mbox{ and } b=\frac{2e}{n^{1/k}},
			\end{align}
			we get
			\begin{align}
				c(k,n) \ge \frac{k}{2e} n^{1+1/k}  \Bigl[  \left(1-\frac{x+1}{n}\right)^{1+1/k}
				+ \frac{x}{n^{1-1/k^2}} \cdot \frac{(2e)^{1/k}}{n^{1/k^2}}  \Bigr] 
				-kn+k
			\end{align}
			
			Now, since $(1+\frac{1}{k})^k \nearrow e$, we have $e^{1/k} > 1+1/k$. This yields
			\begin{align}
				c(k,n) \ge \frac{k}{2e} n^{1+1/k}  \Bigl[  \left(1-\frac{x+1}{n}\right)^{1+1/k}
				+ 2^{1/k}\frac{x}{n}\left(1+\frac{1}{k}\right)  \Bigr] 
				-kn+k
			\end{align}
			Then recall Bernoulli's Inequality: $(1-a)^t \ge 1-at$ if $t \ge 1$ and $a \le 1$. This yields
			\begin{align}
				\begin{split}
					c(k,n) &\ge \frac{k}{2e} n^{1+1/k}  \Bigl[  1 - \frac{x+1}{n}\left(1+\frac{1}{k}\right)
					+ 2^{1/k}\frac{x}{n}\left(1+\frac{1}{k}\right)  \Bigr]  -kn+k\\
					&=  \Bigl[ \frac{k}{2e} n^{1+1/k} -kn \Bigr]  + \frac{k+1}{2e}n^{1/k}\Bigl(\left(2^{1/k}-1\right)x-1\Bigr) +k
				\end{split}
			\end{align}
			Then since by L'H\^opital's rule $k(2^{1/k}-1) \rightarrow \ln{2}$ from above, we have
			\begin{align}
				c(k,n) \ge  \Bigl[ \frac{k}{2e} n^{1+1/k} -kn \Bigr]  + \frac{k+1}{2e}n^{1/k}\left(\frac{x\ln{2}}{k}-1\right) +k\,.
			\end{align}
			Then by invoking our case assumption that $x \ge k/\ln{2}$, we get
			\begin{align}
				c(k,n) &\ge  \Bigl[ \frac{k}{2e} n^{1+1/k} -kn \Bigr]  +k \ge \frac{k}{2e} n^{1+1/k} -kn, \notag
			\end{align}
			as required.	
			\item[Case $x < k/\ln{2}$.] From inequality (\ref{recurrent_inequality}), we get 
			\begin{align}
				c(k,n) \ge c(k,n-x-1) + (n-x) + c(k-1,x)  \ge c(k,n-x-1) + n-\frac{k}{\ln{2}}\,. \notag 
			\end{align}
			By the inductive hypothesis,
			\begin{align}
				c(k,n) \ge \frac{k}{2e} (n-x-1)^{1+1/k} - nk + n- \frac{k}{\ln{2}}
				= \frac{k}{2e} n^{1+1/k} \left(1-\frac{x+1}{n}\right)^{1+1/k} -nk + n-\frac{k}{\ln{2}} \notag 
			\end{align}
			Using the Bernoulli inequality $(1-a)^t \ge 1-at$ with $t\ge 1$ and $a \le 1$, we get 
			\begin{align}
				c(k,n) &\ge \frac{k}{2e} n^{1+1/k} \Bigl[ 1-\frac{x+1}{n}\left(1+\frac{1}{k}\right) \Bigr] -nk + n-\frac{k}{\ln{2}} \notag \\
				&= \Bigl[ \frac{k}{2e} n^{1+1/k} -nk \Bigr] -\frac{k}{2e} n^{1/k} (x+1)\left(1+\frac{1}{k}\right) + n - \frac{k}{\ln{2}} \notag \\
				&\ge \Bigl[ \frac{k}{2e} n^{1+1/k} -nk \Bigr] -\frac{k+1}{2e} n^{1/k} \left(\frac{k}{\ln{2}}+1\right) + n - \frac{k}{\ln{2}}
			\end{align}
			At this point we want to show
			$n > \frac{k+1}{2e} n^{1/k} (1+k/\ln{2}) + k/\ln{2}$.
			It suffices to show both of the following inequalities
			\begin{align}
				(i) \; \;  \frac{3n}{4} > \frac{k+1}{2e} n^{1/k} \left(\frac{k}{\ln{2}}+1\right) \; \mbox{ and } \; (ii) \; \;  \frac{n}{4} > \frac{k}{\ln{2}}
			\end{align}
			Inequality $(i)$ holds if and only if 
			\begin{align}
				n^{1-1/k} > \frac{2(k+1)}{3e} \left(\frac{k}{\ln{2}}+1\right) \notag 
			\end{align}
			Since $n > (2e)^k$, we get that $n^{1-1/k} > (2e)^{k-1}$. For $k\ge2$, we obtain 
			$(2e)^{k-1} > \frac{2(k+1)}{3e} \left(\frac{k}{\ln{2}}+1\right)$, which concludes $(i)$. 
			
			\smallskip
			
			To show $(ii)$, recall that $n > (2e)^k$. Then for $k\ge2$ we get $(2e)^k > 4k/\ln(2)$, 
			which implies $(ii)$.
			This concludes the second case and the proof of the theorem.
		\end{description}
	\end{proof}

	\section{Folklore lemmas} \label{app:folklore_lemmas}
	
	Here  we include  a few folklore lemmas that we use, together with their proofs for completeness. 
	
	\begin{lemma} \label{eq:simple_lemma_sorted_array}
		Let  $\vec{y} = (y_1, \ldots, y_n)$ with $y_1 \geq \ldots \geq y_n  \geq 0$   and $\sum_{i=1}^n y_i = 1$. Then  
		$
		\sum_{j=1}^{i} y_j  \geq {i}/{n} 
		$ $\forall  i \in [n]$.
	\end{lemma}
	\begin{proof}
		Let $i \in [n]$. Since $\vec{y}$ is decreasing, we have 
		$\bigl( {\sum_{j=1}^i y_j} \bigr)/{i} \geq \bigl({\sum_{j=i+1}^n y_j} \bigr)/({n-i})$ $(\dag)$.
		
		\smallskip 
		Assume by contradiction that 
		$y_1 + \ldots + y_i < \frac{i}{n}$ $(\ddag)$. 
		Adding $y_{i+1} + \ldots + y_n$ to both sides of $(\ddag)$, we get 
		\begin{align}
			1 = y_1 + \ldots + y_n & < \frac{i}{n} + y_{i+1} + \ldots + y_n \notag \\
			& \leq \frac{i}{n} + \frac{n-i}{i} \cdot  \left( {\sum_{j=1}^i y_j} \right) \explain{By $(\dag)$} \\
			& < \frac{i}{n} + \frac{n-i}{i} \cdot \left( \frac{i}{n} \right) \explain{By $(\ddag)$} \\
			& = 1\,.
		\end{align}
		We obtained  $ 1 < 1$, thus the assumption  in $(\ddag)$ must have been false and  the lemma holds.
	\end{proof}

	\begin{lemma} \label{lem:helper_k_to_power_1_plus_1_over_k} 
		Let $ x \in \mathbb{R}_{\geq 3}$. Then 
		$
		x^{1+\frac{1}{x}} > x + 1\,.
		$
	\end{lemma}
	\begin{proof}
		Raising both sides to the power $1/(x+1)$, the  inequality is  equivalent to  $x^{\frac{1}{x}}  > \left(x +1\right)^{\frac{1}{x+1}}$, or 
		$(1/x) {\ln{(x)}}  > (1/(x+1)) {\ln{(x+1)}}$ $(\dag)$.
		
		Define $g(x)={(\ln{x})}/{x}$. Its derivative is $g'(x)={(1-\ln{x})}/{x^2}$. Thus  $g$ is increasing on $[1,e)$ and decreasing on $[e,+\infty)$. It follows that $(\dag)$ holds for $x\geq 3$ and the lemma follows.
	\end{proof}
	
	The next lemma shows that if $v$ is an integrable function defined on $[0,1]$, then there is an interval $I$ of length $p$ on the circle where the interval $[0,1]$ is bent such that the point $0$ coincides with $1$, with the property that $\int_{I} v(x) \, dx = p$. 
	\begin{lemma}  \label{lem:length_p_probability_mass_p_ordered}
		Let $v : [0,1] \to \mathbb{R}_{\geq 0}$ be an integrable function with $\int_{0}^{1} v(x) \, dx = 1$. Then there exists  $a  \in [0,1]$ such that  one of the following  holds:
		\begin{itemize}
			\item $\int_{a}^{a+p} v(x) \, dx = p$, where $0 \leq a \leq 1-p$; 
			\item 
			$\int_{0}^{a}v(x) \, dx + \int_{a+1-p}^{1} v(x) \, dx = p $, where $ 1-p < a < 1$.
		\end{itemize} 
	\end{lemma}
	\begin{proof}
		We define a new function $g: [0,1] \to \mathbb{R}_{\geq 0}$, such that 
		\[ 
		g(x) = 
		\begin{cases}
			\int_{x}^{x+p} v(y) \, dy & \text{ if } 0 \leq x \leq 1-p \,. \\
			& \\
			\int_{x}^{1} v(y) \, dy  +   \int_{0}^{x+p-1} v(y) \, dy   & \text{ if } 1-p < x \leq 1\,.
		\end{cases}
		\]
		To prove the lemma it suffices to show that there exists  $c\in[0,1]$ such that $g(c)=p$.
		Indeed, the function $g$ is continuous and so integrable. Let $F : [0,1] \to \mathbb{R}_{\geq 0}$ be 
		$ F(x) = \int_{0}^{x} v(y) \, dy\,.$
		Using this notation, we get:
		\begin{align} 
			\int_{0}^{1} g(x) \, dx & = \left[ \int_{0}^{1-p} \int_{x}^{x+p} v(y) \, dy \, dx \right] + \left[ \int_{1-p}^{1} \left( \int_{x}^{1} v(y) \, dy \right) + \left( \int_{0}^{x+p-1} v(y) \, dy \right) \, dx \right] \notag \\
			& = \int_{0}^{1-p} \Bigl( F(x+p) - F(x) \Bigr) \, dx + \int_{1-p}^{1} \Bigl(\left(  F(1) - F(x)\right) + \left( F(x+p-1) - F(0) \right) \Bigr) \, dx \notag \\
			& = \int_{0}^{1-p} F(x+p) \, dx - \int_{0}^{1-p} F(x) \, dx+ \int_{1-p}^{1} 1 \, dx - \int_{1-p}^{1} F(x) \, dx + \int_{1-p}^{1} F(x+p-1) \, dx\explain{Since $F(1) = 1$ and $F(0) = 0$.}\\
			& = \int_{0}^{1-p} F(x+p) \, dx - \int_{0}^{1-p} F(x) \, dx + p - \int_{1-p}^{1} F(x) \, dx + \int_{1-p}^{1} F(x+p-1) \, dx\,. \label{eq:decomposing_double_integral}
		\end{align}
		We have 
		\begin{align}
			\int_{1-p}^{1} F(x+p-1) \, dx = \int_{0}^p F(y) \, dy 
			\qquad \mbox{and} \qquad \int_{0}^{1-p} F(x+p) \, dx = \int_{p}^{1} F(z) \, dz  \,.  \label{eq:dydz_replace_dx_ordered}
		\end{align}
		Using \eqref{eq:dydz_replace_dx_ordered}  in 
		\eqref{eq:decomposing_double_integral} yields
		\begin{align} 
			\int_{0}^{1} g(x) \, dx & = \int_{p}^{1} F(z) \, dz - \int_{0}^{1-p} F(x)\, dx + p - \int_{1-p}^{1} F(x) \, dx + \int_{0}^{p} F(y) \, dy \,. 
			\label{eq:integral_rewritten_simplified}
		\end{align}
		Notice that $\int_{p}^{1} F(z) \, dz + \int_{0}^{p} F(y) \, dy =\int_{0}^{1} F(x)\, dx$ and $- \int_{0}^{1-p} F(x)\, dx - \int_{1-p}^{1} F(x) \, dx = - \int_{0}^{1} F(x)\, dx$. 
		
		Therefore, the four integrals in \eqref{eq:integral_rewritten_simplified} cancel each other and we get
		$\int_{0}^{1} g(x) \, dx  = p\,.$ $(\dag)$
		
		Applying the first mean value theorem for definite integrals in  $(\dag)$, there exists $c\in [0,1]$ with the property that $g(c) = \frac{1}{1-0}\int_{0}^{1} g(x) \, dx = p$, which concludes the proof.
	\end{proof}

\end{document}